\documentclass[aps,twocolumn]{revtex4-1}
\usepackage[utf8]{inputenc}
\usepackage{amsmath}
\usepackage{stmaryrd}
\usepackage{color}
\usepackage{mathrsfs}
\usepackage{yfonts}
\usepackage{rotating}
\usepackage{graphicx}
\usepackage{amsfonts}
\usepackage{mathrsfs}
\usepackage{amsmath}
\usepackage[american]{babel}
\usepackage{wasysym}
\usepackage{slashed}
\usepackage{pgfplots}
\usepackage{braket}
\usepackage{graphicx}
\usepackage[caption=false]{subfig}
\usepackage{dsfont}
\usepackage{float}
\usepackage{amsthm,bbm}

\usepackage{natbib} 
\PassOptionsToPackage{breaklinks}{hyperref}
\usepackage[colorlinks,allcolors=black]{hyperref}

\usepackage{mathtools}

\theoremstyle{definition}
\newtheorem{definition}{Definition}
\newtheorem{example}{Example}

\newtheorem{observation}{Observation}
\newtheorem{proposition}{Proposition}

\newtheorem{corollary}{Corollary}

\newtheorem{question}{Question}
\newtheorem{fact}{Fact}

\usepackage{cleveref}
\crefname{proposition}{Proposition}{Propositions}
\crefname{lem}{Lemma}{Lemmas}
\crefname{figure}{Fig.}{Figs.}
\crefname{corollary}{Corollary}{Corollary}
\crefname{conjecture}{Conjecture}{Conjectures}
\crefname{section}{Section}{Sections}
\crefname{appendix}{Appendix}{Appendixes}

\newcommand{\ald}[1]{\textcolor{red}{\textsf{(#1)}}}

\usepackage{cleveref}
\crefname{proposition}{Proposition}{Propositions}
\crefname{definition}{Definition}{Definitions}
\crefname{lemma}{Lemma}{Lemmas}
\crefname{figure}{Figure}{Figures}
\crefname{corollary}{Corollary}{Corollary}
\crefname{conjecture}{Conjecture}{Conjectures}
\crefname{section}{Section}{Sections}
\crefname{appendix}{Appendix}{Appendixes}
\crefname{observation}{Observation}{Observation}
\crefname{remark}{Remark}{Remark}
\crefname{example}{Example}{Examples}
\crefname{equation}{Eq.}{Eqs.}
\crefname{table}{Table}{Tables}
\crefname{figure}{Fig.}{Figs.}

\usepackage{tikz}
\usepackage{blochsphere}
\usetikzlibrary{shadows,fit,decorations.pathmorphing}
\usetikzlibrary{trees,calc,positioning}
\usetikzlibrary{decorations.pathreplacing,shapes.misc}
\usetikzlibrary{shapes.geometric,calc}
\newcounter{hexi}

\usepackage{amssymb} 

\makeatletter
\renewcommand\@make@capt@title[2]{%
	\@ifx@empty\float@link{\@firstofone}{\expandafter\href\expandafter{\float@link}}%
	{\textbf{#1}}\@caption@fignum@sep#2\quad}%
\makeatother

\makeatletter 
\renewcommand{\fnum@figure}{\textbf{Figure~\thefigure}}

\usepackage{cancel}

\begin{document}
	\setcitestyle{numbers,acm}
	\title[Symmetry of a quantum state]%
	{Entanglement in highly symmetric multipartite quantum states}
	
	\author{Adam Burchardt$^1$, Jakub Czartowski$^1$, and Karol \. Zyczkowski$^{1,2}$}
	\affiliation{$^1$Institute of Physics, Jagiellonian University, \L{}ojasiewicza 11, 30-348 Krak\'ow, Poland}
	\affiliation{$^2$Center for Theoretical Physics, Polish Academy of Sciences, Warsaw, Poland}	
	
	\date{June 1, 2021} 

	\begin{abstract}
		We present a construction of genuinely entangled multipartite quantum states
		based on the group theory. 
		Analyzed states resemble the Dicke states, whereas the interactions occur
		only between specific subsystems related by the action of the selected group. 
		The states constructed by this technique exhibit desired symmetry properties
		and form a natural resource for less-symmetric quantum informational tasks. 
		We propose quantum circuits efficiently generating such states, which in general have
		smaller complexity than the circuits necessary to create fully symmetric entangled states. 

	\end{abstract}
	
	\maketitle
	
	\section{Introduction}
	Characterization of different classes of entanglement in multipartite quantum systems 
	remains a major issue relevant for various quantum information tasks
	and is interesting from the point of view of foundations of the quantum theory. 
	Since manipulation and relative control over several qubits have already become a standard
	task \cite{PhysRevLett.106.130506,Riofr_o_2017}, analysis of such systems in the context of resources for various information protocols is needed.
	Among different types of entangled states, permutation invariant states have attracted a lot of attention  for both continuous 
	\cite{PhysRevA.71.032349,PhysRevLett.99.150501} 
	and discrete \cite{PhysRevLett.98.060501,Bergmann_2013} variable systems. 
	
	A notable example of such states is due to  Dicke \cite{Dicke54}.
	The {\sl Dicke states} of $N$-qubit system with $k$ excitations are defined as 
	\cite{PhysRevA.67.022112},
	\begin{equation}
		\label{Dicke}
		\ket{\text{D}_N^k} = \frac{1}{\sqrt{N!}}\sum_{\sigma \in \mathcal{S}_N}  \sigma \Big( \ket{ \underbrace{1\;\cdots \;1}_{k \text{ times}} \; \underbrace{0 \;\cdots \; 0}_{N-k \text{ times}}} \Big) ,
	\end{equation}
	where the summation runs over the symmetric group $\mathcal{S}_N$. 
	Permutational symmetry of the Dicke states simplifies 
	their theoretical \cite{Liu_2019} and experimental \cite{L_cke_2014} detection
	and facilitates the tasks of quantum tomography \cite{Mazza_2015}. 
	Multipartite Dicke states  
	   are experimentally accessible \cite{Wieczorek_2009}. 
	Since the entanglement of Dicke states turned out to be maximally persistent and robust for the particle loss \cite{Dicke54,PhysRevLett.86.910}, such states provide inherent resources in numerous quantum information contexts, 
	including decoherence free quantum communication \cite{PhysRevLett.92.107901}, quantum secret sharing \cite{PhysRevLett.98.020503}, open destination teleportation \cite{PhysRevA.59.156}, and quantum metrology \cite{RevModPhys.90.035005}. 
	Entanglement properties and application of mixtures of Dicke states were also studied \cite{Tura2018separabilityof,Aloy_2021,marconi2020entangled}.
	
	Up till now vast majority of the scientific interest was focused on fully symmetric tasks, 
	  like parallel teleportation \cite{helwig2013absolutely} and
	symmetric quantum secret sharing protocols \cite{HelwigAME}. 
	In various realistic situations  it seems natural that 
	such a full symmetry between collaborating systems is not possible, required or even desirable. 
	As an example, it was shown  \cite{HIGUCHI2000213} that a four-qubit state 
	maximally entangled with respect to all possible symmetric partitions does not exist. 
	Such a state would allow one for the
	parallel teleportation of two qubits between any two subsystems to the remaining pair of systems. 
	Nevertheless, the following state:
	\[
	\dfrac{1}{2} \Big( \ket{1100}+\ket{0110}+\ket{0011}+\ket{1001} \Big)
	\]
	allows for the teleportation of a single qubit  to an arbitrary subsystem, 
	and additionally for the parallel teleportation across the partition $13|24$. 
	It is especially reasonable to share resources in a not fully symmetric way in variants of quantum secrets sharing schemes, allowing only some parties for cooperation. 
	Alternatively, such schemes were already considered, modelled by quantum channels \cite{Kimble_2008} or pairs of EPR state shared between nodes of a network \cite{Ac_n_2007,DistributionEntanglement,Perseguers_2010,PhysRevA.77.022308}. 
	Implementation of such communication networks is expected to be developed in the near future \cite{Gisin_2007,Kimble_2008}, with a variety of possible applications \cite{ComplexNetworks2}. 
	
	  In this work 
	we present an approach for constructing highly-entangled quantum 
	states with a given group of symmetry. Their form is similar to \cref{Dicke}, whereas the summation runs over elements of a specific subgroup $H$  of the symmetric
	group $ \mathcal{S}_N$. 
	We refer to such states as {\sl Dicke-like states} and analyse,
	which types of symmetries are feasible. 
	In particular, we present explicit constructions of quantum states
	based on highly symmetric geometrical objects, such as regular polygons,  
	Platonic solids, and regular plane tilings. 
	We compare them as entanglement resources with the original Dicke states,
	which were investigated as  ground states of specific $2$-body Hamiltonians with
	a well defined number of excitations \cite{PhysRevE.67.066203,DickeModelRevisited}. 
	We present the generalized,  Dicke-like states as a ground states of a 
	wider class of  Hamiltonians with $3$-body interactions. 
	
	Dicke-like states constitute a canonical basis of all pure quantum states with a given symmetry. 
	Apart from the theoretical aspect of this statement, such a basis might be efficiently applied in the context of quantum chemistry. 
	Various molecules in Nature (like benzene) stand out with remarkable symmetries. 
	In principle, the spin-model of such states might be efficiently simulated by a proper superposition of Dicke-like states. 
	In recent years, correlations and the entanglement contained in chemical bounds was investigated \cite{Szalay_2017,Szalay_2015,ding2020concept}, and a special attention was dedicated to highly symmetrical molecules \cite{Szalay_2015,ding2020concept}. 
	Although for most molecules the total correlation between orbitals seems to be classical \cite{ding2020concept}, the general significance of entanglement in chemical bonds is not yet clear. 
	Therefore,  general investigations of entanglement in highly symmetric systems may shed some light on the nature of correlation in relevant chemical molecules. 
	
	Furthermore, we introduce a larger class of genuinely entangled states based on an arbitrary network structure, graphically represented as a (hyper)graph. 
	Once more, this construction generalizes the Dicke states \cref{Dicke},
	however, 
	 excitations appear only in particular subsystems represented by (hyper)edges of the graph. 
	We do not impose any restrictions on the network structure, however, the detailed analysis of entanglement properties is performed under the assumption on regularity. 
	In order to construct excitation-states, we propose quantum circuits whose complexity is comparable with the complexity of quantum circuits proposed for Dicke states \cite{DickeCircuits,brtschi2019deterministic}. 
	Joint unitary operations in the circuit are performed on the neighbouring nodes. 
	
	We demonstrate entanglement properties of Dicke-like states and excitation-states in terms of entanglement resources \cite{RevModPhys.91.025001,Eltschka_2014} contained in particular subsystems.  
	In several cases an interesting phenomenon is observed: most of the entanglement is concentrated between nodes of distance two and is absent between immediate neighbours. 
	This indicates a particular advantage of such states in the case, in which  
	collaboration between neighbouring nodes is not desired. 
	In principle, constructed states are genuinely but not maximally entangled, persistent  
	with respect to measurements performed locally on each subsystem
	and with respect to losses of certain subsystems  \cite{PhysRevA.100.062329}.
	
	Mixed-state entanglement is a valuable resource for quantum communication protocols \cite{horodecki2001mixedstate,Bennett_1996}.
	Therefore the Dicke-like states and excitation-states might be distilled on particular subsystems \cite{PMID:11323664} and further used in concrete protocols. 
	Notice that once the excitation-state is established, it does not require any particular actions on nodes (as entanglement swapping \cite{coecke2004logic}), even for the collaboration of parties being far away in the network structure. 
	This is in contrast with the case anlyzed earlier, in which 
	the nodes of the graph share pairs of maximally entangled states. 
	
	This work is organized as follows. 
	In \cref{sec3a}, we introduce the notion of group symmetry of a multipartite state. 
	We demonstrate that not all kinds of symmetries are feasible. 
	Moreover, we define the Dicke-like states and discuss their separability criteria. 
	In \cref{hyper}, we generalize Dicke-like states in a less symmetric manner. 
	Basing on a given (hyper)graph, we introduce the excitation-states. 
	Such a graphical representation is later used in the construction
	 of quantum circuits in \cref{Circuit}, and Hamiltonians in \cref{Hamiltonians}, both relevant to 
	the Dicke-like states and the excitation-states. 
	\cref{cccc} analyzes the entanglement properties of introduced families of states. 
	In \cref{sec3} particular examples of Dicke-like states and excitation-states are presented. 
	We observe a phase transition of bipartite entanglement in excitation-states with respect to the average degree of a relevant (hyper)graph
	and present some more details in \cref{Phase transitions}. 
	In order to demonstrate the viability of the provided constructions, we have simulated one of the considered states on available quantum computers: IBM – Santiago and Athens.
	In  \cref{5-qubit} these results are presented and discussed.
	
	
	\section{Group of symmetry of a quantum state}
	\label{sec3a}
	A pure state $\ket{\psi}$ on an $N$-fold product Hilbert space is called \textit{symmetric} if it is invariant under permutation of $N$ subsystems, i.e. $\sigma \ket{\psi} = \ket{\psi}$ for any element $\sigma$ from the permutation group $ \mathcal{S}_N$. 
	Each such state $\ket{\psi}$ can be written in the computational basis:
	\begin{equation}
		\label{eq1}
		\ket{\psi} = \sum_{\sigma \in \mathcal{S}_N} \dfrac{1}{\sqrt{| \mathcal{S}_N |}}
		\ket{\phi_{\sigma (1)}}\cdots \ket{\phi_{\sigma (N)}},
	\end{equation}
	\noindent
	where the sum is taken over all permutations $\sigma$
	from the symmetric group $\mathcal{S}_N$. 
	Any symmetric state of the $N$-qubit system can be graphically represented 
	 as a collection of $N$ points on the Bloch sphere 
	 corresponding to vectors  $\ket{\phi_i},\ldots , \ket{\phi_N}$. 
	Such a visualization scheme is called the \textit{stellar} or \textit{Majorana} representation,
	 see \cref{fig1}.
	 This representation,
	 originaly used for pure states of simple systems of a finite dimension,
	 was later generalized for mixed states \cite{Serrano_Ens_stiga_2020},
	 symmetric states of systems consisting of several qubits 
	 \cite{Aulbach_2010,Markham_2011,PhysRevA.85.032314},
	 and symmetric states of higher dimensional subsystems \cite{CHGBSZ21}.
	 
	In this section, we discuss a natural generalization of this approach -- the
	 \textit{generalized stellar representation}, suitable for
	 quantum states exhibiting modes symmetries, i.e. for which summation in \cref{eq1} 
	 runs over a subgroup $H$ of the symmetric group $\mathcal{S}_N$. 
	We begin with a very natural definition of the group of symmetry of a quantum state.
	
	\begin{definition}
		A state $\ket{\psi}$ of $N$ subsystems is called $H$-symmetric, where $H$ is a subgroup of the permutation group, $H< \mathcal{S}_N$, iff it is permutation invariant for any $\sigma \in H$, and only for such permutations.
	\end{definition}
	
	\begin{figure}[h!]
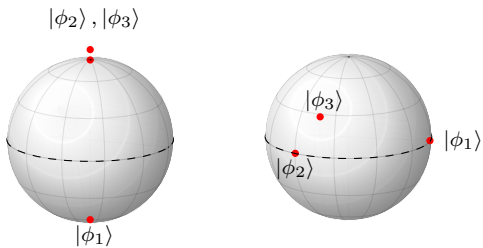

		\subfloat{
			\begin{blochsphere}[radius=1.1 cm,tilt=15,rotation=-20]
				\drawBallGrid[style={opacity=0.1}]{30}{30}
				
				\drawGreatCircle[style={dashed}]{0}{0}{0}
				
				\drawAxis[style={draw=black}]{0}{0}
				
				\drawAxis[color=black]{90}{0}
				
				\labelLatLon{up}{90}{0};
				\labelLatLon{down}{-90}{90};
				\node[circle,fill,red] at (up){.};
				\node[circle,fill,red] at (0,1.2){.};
				\node[above] at (0,1.4){{ \footnotesize $ \ket{\phi_2},\ket{\phi_3}$ }};
				\node[circle,fill,red] at (down){.};
				\node[below] at (down) {{ \footnotesize $\ket{\phi_1}$}};
			\end{blochsphere}  
		}
		\hspace{5mm}
		\subfloat{
			\begin{blochsphere}[radius=1.1 cm,tilt=15,rotation=-10]
				\drawBallGrid[style={opacity=0.1}]{30}{30}
				
				\drawGreatCircle[style={dashed}]{0}{0}{0}
				
				\drawAxis[style={draw=black}]{0}{0}
				
				\drawAxis[color=black]{90}{0}
				
				\labelLatLon{one}{0}{0};
				\labelLatLon{two}{0}{120};
				\labelLatLon{three}{0}{240};
				\node[circle,fill,red] at (one){.};
				\node[circle,fill,red] at (two){.};
				\node[circle,fill,red] at (three){.};
				\node[right] at (one) {{ \footnotesize $\ket{\phi_1}$ }};
				\node[below] at (two) {{\footnotesize  $\ket{\phi_2}$}};
				\node[above] at (three) {{ \footnotesize $\ket{\phi_3}$}};
				
				\node[above,white] at (up) {{ $ \ket{\phi_2},\ket{\phi_3}$ }};
				\node[below,white] at (down) {{ $\ket{\phi_1}$}};
			\end{blochsphere} 
		}
		\caption{Stellar representation of three qubits 
		states. a) the state  
		 $\ket{\text{W}} \propto  \ket{001} + \ket{010} +  \ket{100} $
		is represented by two stars at the North pole and a single star at the South.
		b) the  state $\ket{\text{GHZ}} \propto \ket{000}+ \ket{111}$ 
		   corresponds  to three stars evenly distributed on the equatorial plane:
		    $\ket{0} + e^{2ik/3 }\ket{1}$ for $k=0,1,2$. 
		 The stellar representation  is not unique as
		  any rotation of the Bloch sphere  around the vertical axis 
			yields the same state.}
		\label{fig1}
	\end{figure}
	
	As we shall see below, constructing such symmetric states can be involved. 
	Consider first  a system consisting of three qubits. 
	There exist  two particular classes of states which should be distinguished \cite{ThreeQub}:
	\begin{align*}
		\ket{\text{GHZ}} &\propto \ket{000} +\ket{111} , \\
		\ket{\text{W}} &\propto \ket{001} +\ket{010} +  \ket{100} .
	\end{align*}
	Both states are symmetric with respect to any permutation of particles, see \cref{fig1}. 
	We discuss briefly the problem of constructing three--qubit quantum states
	with other types of symmetry. 
	It is easy to find an example of $\mathcal{S}_2$-type of symmetry in three-qubit setting, namely $ \ket{001} +\ket{010} $. 
	The mentioned state is obviously separable, but a similar example $|\chi\rangle$
	can be found also among genuinely entangled states,
	\[
	|\chi\rangle \propto \ket{001}+\ket{010}+2\ket{100}+ 2\ket{111}.
	\]
	There is one more non-trivial subgroup of the symmetric group $\mathcal{S}_3$, 
	namely the \textit{alternating} group $\mathcal{A}_3$, the group of all even permutations. 
	However, 
	in contrast to the group $\mathcal{S}_3$,
	there is no three-qubit state with such  a symmetry,
	which is not symmetric with respect to the symmetric group.   
	This phenomenon might be explained in the following way: 
	Restricting the group of summation in \cref{eq1} to $\mathcal{A}_3$ does not imply that some
	further symmetries can be present in the resulted state. 
	In particular, summation over all $\mathcal{A}_3$ permutations in \cref{eq1} of the vector
	$\ket{100}$, leads to the state $\ket{W}$, which is fully symmetric. 
	A similar behavior might be observed for stars leading to the state $\ket{GHZ}$. 
	Therefore in both cases the group of symmetry as the largest group stabilizing the state, is 
	the full symmetric group $\mathcal{S}_3$.
	
	
	We discuss the necessary and sufficient conditions 
	for existence of a quantum state with a given symmetry. 
	Let us begin with a simple observation 
	that for any group $H < \mathcal{S}_N$ there exists an $H$-symmetric state $\ket{\psi} \in \mathcal{H}^{\otimes N}_d$ if the number of energy levels $d$ in each 
	subsystem is large enough.
	
	\begin{proposition}
		Consider the subgroup $H < \mathcal{S}_N$. The following state $\ket{\psi} \in \mathcal{H}^{\otimes N}_N$ of the local dimension $N$:
		\begin{equation*}
			\ket{\psi} \propto \sum_{\sigma \in H} \ket{ \sigma_0 (0) \cdots \sigma_{N-1} (N-1)} 
		\end{equation*}
		is $H$ symmetric.
	\end{proposition}
	
	\begin{proof}
		It is easy to see, that the group $H$ stabilizes $\ket{\psi}$. Suppose there is a larger group $H'$ containing $H$, i.e. $H<H'$, and stabilizing $\ket{\psi}$. Take any element $\sigma \in H'\setminus H$. Observe, that there is no term
		$
		\ket{ \sigma_0 (0) \cdots \sigma_{N-1} (N-1)}
		$
		in $\ket{\psi} $, hence $H' $ does not stabilize $\ket{\psi} $.
	\end{proof}
	
	\noindent
	For instance, the three qutrits state
	\begin{equation*}
		\ket{\psi} \propto \ket{012} +\ket{201} +\ket{120}
	\end{equation*}
	is $\mathcal{A}_3$-symmetric. 
	Indeed, one can see, that the cyclic permutation of the last three qutrits does not change the entire state.
	Summarizing, there is no $\mathcal{A}_3$-symmetric three-qubit state, but we found an
	desired three-qutrit state. 
	This observation encourages us to pose the following question.
	
	\begin{question}
		Consider any group of symmetry $H< \mathcal{S}_N$. What is the minimal local dimension $d$ for which there exists a $H$-symmetric state $\ket{\psi} \in \mathcal{H}^{\otimes N}_d$?
	\end{question}
	
	In order to answer the question above, consider a suitable basis for symmetric states. Any symmetric state $\ket{\psi}$ is a superposition of Dicke 
	states \cite{Dicke54,PhysRevA.67.022112}:
	\[
	\ket{\psi} \propto \alpha_0 \ket{D_N^0} +\cdots +  \alpha_N \ket{D_N^N}.
	\]
	Therefore, the Dicke states, $\ket{D_N^k}$ with $k=0,1,\dots, N$
	can be treated as a basis of the symmetric state space. 
	In other words, the coefficient $\alpha_{i_1, \ldots,i_k}$ in front of any term $e_{i_1, \ldots,i_k} $ in the symmetric state is the same, and the symmetric state is uniquely determined by the coefficients $\alpha_0, \alpha_1 ,\alpha_{12},\ldots,\alpha_{1\cdots N}$ in front of terms $\ket{1\cdots 1 0\cdots0}$. 
	A similar decomposition for any $H$-symmetric states is possible, 
	however, its building blocks have to be suitably selected. 
	
	
	\begin{definition}
		\label{Dickelike}
		For a given subgroup $H <\mathcal{S}_N$, we define the $N$-qubit \textit{Dicke-like $H$ state} with $k$ excitations is the following way:
		\begin{equation}
			\label{Dicke-like}
			\ket{\text{D}_N^k}_H: =  \dfrac{1}{\sqrt{|H|}}\sum_{\sigma \in H}  \sigma \Big( \ket{ \underbrace{1\;\cdots \;1}_{k \text{ times}} \; \underbrace{0 \;\cdots \; 0}_{N-k \text{ times}}} \Big)
		\end{equation}
		where the summation runs over all permutations belonging to the group $H$. 
	\end{definition}
	
	Notice that the group symmetry of the Dicke-like states is not necessarily $H$. 
	In general, usage of such states to construct $H$-symmetric states is not straightforward. 
	Example presented below illustrates problems which might occur. 
	
	Suppose, we are looking for the general form of the cyclic, $\mathcal{C}_4$-symmetric state $\ket{\phi}$ (invariant under cyclic permutations). 
	It is of the following form:
	\begin{align*}
		\ket{\phi} =&\alpha_0 \ket{D_4^0}+ \alpha_1 \ket{D_4^1}+ \alpha_{2,1} \ket{D_4^2}_{\mathcal{C}_4}+\alpha_{2,2} \ket{D_4^2}_{\mathcal{S}_2 \times \mathcal{S}_2}+ \\
		& \alpha_3 \ket{D_4^3}+ \alpha_4 \ket{D_4^4},
	\end{align*}
	where 
	\begin{eqnarray}
	\label{D42a}
	\!\!\!\!\!	\ket{D_4^2}_{\mathcal{C}_4}&  = &\frac{1}{2}\Big( \ket{1100}+\ket{0110}+\ket{0011}+\ket{1001}\Big), \  \\
	\label{D42b}
		\ket{D_4^2}_{\mathcal{S}_2 \times \mathcal{S}_2}&  = & \frac{1}{\sqrt{2}}\Big(\ket{1010}+\ket{0101}\Big).
	\end{eqnarray}
	In each of them separate coefficient might be adjusted. 
	Observe that up to the normalization constant 
	\[
	\ket{D_4^2} \;\propto\;\ket{D_4^2}_{\mathcal{C}_4} +\ket{D_4^2}_{\mathcal{S}_2 \times \mathcal{S}_2} ,
	\]
	which means that the fully symmetric term $\ket{D_4^2}$ splits into two classes: $\ket{D_4^2}_{\mathcal{C}_4} $ and $\ket{D_4^2}_{\mathcal{S}_2 \times \mathcal{S}_2}$. 
	This division was obtained by taking appropriate terms $\ket{1\cdots 1 0\cdots0}$ and acting
	on them by the cyclic group $\mathcal{C}_4$, unless not all such terms belong to the prior class. 
	
	Decomposing the state space according to the symmetric subspaces has
	 one additional advantage. 
	Above decomposition implies that 
	cyclic symmetric state of four qubits does not exist. 
	Observe that the symmetry group of the state $\ket{D_4^2}_{\mathcal{C}_4}$ 
	 is the dihedral group $\mathcal{D}_8$ (in fact for $\ket{D_4^2}_{\mathcal{S}_2 \times \mathcal{S}_2} $ is the same). 
	Therefore, by taking $ \alpha_{2,1} \neq \alpha_{2,2}$, we obtain the $\mathcal{D}_8$ symmetric state, in any other case, the $\mathcal{S}_N$ symmetric case. Therefore, 
	$\mathcal{C}_4$ symmetric state $\ket{\psi} \in \mathcal{H}^{\otimes 4}_2$ does
	not exist. 
	In fact, the $\ket{D_4^2}_{\mathcal{C}_4} $ and $\ket{D_4^2}_{\mathcal{S}_2 \times \mathcal{S}_2}$ are non-trivial examples of the dihedral $\mathcal{D}_8$-symmetric states, nevertheless both of them are separable with respect to the partition ($13|24$). 
	
	On the other hand, 
	the $\mathcal{D}_8$-symmetric state can be constructed among genuinely 
	entangled states. 
	Indeed, a superposition of $\ket{D_4^2}_1$ and $\ket{1111}$, 
	\[
	\ket{\phi} \propto \ket{1100}+\ket{0110}+\ket{0011}+\ket{1001} + 2\ket{1111}
	\]
	leads to a  $\mathcal{D}_8$-symmetric state.
	
	Such an analysis might be performed for other subgroups $H<\mathcal{S}_N$ for small number of parties $N$. 
	We list below all subgroups of $\mathcal{S}_N$ for $N=3,4$, which might be realized as a symmetric group of a quantum state of $N$ qubits.
	
	\begin{observation}
		There is no $\mathcal{A}_3$-symmetric three qubit state. 
		All other subgroups of $\mathcal{S}_3$ are allowed, relevant states were priori constructed. 
		In a four-partite setting, the cyclic group $\mathcal{C}_4$, and the alternating subgroups $\mathcal{A}_3$ and $\mathcal{A}_4$ cannot be realized as a group of symmetry of four-partite qubit state. 
	\end{observation}

	The general analysis is tightly connected with the \textit{partially ordered set} (poset) of all subgroups of $\mathcal{S}_N$, which has rather complicated structure \cite{Poset}. 
	We conjecture that alternating groups $\mathcal{A}_N$ are rather difficult to realize as a group of symmetries of quantum states, impossible in the $N$-qubit setting.
	
	Symmetric states of $N$-qubit system
	have an effective representation, 
	called \textit{Stellar representation} \cite{M_kel__2010}, as $N$ points (stares) on the Bloch sphere, related to the roots of the Majorana Polynomial. 
	It was recently generalized for mixed states setting \cite{Serrano_Ens_stiga_2020}. 
	Stellar representation turned out to be useful for 
	classification of entanglement in symmetric quantum states \cite{M_kel__2010,CrossRatioNqubits,PhysRevA.85.032314}. 
	Moreover, by imposing the special symmetry conditions on the stars, related states exhibits highly entanglement properties \cite{Martin10,Martin15,Martin17}. 
	Stellar representation might be also generalized to the $H$-symmetric states. 
	As it was before, consider $N$ points on the Bloch sphere: 
	$
	\ket{\phi_1},\ldots , \ket{\phi_N},
	$
	and the following product
	\begin{equation}
		\label{eq2}
		\ket{\psi} =
		\dfrac{1}{\sqrt{|H|}}
		\sum_{\sigma \in H} 
		\ket{\phi_{\sigma (1)}}\cdots \ket{\phi_{\sigma (N)}},
	\end{equation}
	\noindent
	where the sum is taken over all permutations $\sigma$ from the group $H$. 
	We shall represent state $\ket{\psi}$ as a collection of $N$ points 
	on the Bloch sphere relevant to vectors $\ket{\phi_i},\ldots, \ket{\phi_N}$
	with indicated the group $H$ which is acting in \cref{eq2}, see \cref{StarsN3}. 
	
	Note that a given constellation of `stars' at the Bloch sphere together with the
	selected symmetry group $H$ do not represent uniquely the quantum state. 
	The important information is carried in how the group $H$ is contained in $\mathcal{S}_N$, which mathematically might be expressed by immersion $H\hookrightarrow \mathcal{S}_N$ of the group $H$ into the symmetric group $\mathcal{S}_N$. 
	
	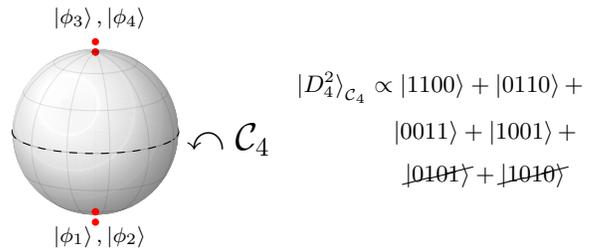
\begin{figure}[h!]
		\subfloat{
			\begin{blochsphere}[radius=1.1 cm,tilt=15,rotation=-20]
				\drawBallGrid[style={opacity=0.1}]{30}{30}
				
				\drawGreatCircle[style={dashed}]{0}{0}{0}
				
				\drawAxis[style={draw=black}]{0}{0}
				
				\drawAxis[color=black]{90}{0}
				
				\labelLatLon{up}{90}{0};
				\labelLatLon{down}{-90}{90};
				\node[circle,fill,red] at (up){.};
				\node[circle,fill,red] at (down){.};
				\node[circle,fill,red] at (up){.};
				\node[circle,fill,red] at (0,1.2){.};
				\node[above] at (0,1.3){{ \footnotesize $ \ket{\phi_3},\ket{\phi_4}$ }};
				\node[circle,fill,red] at (down){.};
				\node[circle,fill,red] at (0,-1.2){.};
				\node[above] at (0,-1.6){{ \footnotesize $ \ket{\phi_1},\ket{\phi_2}$ }};
				\labelLatLon{one}{0}{0};
				\node[right] at (one) {{ \Large $\curvearrowleft \mathcal{C}_4$ }};
			\end{blochsphere}  
		}
		\begin{tikzpicture}
			\begin{scope}  
				\tikzstyle{every node}=[] 
				\node[] at (2,5) {$\ket{D_4^2}_{\mathcal{C}_4} \propto\ket{1100}+ \ket{0110}+$};
				\node[] at (2.6,4.4) {$\ket{0011}+ \ket{1001}+$};
				\node[] at (2.6,3.8) {$\cancel{\ket{0101}}+ \cancel{\ket{1010}}$};
				\node[text=white] at (0,3) {.};
			\end{scope}  
		\end{tikzpicture}
		\caption{Stellar representation of 
			a Dicke state $\ket{D_N^k}$ consists 
			of $k$ stars at the South Pole and $N-k$ stars at the North Pole \cite{PhysRevA.85.032314}. A similar picture holds true for the Dicke-like states: 
			 $\ket{D_4^2}_{\mathcal{C}_4}$  defined in Eq. (\ref{D42a}) and
			not equivalent to the Dicke state $\ket{D_4^2}$,
			is represented by four stars on the Bloch sphere, with indicated action of the Cyclic group 
			${\mathcal{C}_4}$. 
			Instead of choosing  this group
			one could consider the action of
			the Dihedral group $\mathcal{D}_8$, see \cref{Platonic states}. 
			The resulting states are the same, $\ket{D_4^2}_{\mathcal{C}_4} = \ket{D_4^2}_{\mathcal{D}_8} $, and
			they are
			based on the graph construction, see \cref{figY}.
		}
		\label{StarsN3}
	\end{figure}

	
	\section{States from (hyper)graphs}
	\label{hyper}

	In the previous section we discussed the problem of determining quantum states invariant under specified permutations $\sigma \in H$, where $H < \mathcal{S}_N$. 
	Moreover, we presented the Dicke-like states, 
	which are building blocks for all such modestly symmetric states. 
	They were obtained by limiting the group of summation in the 
	standard construction of the classical Dicke states. 
	As specification of all symmetries of a given quantum state is a challenging task,
	we propose here an alternative, a more geometric approach to the problem. 
	Certain quantum states can be associated with (hyper)graphs, whereas their symmetries as automorphisms of (hyper)graphs. 
	
	Let us emphasize here  that graphs \cite{Graph1,Graph2} and hypergraphs \cite{Hypergraphs_2013}  were already successfully used in order to construct genuinely entangled states, and the notion of \textit{graph states} is well established. 
	Graph states (known also as cluster states \cite{PhysRevLett.86.910}) exhibit genuine entanglement properties \cite{helwig2013absolutely}. They are useful in the context of the one-way quantum computer or quantum error-correcting codes \cite{LULCsupport,Raussendorf_2007}. 
	Another approach allows one to generate 
	a family of random quantum states related to a given graph \cite{CNZ10,CNZ13},
	the vertices of which describe the interaction between subsystems.
	
	In this work we propose another scheme to associate 
	to a given graph with $N$ vertices
	a single pure quantum state of an $N$-party system.
	Such a representation reflects not only the symmetry,
	 but also the structure
	of a quantum circuit under which presented families of quantum states 
	can be constructed. 
	
	Recall that a \emph{graph} $G$ is a pair $(V,E)$ where $V$ is a finite set, and $E$ is a collection of two-element subsets of $V$. 
	We refer to elements of $V$ as \textit{vertices}, and elements of $E$ as \textit{edges} respectively. 
	A successful generalization of the notion of a graph is \textit{hypergraph} 
	\cite{Hypergraphs_2013},
	for which edges are arbitrary (not necessarily 2-element) subsets of $V$. 
	A hypergraph is called \textit{uniform} if its all edges consist of the same number of elements equal to $k$, we refer to such an object as $k$-hypergraph.
	For example, a 2-hypergraph is simply a graph.
	We denote the number of vertices by $N$. 
	For simplicity, we assume that vertices of the (hyper)graph are labeled by numbers $1,\ldots,N$.
	
	\begin{definition}
		With a given (hyper)graph $G=(V,E)$, we associate a quantum state of $N$ qubits in the following way: 
		\begin{equation}
			\ket{G} := \dfrac{1}{\sqrt{|E|}}  \sum_{e \in E} \ket{\psi_e}
		\end{equation} 
		where $\ket{\psi_e}$ is a tensor product of $\ket{1}$ on positions labelled by indices form the (hyper)edge $e$ 
		and $\ket{0}$ on other positions. We shall refer to such states as \textit{excitation--states}. 
		\cref{ex1p} illustrates this definition. 
	\end{definition}
	
	\noindent
	In the definition above, we could consider \textit{weighted hypergraphs}, i.e. in which each edge $e\in E$ has a \textit{weight} given by a complex number. 
	Obviously, the normalization of a state should be changed. 
	Interestingly, in such a way, we can  obtain a general form of a quantum state. 
	Indeed, a given quantum state $\ket{\psi} \in \mathcal{H}^N_d$ might be presented in a computational basis. 
	From this form, we can construct a weighted hypergraph, where edges are relevant to the non-vanishing terms and weights are given by coefficients. 
	Nevertheless, to keep the presentation simple,
	in this work we do not consider weighted graphs. 
	
	\begin{figure}[h!]
		\definecolor{pur}{RGB}{186,146,162}
		\definecolor{greeo}{RGB}{91,173,69}
		\begin{tikzpicture}
			\begin{scope}[scale=.3]
				\tikzstyle{every node}=[circle,draw,minimum size = 0.4cm]
				\tikzstyle{every node}=[circle,inner sep=4pt,draw=gray!95!white,fill=greeo!10!white,thick,anchor=base];
				\node(1) at (0,5) {};
				\node(3) at (5,0) {};
				\node(2) at (5,5){};
				\node(4) at (0,0) {};
				
				\tikzstyle{every node}=[]
				\node[] at (0,5) {\footnotesize 2};
				\node[] at (5,5) { \footnotesize 4};
				\node[] at (5,0) { \footnotesize 3};
				\node[]  at (0,0) {\footnotesize 1};
				
				\draw [gray, thick] plot [smooth cycle, tension=0.2] coordinates { (0,6.1)   (6,0.1)    (5,-1.1) (-1,4.9) };
				\draw [blue,thick] plot [smooth cycle, tension=0.3] coordinates { (-1.2,-1.2) (6.2,-1.2) (6.2,6.2)  (-1.2,6.2) };
				\draw [purple, thick] plot [smooth cycle, tension=0.7] coordinates { (-1,-1)   (-1,1)    (1,1) (1,-1) };
				
				\tikzstyle{every node}=[]
				\node(1a) at (2.5,-3) { $\ket{\Psi_\text{tel}}=$ \textcolor{blue}{$\ket{1111}$} + \textcolor{purple}{$\ket{1000}$} +\textcolor{gray}{$\ket{0110}$} };
			\end{scope} 
		\end{tikzpicture}
		\caption{Telescope state $\ket{\Psi_\text{tel}}$ \cite{ProblemTangle} represented as an excitation-state. The (hyper)graph is not uniform since (hyper)edges are of different capacities. The relevant state is simply obtained by reading over all edges. It has only one non-trivial symmetry, which exchange parties $2$ and $3$. }
		\label{ex1p}
	\end{figure}
	
	Observe that the $\ket{\text{W}}$ state, the generalized $\ket{\text{W}_N}$ states, 
	and the Dicke states $\ket{\text{D}_N^k}$ 
	are associated with the complete $k$-hypergraphs under provided construction, see \cref{ex1}. 
	All these states are fully symmetric. 
	This reflects the fact that \textit{automorphisms group} of a complete $k$-hypergraph is 
	the symmetric group $\mathcal{S}_N$. 
	This relation might be generalized for all excitation-states. 
	
	\begin{corollary}
		Consider an excitation-state $\ket{G}$ related to the hypergraph $G$. The group of symmetries of a state $\ket{G}$ is an automorphisms group of the related hypergraph $G$.
	\end{corollary}
	
	\noindent
	Even though determining the automorphisms group for a graph is NP-hard problem, it is widely studied for several classes of graphs. 
	Therefore, this graphical representation of quantum states 
	may be conclusive for determining its group of symmetries. 
	We shall use this relation directly on various examples of Dicke-like states. 
	
	The Dicke-like states introduced in Section \ref{sec3a}
	 belong to the class of excitation-states
	and form building blocks for all states of intermediate symmetry. 
	Notice that the (hyper)graph construction of excitation-states is far more general, in particular, 
	it does not require any symmetries.
	
	\begin{corollary}
		Each Dicke-like state $\ket{D_N^k}_H$ is equivalent to the excitation-state of a 2-transitive $k$-hypergraph.
	Furthermore, such a hypergraph state is a Dicke-like state with the group of symmetries 
	determined by the group automorphisms of the graph.
	\end{corollary}
	
	
	\begin{proof}
		Recall that the state $\ket{D_N^k}_H$ was obtained via the action of a group $H$ on the vector $\ket{1\cdots 1 0\cdots 0}$ with $k$ entries equal to unity. 
		The group $H<\mathcal{S}_N$, so it acts on the $N$ element set $\{ 1,\ldots ,N\}$. 
		Consider now the following (hyper)graph. 
		For any element $h \in H$, join $h(1),\ldots, h(k)$ with an edge. 
		The relevant excitation-state coincides with $\ket{D_N^k}_H$. 
		The group of symmetry for a Dicke-like state corresponds to the group of automorphisms of a relevant (hyper)graph. 
		By the above construction, related graph is \textit{edge-transitive}, i.e. for any two edges $e_1=\{v_1 ,\ldots ,v_k \}$ and 
		$e_2=\{w_1 ,\ldots ,w_k \}$, 
		there exists an element $h \in H$ such that $h(v_i)=w_i$ for any $i=1,\ldots,k$. 
		Interestingly, the reverse statement is also true. 
		Indeed, take the group of automorphisms of a graph, and suppose it is edge-transitive. 
		Without loss of generality suppose that vertices $1,\ldots,k$ are connected in an edge $e_0$. 
		Observe that any other edge $e$, is a result of some automorphism $h$, i.e. $e =h(e_0)$ (all terms in an excitation-state might be obtained from $\ket{1\cdots 1 0\cdots 0}$ by an action of automorphisms). 
		Moreover, none of the edges is distinguished, hence each of them is achieved the same number of times, as the coefficients in front of the edges are equal. 
	\end{proof}
	
	Let us briefly compare the concept of constructing genuinely entangled quantum states 
	presented above with the 	notion of graph states. 
	Firstly, notice that the excitation state is evidently different than the 
	graph state associated with the same graph. 
	The latter is strictly connected to the quantum circuits, in which edges of a graph are replaced by $CNOT$ gates. 
	Therefore, the graph state is closely related to its deterministic realization by the aforementioned two-qubit gates. 
	At the same time, encoding the closed formula of a graph state (written in computational basis), and its entanglement properties is rather a challenging task. 
	For an excitation-state the situation is different. 
	The exact form of a state and its entanglement properties are straightforward, 
	contrary to its physical realization under quantum circuits. 
	In fact, the quantum circuit for excitation-states refers to the graph construction as well. Nevertheless, it is significantly more complex than for the graph states, see \cref{Circuit} for further discussion. 
	
	
	\section{Entanglement properties}
	\label{cccc}
	
	An excitation-state $\ket{G}$ forms in general a genuinely entangled state. 
	As we shall see, its entanglement properties reflect the structure of the (hyper)graph. 
	We analyse reductions of the $N$-partite pure state $\rho^G=|G\rangle \langle G|$
	onto a bipartite system $S$ 
	and describe its entanglement
	in terms of the concurrence $C$ \cite{PhysRevLett.78.5022}. 
	This quantity is a faithful entanglement measure 
	easy to compute for any two-qubit mixed state 
	and for any bipartition of a multipartite pure state. 
	
	For any two-qubit mixed state $\rho$  its concurrence reads
	\cite{PhysRevLett.78.5022},
	\begin{equation}
		\label{concurrence}
		C (\rho ):= \text{max} \{\lambda_1 -\lambda_2 -\lambda_3 -\lambda_4 ,0\},
	\end{equation}
	in which $ \lambda _{1},\ldots,\lambda _{4}$ denote square 
	roots of the eigenvalues of a Hermitian matrix ${\sqrt {\rho }}{\tilde {\rho }}{\sqrt {\rho }}$
	ordered  decreasingly,
	where
	\[
	{\displaystyle {\tilde {\rho }}=(\sigma _{y}\otimes \sigma _{y})\rho ^{*}(\sigma _{y}\otimes \sigma _{y})}
	\]
	is, the so-called, spin-flipped form of $\rho $. 
	Concurrence forms an entanglement monotone \cite{EntMonotones}, 
	and it is faithful, as it admits strictly positive values for entangled states and 
	it vanishes for separable states.

	The generalized concurrence $C_{v |\text{rest}}$ 
	measures the entnaglement between the subsystem $v$ and the rest of the system. 
	For pure states, this quantity is determined by the relevant reduced density matrix \cite{PhysRevLett.78.5022},
	$C_{v |\text{rest}} =2 \; \sqrt{\text{det} \rho_v} $,
	so its values belong to the range $[0,1]$.
	
	The distribution of bipartite quantum entanglement, measured by the concurrence amongst 
	$N$ qubits satisfies a monogamy inequality \cite{DistributedEntanglement,Osborne_2006}:
	\begin{equation}
		\label{eq5}
		C_{v_1  |v_2 \ldots v_N }^2 \geq C_{v_1 v_2 }^2 + \ldots C_{v_1   v_N}^2 ,
	\end{equation}
	where $v_2 \ldots v_N$ are vertices relevant to subsystems. 
	In that sense, the squared concurrence, also called {\sl tangle}
	\cite{PhysRevLett.78.5022},
	 correctly quantifies bipartite entanglement in multi-partite systems, inasmuch as the entanglement between subsystem $v_1$ and the rest of the system does not exceed the sum of entanglement between $v_1$ and any other qubit. 

	For any subsystem $S$ of the studied system we define its complementary
	subsystem $\bar S$.  For any vertex $v \in S$
	we introduce the set $S_v:=S\setminus v$.
	Firstly, we show that for any subsystem $S$ and any distinguished particle $v \in S$ 
	which is far enough from the rest of the subsystem $S$, the reduced state $\rho^G_S:={\rm Tr}_{\bar S} \rho^G$ separates across 
	the partition $v \big\vert  S_v$.
	Moreover, we show that this is necessary, but not sufficient condition for separability in any subsystem. 
	Secondly, we provide exact formulae relating the two-party concurrence $C_{vw}$ 
	between any two particles, $v$ and $w$,
	with the generalized concurrence $C_{v |\text{rest}}$ 
	between particle $v$ and the rest of the system. 
	We are going to use this result for quantifying entanglement of various examples in \cref{sec3}. 
	
	As we shall see, inequality (\ref{eq5}) is generally not saturated, 
	which implies presence of multipartite entanglement.
	Finally, we discuss separability of an excitation-state $\ket{G}$
	and  provide necessary and sufficient conditions for separability
	in terms of edge structure of the relevant graph. 
	Our results show that an excitation-state $\ket{G}$ 
	is typically  
	strongly entangled.

	We restrict the analysis to the $k$-uniform hypergraphs, and assume our graphs to be connected.
	Not connected graphs are relevant for the tensor product of two excitation-states, and hence might be analyzed separately. 
	By distance $d$ between vertices $v_0$ and $w$, we understand the minimal number of vertices $v_1 , \ldots v_d$, such that there exist (hyper)edges $e_1, \ldots, e_d$:
	\[
	v_{i-1}, v_i \in e_i
	\]
	and $v_d =w$. 
	
	Consider a subsystem $S$ with one distinguished particle $v$ and suppose that the distance between corresponding nodes in the graph is $d (v, w ) >2$ 
	for any $w \in S_v=S \setminus v$. 
	The reduced density matrix $\rho^G_S$ is separable across the partition 
	$v \big\vert S_v$. 
	Indeed, one can show that
	\begin{align*}
		\rho^G_S \propto  \; d_v
		&  \ket{1}_v \bra{1} \otimes \ket{0\ldots 0}_{v^c} \bra{0\ldots 0} + \\
		& \ket{0}_v \bra{0}   \otimes \Big(  \underbrace{\rho^G_{S_v} -d_v  \ket{0\ldots 0}_{v^c} \bra{0\ldots 0} }_{\text{all terms are non-negative}}\Big) ,
	\end{align*}
	where $d_v$ denotes the degree of the relevant vertex. 
	Obviously, presented form of reduced density matrix is completely separable.
	
	\begin{corollary}
		The reduced density matrix $\rho_G^S$ of a subsystem $S = v \cup S'$ is  completely separable with respect to the partition
		$v \big\vert S'$ for any particle $v$ such that $d (v, S' ) >2$. 
	\end{corollary}
	
	
	
	The \textit{degree} $d_v$ of the vertex $v$ is the number of edges on which $v$ is incident. 
	The \textit{joint neighborhood} $n_{vw}$ of two vertices is the number of sets $W$ such that both: $W \cup v$ and $W \cup w$ constitute an edge. 
	The \textit{section} $s_{vw}$ is the number of edges on which $v$ and $w$ are incident. 
	Note that for graphs, the joint neighborhood is simply the number of vertices adjacent to both: $v$ and $w$, while section $s_{vw}=1$ or $s_{vw}=0$ depending whather vertices $v$ and $w$ are connected. 
	
	Consider now reduction of a pure state, $\rho^G=|G\rangle \langle G|$
	to the  subsystem $S= \{ v, w\}$ consisting of two parties. 
	Careful analysis of a (hyper)graph structure yields the following form of reduced density matrix $\rho^G_{vw}={\rm Tr}_{\bar S} \rho^G$:
	\begin{equation*}
		\rho^G_{vw}  =\dfrac{1}{|E|}
		{\displaystyle
			\begin{pmatrix}
				\lambda & 0 & 0 & 0 \\
				0 & d_{v}-s_{vw} & n_{vw}& 0 \\
				0 & n_{vw} & d_{w}-s_{vw}& 0 \\
				0 & 0 & 0 & s_{vw}\\
		\end{pmatrix}} ,
	\end{equation*}
	where the first element $\lambda =|E| - d_v -d_w   +s_{vw}$. 
	By the PPT test the reduced state $\rho^G_{vw}$ is entangled 
	iff $n_{vw}^2 > \lambda s_{vw}$,
	while  the exact amount of entanglement 
	can be characterized by the concurrence.

	To simplify the notation,
	for a given excitation state $\ket{G}$ and  two selected subsystems  $v$ and $w$,
	the concurrence of the reduced state $\rho_{vw}^G$, will be denoted as 
	\[
	C_{vw} := C (\rho_{vw}^G ) .
	\]
	We compute the two-party concurrence $C_{vw}$ according to \cref{concurrence}: 
	\begin{equation}
		\label{condition}
		C_{vw} =\text{max} \Big\{ 0, \dfrac{2}{|E|} \big( n_{vw} - \sqrt{s_{vw} \lambda } \big) \Big\}
	\end{equation}
	where $\lambda =|E| - d_v -d_w   +s_{vw}$. 
	
	Since concurrence is a faithful entanglement measure for 2-parties systems  \cite{EntMonotones,GeneralizedConcurrence}, the positivity of the right-hand side of 
	the above equation 
	is necessary and sufficient for entanglement between subsystems $v$ and $w$. 
	Therefore we shall first analyze the case for which \cref{condition} 
	reduces to zero. 
	Notice that the concurrence trivially vanishes if $n_{vw} =0$, which means that $v$ and $w$ have no common neighborhood. 
	This condition is fulfilled precisely in two situations, either $d(v,w)=1$ or $d(v,w)>2$. 
	The above observation 
	matches the previous discussion 
	that the system $\rho^G_{vS'}$ always separates if  $d (v, S' ) >2$. 
	Consider now the case in which  \cref{condition} takes strictly positive values. 
	Indeed, for $v$ and $w$ of distance $d(v,w) =2$,
	the neighborhood $n_{vw} \neq 0$ and the section $s_{vw}=0$.
	
	\begin{corollary}
		The reduced 2-parties state $\rho^G_{vw}$ is entangled iff $d(v,w) =2$ or $d(v,w) =1$ and 
		\begin{equation*}
			n_{vw}^2 > s_{vw} \big(|E|-d_v-d_w+s_{vw} \big).
		\end{equation*}
	\end{corollary}
	
	Roughly speaking, we observe the entanglement between 
	parties $v$ and $w$ if they share common neighbors. 
	On the other hand, entanglement  can disappear if both nodes
	become the closest neighbors.   
	The inequality above is usually satisfied in almost complete (hyper)graphs. 
	For graphs with a local-like structure, it will be violated. 
	In such a case the bipartite entanglement will be present only for distant two parties. 
	Some statements in this spirit are presented below. 
	

	Recall that a hypergraph might be identified with a pair $(V,E)$. 
	We define the \textit{product} of two disjoint hypergraphs $(V_1,E_1) \sqcup (V_2,E_2)$ as a hypergraph $(V_1 \cup V_2,E_1 \sqcup E_2  )$ with vertices being the union $V_1 \cup V_2$ of vertices sets and with edges of the following form:
	\[
	E_1 \sqcup E_2 : =\sum_{e_1 \in E_1 ,e_2 \in E_2} e_1 \cup e_2 .
	\] 
	In short, we say that such a hypergraph is a \textit{product hypergraph} with respect to the division $V_1 | V_2 $. 
	We derive the following criterion for separability of the excitation-state $\ket{G}$ corresponding to a $k$-hypergraph in terms of separation of  its edges. 
	
	\begin{proposition}
		The excitation-state $\ket{G}$ corresponding to a $k$-hypergraph is separable, $\ket{G} =\ket{G}_{V_1} \otimes \ket{G}_{V_2}$, iff it is a product hypergraph with respect to the division $V_1 | V_2 $.
	\end{proposition}
	
	\begin{proof}
		On one hand, it is a straightforward observation that for a product hypergraph $(V_1 \cup V_2,E_1 \sqcup E_2  )$ the relevant excitation-state $\ket{G}$ is separable across the division $V_1 | V_2 $, i.e. $\ket{G} =\ket{G}_{V_1} \otimes \ket{G}_{V_2}$. 
		
		On the other hand, the following two conditions are necessary for the excitation-state $\ket{G}$ to be separable:
		\begin{enumerate}
			\item the number of excitations in $V_1$ and $ V_2 $ is the same, i.e. there are numbers $c_1,c_2$ such that $c_1+c_2 =k$ and for each edge $e\in E$, $e \cap V_i =c_i$;
			\item the set of edges factors across the partition $V_1 |V_2 $, i.e. there are sets $E_1\in P(V_1),E_2\in P(V_2)$, such that $E =\sum_{e_1 \in E_1 ,e_2 \in E_2} e_1 \cup e_2$.
		\end{enumerate}
		One may observe that those two conditions are equivalent to the fact that $G$ is the product hypergraph. 
		%
		%
	\end{proof}
	
	\noindent
	Some exemplary separable excitation-states are
	presented in \cref{figY,figYb}. 
	Verification of the above criterion is 
	straightforward. 
	Observe that for excitation-states $\ket{G}$ the
	separability criterion is fulfilled by the complete bipartite graph only. 
	Indeed, the first condition simplifies due to the fact that there might be edges between $V_1 \cup V_2 =V$ only. 
	The second condition, comes to completeness of the bipartite graph under the assumption of connectivity of the graph. 
	Following the usual notation in graph theory, let $K_{V_1 V_2}$ be a \textit{complete bipartite graph}, while $G= K_{V_1 \ldots V_k}$ \textit{complete multi-partite} $k$-hypergraph. 
	Recall that a multi-partite ($k$-partite) hypergraph $G$ is a hypergraph whose vertices can be partitioned into $k$ different independent sets. 
	Equivalently, it is a hypergraph that can be colored with $k$ colors.
	
	\begin{corollary}
		(Only for graphs) 
		\label{cor5}
		The excitation-state $\ket{G}$ is separable $\ket{G} =\ket{G}_{V_1} \otimes \ket{G}_{V_2}$ iff the relevant graph $G$ is complete bipartite graph, $G= K_{V_1 V_2}$. 
		In fact, such a state forms a tensor product of two $|W\rangle$-like states:
		\[
		\ket{G} =\ket{W}_{V_1} \otimes \ket{W}_{V_2} .
		\]
	\end{corollary}
	
	\noindent
	We reformulate these separability criteria in terms of symmetries of the Dicke-like states. 
	
	\begin{corollary}
		A Dicke-like state  $\ket{D_N^2}_H$ with two excitations is separable across 
		the bipartition $N_1|N_2 $ iff
		its group of symmetry is equal to $\mathcal{S}_{N_1} \times \mathcal{S}_{N_2}$.
	\end{corollary}
	
	\noindent
	We can derive a similar conclusion to \cref{cor5} in the case of $k$-hypergraphs and multi-separability.

	\begin{corollary}
		\label{cor6}
		The $k$-regular excitation-state $\ket{G}$ separates $\ket{G} =\ket{G}_{V_1} \otimes \cdots \otimes \ket{G}_{V_k}$ iff the relevant hypergraph is complete multi-partite hypergraph $G= K_{V_1 \ldots V_k}$.   
		Such a state is a tensor product of $k$ states from the class $|W\rangle$,
		\[
		\ket{G} =\ket{W}_{V_1} \otimes \cdots \otimes \ket{W}_{V_k} .
		\]
	\end{corollary}

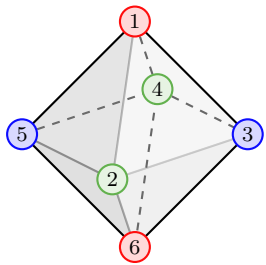
	\begin{figure}[h!]
		\definecolor{pur}{RGB}{186,146,162}
		\definecolor{greeo}{RGB}{91,173,69}
		\begin{tikzpicture}
			\begin{scope}[scale=.3]
				\tikzstyle{every node}=[circle,draw,minimum size = 0.4cm]
				\tikzstyle{every node}=[circle,inner sep=4pt,draw=greeo!95!white,fill=greeo!15!white,thick,anchor=base];
				\node(1) at (0,5) {};
				\node(3) at (5,0) {};
				\tikzstyle{every node}=[circle,inner sep=4pt,draw=pur!95!white,fill=pur!15!white,thick,anchor=base];
				\node(2) at (5,5){};
				\node(4) at (0,0) {};
				
				\tikzstyle{every node}=[]
				\node[] at (0,5) {\footnotesize 1};
				\node[] at (5,5) { \footnotesize 2};
				\node[] at (5,0) { \footnotesize 3};
				\node[]  at (0,0) {\footnotesize 4};
				\node[] at (6.9,2.5)  {\Large $\cong$};
				
				\tikzstyle{every node}=[]
				
				\node(1a) at (21,5) {\footnotesize $\ket{C_4}= \ket{1100} +\ket{0110} +$};
				\node(2a) at (22.2,3.7) { \footnotesize $\ket{0011} +\ket{1001} $};
				
				\draw[ultra thick,gray] (1)--(2);
				\draw[ultra thick,gray] (2)--(3);
				\draw[ultra thick,gray] (3)--(4);
				\draw[ultra thick,gray] (1)--(4);
			\end{scope}

			\begin{scope}[shift={(-0.2,0)},scale=.3]
				\tikzstyle{every node}=[circle,draw,minimum size = 0.4cm]
				\tikzstyle{every node}=[circle,inner sep=4pt,draw=greeo!95!white,fill=greeo!15!white,thick,anchor=base];
				\node(1) at (10.5,5) {};
				\node(3) at (14.5,5) {};
				\tikzstyle{every node}=[circle,inner sep=4pt,draw=pur!95!white,fill=pur!15!white,thick,anchor=base];
				\node(2) at (9,0){};
				\node(4) at (16,0) {};
				
				\tikzstyle{every node}=[]
				\node[] at (10.5,5)  {\footnotesize 1};
				\node[] at (14.5,5) { \footnotesize 3};
				\node[] at (9,0) { \footnotesize 2};
				\node[]  at (16,0) {\footnotesize 4};
				
				\draw[ultra thick,gray] (1)--(2);
				\draw[ultra thick,gray] (2)--(3);
				\draw[ultra thick,gray] (3)--(4);
				\draw[ultra thick,gray] (1)--(4);
			\end{scope}

			\begin{scope}[shift={(-2.1,2.5)},scale=.3]
				
				\tikzstyle{every node}=[circle,draw,minimum size = 0.4cm]
				\node(1) at (13,5) {};
				\node(2) at (17,5) {};
				\node(4) at (17,0) {};
				\node(5) at (13,0) {};
				\node(3) at (19,2.5) {};
				\node(6) at (11,2.5) {};
				
				\tikzstyle{every node}=[]
				\node[] at (13,5) {\footnotesize 1};
				\node[] at (17,5) { \footnotesize 2};
				\node[] at (17,0) { \footnotesize 4};
				\node[] at (13,0) {\footnotesize 5};
				\node[] at (19,2.5) { \footnotesize 3};
				\node[] at (11,2.5) {\footnotesize 6};
				
				\draw[ultra thick,gray] (1)--(2);
				\draw[ultra thick,gray] (2)--(3);
				\draw[ultra thick,gray] (3)--(4);
				\draw[ultra thick,gray] (4)--(5);
				\draw[ultra thick,gray] (6)--(5);
				\draw[ultra thick,gray] (6)--(1);
				
				\tikzstyle{every node}=[]
				\node(1a) at (26,5) {\footnotesize $\ket{C_6}= \ket{110000} +\ket{011000} +$};
				\node(2a) at (27.55,3.7) { \footnotesize $\ket{001100} +\ket{000110} + $};
				\node(2a) at (27.1,2.4) { \footnotesize $\ket{000011} +\ket{100001}  $};
				
			\end{scope}
		\end{tikzpicture}
		\caption{Two cycle graphs: $C_4$ and $C_6$ with the corresponding states $\ket{C_4}$ and $\ket{C_6}$ listed on the right. Observe that the cycle $C_4$ forms a complete bipartite graph, 
			$C_4 =K_{22}$. According to \cref{cor5}, $\ket{C_6}$ is genuinely entangled while $\ket{C_4}$ is separable, as 
			$\ket{C_4} =(\ket{01} +\ket{10})_{13} \otimes (\ket{01} +\ket{10})_{24} $. 
		}
		\label{figY}
	\end{figure}
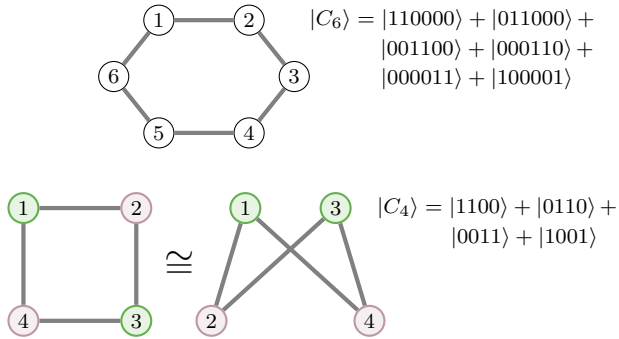
	
	\begin{figure}[h!]
		\definecolor{cof}{RGB}{219,144,71}
		\definecolor{pur}{RGB}{186,146,162}
		\definecolor{greeo}{RGB}{91,173,69}
		\definecolor{greet}{RGB}{52,111,72}
		\begin{tikzpicture}[thick,scale=3,
			vertex1/.style={circle,inner sep=4pt,draw=blue!95!white,fill=blue!15!white,thick,anchor=base},
			vertex2/.style={circle,inner sep=4pt,draw=greeo!95!white,fill=greeo!15!white,thick,anchor=base},
			vertex3/.style={circle,inner sep=4pt,draw=red!95!white,fill=red!15!white,thick,anchor=base}]
			\coordinate (A1) at (0,0) ;
			\coordinate (A2) at (0.6,0.2);
			\coordinate (A3) at (1,0);
			\coordinate (A4) at (0.4,-0.2);
			\coordinate (B1) at (0.5,0.5);
			\coordinate (B2) at (0.5,-0.5);
			\begin{scope}[thick,dashed,,opacity=0.6]
				\draw (A1) -- (A2) -- (A3);
				\draw (B1) -- (A2) -- (B2);
			\end{scope}
			\draw[fill=gray,opacity=0.2] (A1) -- (A4) -- (B1);
			\draw[fill=gray,opacity=0.24] (A1) -- (A4) -- (B2);
			\draw[fill=gray,opacity=0.07] (A3) -- (A4) -- (B1);
			\draw[fill=gray,opacity=0.12] (A3) -- (A4) -- (B2);
			\draw (B1) -- (A1) -- (B2) -- (A3) --cycle;
			
			\node[vertex1] at (A1)     {}; 
			\node(1) at (A1)  {\footnotesize 5};
			\node[vertex2] at (A2)     {}; 
			\node(1) at (A2)  {\footnotesize 4};
			\node[vertex1] at (A3)     {}; 
			\node(1) at (A3)  {\footnotesize 3};
			\node[vertex2] at (A4)     {}; 
			\node(1) at (A4)  {\footnotesize 2};
			\node[vertex3] at (B1)     {}; 
			\node(1) at (B1)  {\footnotesize 1};
			\node[vertex3] at (B2)     {}; 
			\node(1) at (B2)  {\footnotesize 6};
			
			\node(text) at (0.6,-0.8) { \footnotesize $(\ket{10} +\ket{01} )_{16} \otimes (\ket{10} +\ket{01} )_{24} \otimes (\ket{10} +\ket{01} )_{35} $};
		\end{tikzpicture}
		\caption{3-regular hypergraph on 6 vertices. Each hyperedge is represented by a triangle. Observe that the presented hypergraph is complete 3-partite hypergraph $K_{222}$. Indeed, vertices $1,6$ and $2,4$ and $3,5$ are coloured differently according to the partition. All hyperedges containing vertices of different colours are present. According to \cref{cor6} the state becomes separable across partition $16|24|35$, as it is indicated in the bottom.
		}
		\label{figYb}
	\end{figure}
	
	We shall discuss now how tight are the monogamy relations (\ref{eq5})
	for parties in an excitation-state. 
	Recall, that if aforementioned bound is saturated, the entanglement of a given part has 
	the bipartite form with other parties, 
	otherwise multipartite entanglement can be observed.
	
	For further discussion, we assume \textit{regularity} of a graph, 
	which means  that  the degree $d_v$ is constant for any vertex $v$. 
	Moreover, we assume an additional condition: 
	If vertices $v$ and $w$ are connected, they are connected with the same number of hyperedges. 
	In other words, the section $s_{vw} $ takes the same values depending on the distance between vertices:
	\begin{align*}
		s_{vw}=&\begin{cases}
			s \quad \; \; &\text{for} \quad d(v,w)=1,\\
			0 \quad &\text{for} \quad  d(v,w)>1 .
		\end{cases} 
	\end{align*}
	We call such a graph a \textit{distance-1 regular graph}. 
	In particular, for standard graphs this value reads 1 for neighboring nodes and 0 otherwise. 
	We shall mention here a particular class of graphs widely discussed in the graph theory. 
	A \textit{distance-regular graph} \cite{BANG20151} is a regular graph such that for any two vertices $v$ and $w$, the number of vertices at distance $j$ from $v$ and at distance $\ell$ from $w$ depends only upon $j, \ell$, and the distance $ d(v, w)$. 
	Our assumptions on (hyper)graphs 
	correspond to  distance-regular graphs with distances $d(v,w)=1$ and $d(v,w)=2$. 
	We conjecture  
	that regularity concerning the vertices of larger distances, $d(v,w)>2$ yields regularity in the amount of entanglement in subsystems consisting of 
	3 or more parties. 
	Unfortunately,  for larger systems  
	we have no satisfactory entanglement quantification \cite{Eltschka_2014} 
	required to formalize this statement. 
	From \cref{condition}, we derive an expression  
	for the concurrence $C_{vw}$ between subsystems
	corresponding to vertices of a distance-1 regular graph $G$.
	
	\begin{corollary}
		\label{A}
		For connected, and regular graph $G$ the concurrence $C_{vw}$  between two nodes $v$ and $w$ reads, 
		\begin{align}
			\label{eq7}
			C_{vw}(G) = \dfrac{2}{|E|} \cdot
			\begin{cases}
				\text{max} \{ 0,C \} \quad  &\text{for} \quad  d(v,w)=1 \\
				n_{vw} \quad &\text{for} \quad d(v,w)=2  \\
				0 \quad &\text{for} \quad  d(v,w)>2
			\end{cases} 
		\end{align}
		where $C= n_{vw} - \sqrt{s (|E| -2d+s)}$, $d$ is the degree of each node, 
		and $s$ is a section for each adjacent vertices.
	\end{corollary}
	
	\noindent
	Elementary calculations lead to the following result.
	
	\begin{corollary}
		\label{B}
		Square of the generalized concurrence $C_{v |\text{rest}}^2$ between 
		the particle $v$ and the rest of the system,
		expressed as a function of the number of edges $|E|$ 
		and the number of vertices $N$,
		 reads, 
		\begin{equation}
			\label{eq8}
			C_{v |\text{rest}}^2 = 4 \; \dfrac{d_v \big( |E|-d_v\big) }{|E|^2} =4 \; \dfrac{k(N-k)}{N^2},
		\end{equation}
		 see \cref{figZ}. 
		The second equation is valid under the assumption of the regularity of a
		$k$-hypergraph, i.e. the degrees of vertices are the same. 
	\end{corollary}
	
	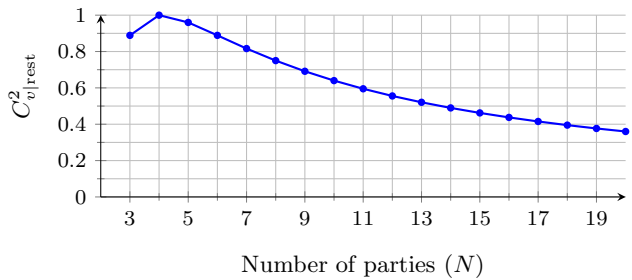
\begin{figure}[h!]
		\begin{tikzpicture}
			\begin{axis}[
				width=0.99\linewidth,
				height=4cm,
				axis lines=left,
				grid=both,
				legend style={draw=none, font=\small},
				legend cell align=left,
				legend style={at={(0.5,-0.25)}, anchor=north},
				label style={font=\small},
				ticklabel style={font=\footnotesize},
				xlabel style={at={(0.5, -0.05)}},
				ylabel style={at={(0.03, 0.5)}},
				ymin=0,
				xmin=2,
				ymax=1,
				xmax=20,
				xtick={3,5,...,19},
				ytick={0,0.2,...,1},
				minor tick num = 1,
				xlabel=Number of parties ($N$),
				ylabel={\color{white} shs}  $  C_{v|\text{rest}}^2$
				]
				\addplot[color=blue, mark=*, thick, mark options={scale=0.5}] coordinates {
					(3, 0.888889) (4, 1) (5,0.96) (6,0.888889) (7,0.8163) (8,0.75) (9,0.691358) (10,0.64) (11,0.595041) (12,0.555556) (13,0.52071) (14,0.489796) (15,0.462222) (16,0.4375) (17,0.415225) (18,0.395062) (19,0.376731) (20,0.36)};
			\end{axis}
		\end{tikzpicture}
		\caption{Squared concurrence (2-tangle) $C_{v |\text{rest}}^2$ between particle $v$ and the rest of the system depends on the number of vertices $N$ and uniformity $k$ only, so 
			it takes the same value for an arbitrary party of a regular $k$-hypergraph.
			In general, the tangle 
			$C_{v |\text{rest}}^2$ takes the maximal value for $k=N/2$. For instance, in the case of graphs, it takes the maximal permissible value for concurrence for regular graphs on $N=4$ vertices. Asymptotically, $C_{v |\text{rest}}^2 \sim 8/N$.}
		\label{figZ}
	\end{figure}
	
	\section{Examples of Dicke-like states}
	\label{sec3}
	
	In  previous sections we presented two similar, but different,
	constructions of genuinely entangled states:  excitation-states $\ket{G}$, related to a graph $G$, and Dicke-like states, determined by  a subgroup,
	which restricts the summation in the classical expression of Dicke. 
	The  (hyper)graph construction is by far more general, 
	as it 
	leads to quantum circuits corresponding to excitation states $\ket{G}$. 
	A Dicke-like state can be considered as
	a special case of the excitation-states,
	 which exhibits a certain symmetry structure. 
	
	We combine both representations and investigate highly symmetric objects, such as \textit{regular polygons},  \textit{Platonic solids}, and \textit{regular plane tilings}. 
	It is worth mentioning that such symmetric objects, especially Platonic solids, were already used in various contexts concerning multipartite entanglement including: 
	quantification of entanglement of permutation-symmetric states \cite{Markham_2011,Martin15}, identification of quantumness of a state \cite{Goldberg_2020}, search for the maximally entangled symmetric state \cite{Aulbach_2010}, or general geometrical quantification of entanglement \cite{PhysRevA.85.032314,Giraud_2010} especially among states with imposed symmetries on the roots of Majorana polynomial \cite{Martin10,Martin15}. 

	We present particular examples of quantum states shared among $N$ parties positioned in a highly symmetric way. 
	By using results from \cref{cccc}, we discuss their entanglement properties. 
	Firstly, we computed the concurrence in two-partite subsystems for all presented examples. 
	Recall, that for a given number of excitations $k$,
	the entanglement shared between a particular node $v$ and the rest of the system in an excitation-state depends only on the number of parties $N$, 
	\[
	C_{v |\text{rest}}^2 = 4\; \dfrac{k(N-k)}{N^2},
	\]
	as it was shown in \cref{eq8}.
	We define the \textit{entanglement ratio} $\Gamma_v$ for the node $v$ as:
	\[
	\Gamma_v :=\dfrac{\sum_{i \neq v} C_{v|i}^2}{ C_{v|\text{rest}}^2} \in \left[ 0,1\right] , 
	\]
	which measures the ratio of entanglement shared between particular parties in a two-partite way in comparison to the amount of entanglement shared in the multi-partite way. 
	Since concurrence satisfies monogamy inequality (\ref{eq5}),
	the parameter  $\Gamma_v$ takes values in the range $[0,1]$. 
	\cref{Comparison} compares the entanglement ratio of various examples in the context of the application of provided states.

	\subsection{Dicke states}
	
	The Dicke states $\ket{D_N^k}$ are excitation-states for complete $k$-regular hypergraphs on $N$ vertices, see \cref{ex1}. Their group of symmetry is the
	full permutation group $\mathcal{S}_N$. 
	In a particular case, if the number of excitations is equal to 1, 
	the Dicke state coincides with the state $\ket{W}_N$
	and the bound in \cref{eq5} is tight \cite{Osborne_2006}. 
	
	Consider any vertex $v$ and its neighbour $w$
	of a complete $k$-regular hypergraph. The parameteres read, 
	the section is $s = {{N-2}\choose{k-2}}$, the degree $d = {{N-1}\choose{k-1}}$, neighbourhood $n_{vw} ={{N-2}\choose{k-1}}$, and the number of all
	 hyperedges $|E|= {{N}\choose{k}}$. By \cref{eq7,eq8}, and elementary calculation, the concurrence in two-partite subsystems 
	takes the following value:
	\small
	\begin{align*}
		C_{v|w} &=2 {{N}\choose{k}}^{-1}  \Bigg( {{N-2}\choose{k-1}}- \sqrt{{{N-2}\choose{k}}{{N-2}\choose{k-2}}} \Bigg) .
	\end{align*} 
	\normalsize
	Observe that formula above is invariant 
	with respect to the change $k \leftrightarrow N-k$. 
	Indeed, Dicke states $\ket{D_N^k}$ and $\ket{D_N^{N-k}}$ are equivalent up to local change of basis $\ket{0} \leftrightarrow \ket{1}$, hence they share
	the same entanglement properties. 
	The entanglement ratio $\Gamma_v$ for a given node $v$ is equal to
	\footnotesize
	\[
	\Gamma_v  (\text{D}^{k}_N) 
	=\frac{N-1}{
	{{N-1}\choose{k}} 
	{{N-1}\choose{k-1}}  }
	\Bigg( {{N-2}\choose{k-1}}- \sqrt{{{N-2}\choose{k}}{{N-2}\choose{k-2}}} \Bigg)^2
	.
	\]
	\normalsize
	In a paricular case, for $k=1,N-1$, the entanglement ratio saturates its range, $\Gamma_v =1$. 
	In these  two cases the Dicke states belong to the family $\ket{W}$ 
	and saturate the monogamy bound (\ref{eq5})
	derived in \cite{DistributedEntanglement,Osborne_2006}.
	
	\begin{figure}[h!]
		\begin{tikzpicture}
			\begin{scope}[shift={(2.2,0)},scale=.35]
				\def\r{2cm}
				\foreach \t in {0,...,4} {
					\node[circle,draw, outer sep=0mm, inner sep=0mm, minimum size=1mm] (a\t) at ({90+\t*72}:\r) {} ;
				}
				
				\draw (a0) edge (a1) ;
				\draw (a1) edge (a2) ;
				\draw (a2) edge (a3) ;
				\draw (a3) edge (a4) ;
				\draw (a4) edge (a0) ;
				\draw (a0) edge (a2) ;
				\draw (a1) edge (a3) ;
				\draw (a2) edge (a4) ;
				\draw (a3) edge (a0) ;
				\draw (a4) edge (a1) ;
				
				\tikzstyle{every node}=[]
				\node[] at (0,-3) {$\ket{D_5^2}$};
			\end{scope}
			\begin{scope}[shift={(4.2,0)},scale=.35]
				\def\r{2cm}
				\foreach \t in {0,...,5} {
					\node[circle,draw, draw=black!95!white,fill=black!95!white,outer sep=0mm, inner sep=0mm, minimum size=1mm] (a\t) at ({90+\t*60}:\r) {} ;
				}
				
				\draw (a0) edge (a1) ;
				\draw (a1) edge (a2) ;
				\draw (a2) edge (a3) ;
				\draw (a3) edge (a4) ;
				\draw (a4) edge (a5) ;
				\draw (a5) edge (a0) ;
				\draw (a1) edge (a3) ;
				\draw (a2) edge (a4) ;
				\draw (a3) edge (a5) ;
				\draw (a4) edge (a0) ;
				\draw (a5) edge (a2) ;
				\draw (a1) edge (a4) ;
				\draw (a2) edge (a5) ;
				\draw (a3) edge (a0) ;
				\draw (a1) edge (a5) ;
				\draw (a2) edge (a0) ;
				
				\tikzstyle{every node}=[]
				\node[] at (0,-3) {$\ket{D_6^2}$};
			\end{scope}

			\begin{scope}[shift={(0,0)},scale=.4]
				\def\r{2cm}
				\foreach \t in {0,...,5} {
					\node[circle,draw, draw=black!95!white,fill=black!0!white,outer sep=0mm, inner sep=0mm, minimum size=4mm] (a\t) at ({90+\t*60}:\r) {} ;   
					\node[circle,draw, draw=black!95!white,fill=black!95!white,outer sep=0mm, inner sep=0mm, minimum size=1mm] (a\t) at ({90+\t*60}:\r) {} ;
				} 
				
				\tikzstyle{every node}=[]
				\node[] at (0,-3) {$\ket{D_6^1} =\ket{W}_6 $};
			\end{scope}

			\begin{scope}[shift={(5.5,-0.2)},scale=2]
				\coordinate (A1) at (0,0) ;
				\coordinate (A2) at (0.8,0.2);
				\coordinate (A4) at (0.4,-0.2);
				\coordinate (B1) at (0.5,0.5);
				\begin{scope}[thick,dashed,,opacity=0.6]
					\draw (A1) -- (A2) ;
				\end{scope}
				\draw[fill=gray,opacity=0.25] (A1) -- (A4) -- (B1);
				\draw[fill=gray,opacity=0.12] (A2) -- (A4) -- (B1);
				\draw (B1) -- (A1) -- (A4) -- (A2) -- (B1) -- (A4) -- cycle;
				
				\tikzstyle{every node}=[circle,draw, draw=black!95!white,fill=black!95!white,outer sep=0mm, inner sep=0mm, minimum size=1mm] 
				\node[] at (A1)     {}; 
				\node[] at (A2)     {}; 
				\node[] at (A4)     {}; 
				\node[] at (B1)     {}; 
				\tikzstyle{every node}=[] 
				\node[] at (0.5,-0.5) {$\ket{D_4^3}$};
			\end{scope}  
		\end{tikzpicture}
		\caption{ Dicke state $\ket{\text{D}_N^k}$ is associated with a complete $k$-regular hypergraph, while 1-uniform complete hypergraph leads to 
		the state $\ket{\text{W}_N}$.
			The graph and the corresponding state  is completely symmetric, so 
			the labels of the nodes can be omitted. Dicke states $\ket{\text{D}_N^2}$ with two excitations are related to complete graphs on $N$ vertices.  
			A tetrahedron represents the Dicke states $\ket{\text{D}_4^1}$, $\ket{\text{D}_4^2}$ and $\ket{\text{D}_4^3}$, depending on whether we
			consider  the vertices, the edges or the faces of the tetrahedron. 
		Each  face represents the hyperedge of a 3-uniform complete hypergraph on 4 vertices,  relevant to $\ket{\text{D}_4^3}$. The graph formed by edges of a tetrahedron forms the complete graph on 4 vertices,  relevant to $\ket{\text{D}_4^2}$. Four vertices of the tetrahedron provide a 1-uniform graph, related to $\ket{\text{D}_4^1}$. This generalizes to the fact that the simplex of dimesnion $N-1$ represents the full family $\left\{\ket{\text{D}_N^k}\right\}_{k=1}^{N-1}$ of Dicke states of $N$ vertices.}
		\label{ex1}
	\end{figure}
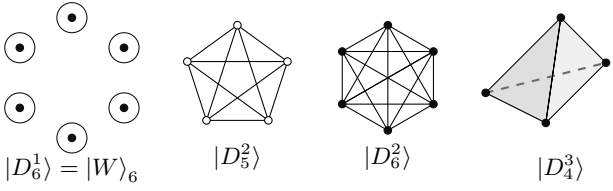

\begin{corollary}
		For the Dicke states $\ket{\text{D}^{k}_N}$ the entanglement ratio at infinite dimension is nonzero,
		\begin{equation*}
		\lim_{N\rightarrow\infty} \Gamma_v (\text{D}^{k}_N) = 2k-1- \sqrt{k(k-1)}.
		\end{equation*}
In particular, for states related to graphs, $k=1$, we find
		\begin{equation*}
			\lim_{N\rightarrow\infty} \Gamma_v(\text{D}^{1}_N) = 3 - 2\sqrt{2}.
		\end{equation*}
	\end{corollary}

	\subsection{Cyclic states}
	
	The simplest non-trivial subgroup of the permutation group
	$\mathcal{S}_N$ is a cyclic group $\mathcal{C}_N$. 
	In general states $\ket{D_N^k}_{\mathcal{C}_N}$ are \textit{translationally invariant}, i.e. invariant under a cyclic permutation of qubits  \cite{TranslationallyInvariantStates,watson2020complexity}. 
	 Family of such states is widely considered in several 1D models 
	 applied in condensed-matter physics, like XY model or the  Heisenberg model. 
	For further considerations, assume that the number of excitations is equal to $k=2$. 
	According to \cref{Dickelike}, we can construct the Dicke-like state $\ket{D_N^2}_{\mathcal{C}_N}$ by taking the superposition of all elements:
	\[
	\ket{0\ldots 0110\ldots 0},
	\] 
	where excitations are always on adjacent positions, see \cref{CyclicStates}. 
	
	On the other hand,  the state $\ket{D_N^2}_{\mathcal{C}_N}$ can be constructed as an excitation-state. To this end consider the cyclic graph on $N$ vertices. 
	The relevant excitation-state, denoted by $\ket{C_N}$, 
	matches perfectly $\ket{D_N^2}_{\mathcal{C}_N}$. 
	The group of automorphisms of a cyclic graph $C_N$ forms
	a \textit{Dihedral} group $\mathcal{D}_{2N}$, where the lower index 
	stands for the number of elements in the group $|\mathcal{D}_{2N} |=2N$. 
	It describes the group of symmetries of a regular $N$-polygon, which 
	consists of rotations and reflections of the figure.
	Each reflection is an idempotent element, the order of rotations may vary. 
	Using those properties of a cyclic graph $C_N$, we conclude that the group of symmetries
	of the quantum state $\ket{D_N^2}_{\mathcal{C}_N} = \ket{C_N}$ is a 
	dihedral group $\mathcal{D}_{2N}$. 
	Notice that the systems with dihedral $\mathcal{D}_{2N}$ symmetries were consider in a context of correlation theory of the chemical bond \cite{ding2020concept}. 
	Molecules invariant under a rotation and inversion were investigated ibidem. 
	\cref{AllStellar} presents the stellar representation of the six qubit state
	$\ket{C_6}$, which serves as an example of the relation
	between the Dicke-like states  and the construction of excitation-states. 
	
	In order to analyze entanglement properties of $\ket{C_N}$
	we calculate the parameters of the cyclic graph $C_N$. 
	The number of edges $|E|= N$, while for any vertex $v$ its  degree 
	reads, $d_v=2$. 
	Any two vertices $v$ and $w$ might be related in three following ways: either they are neighbours and then $s=1$, $n_{vw} =0$; or they are of distance $d(v,w)=2$ and then $n_{vw}=2$, or the distance $d(v,w)>0$. 
	The concurrence in two-partite subsystems takes the following value:
	\small
	\begin{align*}
		C_{v|w} ({C_N}) &=
		\begin{cases}
			2/N \quad \; \; &\text{for} \quad d(v,w)=2\\
			0 \quad &\text{otherwise} 
		\end{cases} ,
	\end{align*} 
	\normalsize
	with the small correction for $N=4$, where $C^2_{v|w} =1$ for  vertices
	of distance two. 
	The entanglement ratio $\Gamma_v$ for a given node reads,
	\[
	\Gamma_v (C_N) =\dfrac{1}{N-2},
	\]
	with the same correction in the case  $N=4$, for which $\Gamma_v=1$.
	
	\begin{figure}[h!]
		\begin{tikzpicture}
			\definecolor{greeo}{RGB}{91,173,69}
			\begin{scope}[shift={(-2.5,0)},scale=.45]
				\def\r{2cm}
				\foreach \t in {0,...,5} {
					\node[circle,draw, draw=black!95!white,fill=black!95!white,outer sep=0mm, inner sep=0mm, minimum size=1.5mm] (a\t) at ({90+\t*60}:\r) {} ;
				}
				
				\draw[thick,gray] (a0) edge (a1) ;
				\draw[thick,gray]  (a1) edge (a2) ;
				\draw[thick,gray]  (a2) edge (a3) ;
				\draw[thick,gray]  (a3) edge (a4) ;
				\draw[thick,gray] (a4) edge (a5) ;
				\draw[thick,gray] (a5) edge (a0) ;

				\draw[greeo]  (a0) edge (a2) ;
				\draw[greeo]  (a1) edge (a3) ;
				\draw[greeo]  (a2) edge (a4) ;
				\draw[greeo]  (a3) edge (a5) ;
				\draw[greeo]  (a4) edge (a0) ;
				\draw[greeo]  (a5) edge (a1) ;
				\tikzstyle{every node}=[]
				\node[] at (0,-3) {$\ket{C_6}$};
			\end{scope}

			\begin{scope}[shift={(0,0)},scale=.45]
				\def\r{2cm}
				\foreach \t in {0,...,6} {
					\node[circle,draw, draw=black!95!white,fill=black!95!white,outer sep=0mm, inner sep=0mm, minimum size=1.5mm] (a\t) at ({90+\t*(360/7)}:\r) {} ;
				}
				
				\draw[thick,gray] (a0) edge (a1) ;
				\draw[thick,gray]  (a1) edge (a2) ;
				\draw[thick,gray]  (a2) edge (a3) ;
				\draw[thick,gray]  (a3) edge (a4) ;
				\draw[thick,gray] (a4) edge (a5) ;
				\draw[thick,gray] (a5) edge (a6) ;
				\draw[thick,gray] (a6) edge (a0) ;

				\draw[greeo]  (a0) edge (a2) ;
				\draw[greeo]  (a1) edge (a3) ;
				\draw[greeo]  (a2) edge (a4) ;
				\draw[greeo]  (a3) edge (a5) ;
				\draw[greeo]  (a4) edge (a6) ;
				\draw[greeo]  (a5) edge (a0) ;
				\draw[greeo]  (a6) edge (a1) ;
				\tikzstyle{every node}=[]
				\node[] at (0,-3) {$\ket{C_7}$};
				\node[] at (0,-4.6) {$\ket{C_7} = \ket{1100000}+\ket{0110000}+\ket{0011000} +\ket{0001100}+$};
				\node[] at (-1.25,-5.6) {$\ket{0000110}+\ket{0000011}+\ket{1000001} $};
			\end{scope}    
			
			\begin{scope}[shift={(2.5,0)},scale=.45]
				\def\r{2cm}
				\foreach \t in {0,...,7} {
					\node[circle,draw, draw=black!95!white,fill=black!95!white,outer sep=0mm, inner sep=0mm, minimum size=1.5mm] (a\t) at ({90+\t*45}:\r) {} ;
				}
				
				\draw[thick,gray] (a0) edge (a1) ;
				\draw[thick,gray]  (a1) edge (a2) ;
				\draw[thick,gray]  (a2) edge (a3) ;
				\draw[thick,gray]  (a3) edge (a4) ;
				\draw[thick,gray] (a4) edge (a5) ;
				\draw[thick,gray] (a5) edge (a6) ;
				\draw[thick,gray] (a6) edge (a7) ;
				\draw[thick,gray] (a7) edge (a0) ;

				\draw[greeo]  (a0) edge (a2) ;
				\draw[greeo]  (a1) edge (a3) ;
				\draw[greeo]  (a2) edge (a4) ;
				\draw[greeo]  (a3) edge (a5) ;
				\draw[greeo]  (a4) edge (a6) ;
				\draw[greeo]  (a5) edge (a7) ;
				\draw[greeo]  (a6) edge (a0) ;
				\draw[greeo]  (a7) edge (a1) ;
				\tikzstyle{every node}=[]
				\node[] at (0,-3) {$\ket{C_8}$};
			\end{scope}  
		\end{tikzpicture}
		\caption{Three cycle graphs: $C_6$, $C_7$ and $C_8$ indicated by gray edges. One of the 
			corresponding states, $\ket{C_7}$, is listed below. 
			Entanglement  on two-parties subsystem 
			is vanishing except of parties of distance two, 
			$d(v,w)=2$, 
			for which the concurrence $C_{vw}$ reads,
			$2/6, 2/7$ and $2/8$ for states  $\ket{C_6}$, $\ket{C_7}$, and $\ket{C_8}$,
			respectively. 
			Non-vanishing concurrence is indicated by green lines connecting relevant parties,
			so  that the  graph of entnaglement splits into two for an even number of parties,
			$N=6$ or $8$. 
		}
		\label{CyclicStates}
	\end{figure}
	
	We shall emphasize that the cyclic and Dihedral groups were already successfully used for constructing fully symmetric states with additional (cyclic/Dihedral) symmetry imposed on roots of related Majorana polynomials \cite{Martin15}. 
	In our work, instead of imposing additional symmetries on roots, we restrict the symmetrization procedure in (\ref{eq2}) to the subgroup $H<\mathcal{S}_N$ of the symmetric group $\mathcal{S}_N$.
	
	\subsection{Platonic states}
	\label{Platonic states}
	
	Platonic solids --  highly symmetric geometric objects 
	in ${\mathbbm R}^3$
	--  were used to construct quantum states in various contexts \cite{Giraud_2010,PhysRevA.85.032314,Goldberg_2020}.  
	With any  Platonic solid we associate 	in this work
	two quantum excitation-states. 
	Firstly, we can construct an excitation-state simply by looking at the edges of a solid. 
	Second, we may associate faces of a solid with hyperedges of a hypergraph. 
	For instance, the tetrahedron is linked to the Dicke state $\ket{\text{D}_4^2}$ (by reading edges) and $\ket{\text{D}_4^3}$ (by reading faces), see \cref{ex1}. 
	
	We use the notation $\ket{P^e}_N$ for Platonic states constructed by looking at the edges of a solid, while $\ket{P^f}_N$ refers to the construction involving faces ($N$ denotes the number of vertices in related Platonic solid). 
	
	Due to the extraordinary symmetry of Platonic solids, it is not surprising that both constructions of quantum state lead to Dicke-like states. 
	In order to show it, we need to find appropriate groups, which action yield to Platonic states. 
	Consider the automorphism group of relevant graphs or hypergraphs. 
	An elementary argument from the representation theory shows that it is determined by the group of symmetries of the Platonic solid.
	
	\begin{fact}
		Consider a Platonic solid. Each permutation $\sigma \in \mathcal{S}_{|V|}$ of vertices which preserves the structure of faces or edges represents a symmetry of the solid.
	\end{fact}
	\noindent
	Since groups of symmetries of Platonic solids are transitive on edges and faces, 
	any such state can be realized by an action of the aforementioned group.

	There exist five Platonic solids: a self-dual tetrahedron and two dual pairs
	the same group of symmetries:
	a) cube and octahedron and b) dodecahedron and icosahedron. 
	Any symmetry of a solid determines a permutation $\sigma \in \mathcal{S}_N$ on vertices. 
	For tetrahedron, the reverse statement holds true. 
	For the first dual pair, consider a cube. 
	Symmetries of the cube coincide with permutations of its four diagonals plus
	one non-rotation symmetry, the 
	point reflection with respect to the center.
	For the second dual pair, consider dodecahedron.
	One can embed five cubes in the dodecahedron,
	forming the compound of five cubes.
	Any rotation of dodecahedron determines an even permutation of the cubes.  
	As in the previous case, there exists additionally a reflection symmetry with respect 
	to  the center.  This classical analysis 
	implies the following statement.
	
	
	\begin{corollary}
		\label{prop3}
		The symmetry group of  Platonic states  $\ket{P^e}_4$ and  $\ket{P^f}_4$
		related to  the tetrahedron is the full symmetric group $\mathcal{S}_4$. Indeed, 
		\[
		\ket{P^e}_4 = \ket{D^2_4} ,\quad \ket{P^f}_4 = \ket{D^3_4}.
		\]
		The symmetry group of $\ket{P^e}_6$ and $\ket{P^f}_6$ related to the octahedron,
		and $\ket{P^e}_8$, $\ket{P^f}_8$ related to the cube, is 
		\[
		\mathcal{S}_4 \times \mathcal{S}_2, 
		\]
		non-trivially embedded into $\mathcal{S}_6$ and $ \mathcal{S}_8$, respectively. 
		The symmetry group of states $\ket{P^e}_{12}$ and  $\ket{P^f}_{12}$ related to
		the  dodecahedron, and $\ket{P^e}_{20}$, $\ket{P^f}_{20}$ related to  icosahedron, is 
		\[
		\mathcal{A}_5 \times \mathcal{S}_2, 
		\]
		non-trivially embedded into $\mathcal{S}_{12}$ and $ \mathcal{S}_{20}$, respectively. 
	\end{corollary}
	
	\noindent
	Note that the Platonic states listed above form non-trivial examples of symmetric states. 
	In particular, in the case of dodecahedron
	the alternating symmetry $A_5$ is observed, which is not easy to construct. 
	
	
	For Platonic states determined by the edges of the solid,
	the concurrence in two-partite subsystems takes positive value only for nodes of distance two. 
	For such a nodes $1$ and $2$, the concurrence $C_{12}  $ and the entanglement ratio $\Gamma_v$ is presented in  Table I.
	\begin{table}[ht]
		\label{tab1}
		\centering 
		\begin{tabular}{l c c c c c} 
			\hline
			& $\quad \; \ket{P^e}_4\;$ & $\;\ket{P^e}_6$ & $\;\ket{P^e}_{8}$ &$\;\ket{P^e}_{12}$&$\;\ket{P^e}_{20}$ \\ [0.5ex] 
			\hline 
			Concurrence $C_{12}  $& 0.333 & 0.667 & 0.333 & 0.133 & 0.067 \\ [1ex]
			Ent. Ratio $\Gamma_v$ & 0.333& 0.500 & 0.444 & 0.160 & 0.074\\
			\hline 
		\end{tabular}
		\caption{Concurrence  $C_{12}$ and entanglement ratio $\Gamma_v$
		for two-qubit systems obtained by partial trace of $N-2$ subsystems of distance-two of Platonic states determined by the edges of the solid.}
	\end{table}
	

	\begin{figure}
		\def \phi {1.617}
		\begin{tikzpicture}[
			x={(-0.86in, -0.5in)}, y = {(0.86in, -0.5in)}, z = {(0, 1in)},
			rotate = 22,
			scale = 0.25,
			every node/.style = {
				circle, fill = black!95, inner sep = 0pt, minimum size = 0.2cm
			},
			foreground/.style = {ultra thick, gray },
			background/.style = { dashed, gray }
			]
			\definecolor{greeo}{RGB}{91,173,69};
			\coordinate (9) at (0, -\phi*\phi,  \phi);
			\coordinate (8) at (0,  \phi*\phi,  \phi);
			\coordinate (12) at (0,  \phi*\phi, -\phi);
			\coordinate (5) at (0, -\phi*\phi, -\phi);
			\coordinate (7) at ( \phi, 0,  \phi*\phi);
			\coordinate (3) at (-\phi, 0,  \phi*\phi);
			\coordinate (6) at (-\phi, 0, -\phi*\phi);
			\coordinate (4) at ( \phi, 0, -\phi*\phi);
			\coordinate (2) at ( \phi*\phi,  \phi, 0);
			\coordinate (10) at (-\phi*\phi,  \phi, 0);
			\coordinate (1) at (-\phi*\phi, -\phi, 0);
			\coordinate (11) at ( \phi*\phi, -\phi, 0);
			
			\draw[greeo] (7) -- (1);
			\draw[greeo] (7) -- (10);
			\draw[greeo] (7) -- (5);
			\draw[greeo] (7) -- (4);
			\draw[greeo] (7) -- (12);

			\draw[foreground] (10) -- (3) -- (8) -- (10) -- (12) -- (8);
			\draw[foreground] (4) -- (12) -- (2) -- (4) -- (11) -- (2);
			\draw[foreground] (9) -- (3) -- (7) -- (9) -- (11) -- (7);
			\draw[foreground] (7) -- (8) -- (2) -- cycle;
			\draw[background] (12) -- (6) -- (10) -- (1) -- (6) -- (5) -- (1)
			-- (9) -- (5) -- (11);
			\draw[background] (5) -- (4) -- (6);
			\draw[background] (3) -- (1);

			\foreach \n in {1,...,12}
			\node at (\n) {};
			\node[circle,draw=gray!95!white, fill = greeo!10, inner sep = 0pt, minimum size = 0.6cm] at (7) {\footnotesize 1};
			\node[circle,draw=gray!95!white, fill = greeo!95, inner sep = 0pt, minimum size = 0.2cm] at (10) {};
			\node[circle,draw=gray!95!white, fill = greeo!95, inner sep = 0pt, minimum size = 0.2cm] at (12) {};
			\node[circle,draw=gray!95!white, fill = greeo!95, inner sep = 0pt, minimum size = 0.2cm] at (4) {};
			\node[circle,draw=gray!95!white, fill = greeo!95, inner sep = 0pt, minimum size = 0.2cm] at (5) {};
			\node[circle,draw=gray!95!white, fill = greeo!95, inner sep = 0pt, minimum size = 0.2cm] at (1) {};
		\end{tikzpicture}
		\caption{Icosahedron, one of five Platonic solids. Its edges (indicated with gray)
			determine the excitation state $\ket{P^e}_{20}$ on 20 qubits. 
			For a chosen node $1$, the concurrence of the two-party subsystem $C_{1 v}$ is vanishing except for subsystems corresponding to distance-two nodes
			indicated by green lines and green nodes. }
		\label{PlatonicStates}
	\end{figure}
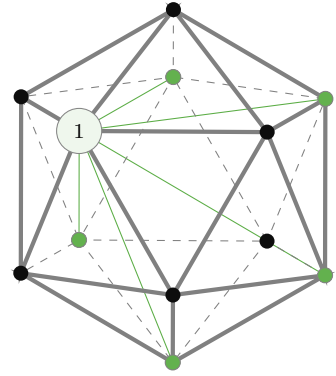

	\begin{figure}[h!]
		\subfloat{
			\begin{blochsphere}[radius=1.1 cm,tilt=15,rotation=-20]
				\drawBallGrid[style={opacity=0.1}]{30}{30}
				
				\drawGreatCircle[style={dashed}]{0}{0}{0}
				
				\drawAxis[style={draw=black}]{0}{0}
				
				\drawAxis[color=black]{90}{0}
				
				\labelLatLon{up}{90}{0};
				\labelLatLon{down}{-90}{90};
				\node[circle,fill,red] at (0,1.1){.};
				\node[circle,fill,red] at (0,1.3){.};
				\node[circle,fill,red] at (0,1.5){.};
				\node[circle,fill,red] at (0,1.7){.};
				\node[circle,fill,red] at (0,-1.1){.};
				\node[circle,fill,red] at (0,-1.3){.};
				\node[above] at (0,1.8) {{ \footnotesize $ \ket{\phi_3},\ket{\phi_4},\ket{\phi_5},\ket{\phi_6}$ }};
				\node[below] at (0,-1.4) {{ \footnotesize $\ket{\phi_1},\ket{\phi_2}$}};
				\labelLatLon{one}{0}{0};
			\end{blochsphere}  
		}
		\begin{tikzpicture}
			\begin{scope}[scale=.3]
				\tikzstyle{every node}=[] 
				\node[] at (2,8) {$\curvearrowleft \mathcal{S}_6$};
				\node[] at (2,6) {$\curvearrowleft \mathcal{C}_6$};
				\node[] at (2,4) {$\curvearrowleft \mathcal{D}_{12}$};
				\node[] at (1.8,2) {$\curvearrowleft \mathcal{S}_4 \times \mathcal{C}_2$};
			\end{scope}  
			
			\begin{scope}[scale=.3]
				\tikzstyle{every node}=[] 
				\node[] at (6,8) {$\Rightarrow$};
				\draw [decorate,decoration={brace,amplitude=4pt}]
				(5,7)--(5,3) node[midway, right, font=\footnotesize, xshift=2pt] {$\Rightarrow$};
				\node[] at (6,2) {$\Rightarrow$};
			\end{scope}  
			
			\begin{scope}[scale=.3]
				\tikzstyle{every node}=[] 
				\node[] at (8,8) {$\ket{D_6^2}$};
				\node[] at (8,5) {$\ket{C_6}$};
				\node[] at (8,2) {$\ket{P^e}_6$};
			\end{scope}  
		\end{tikzpicture}
		\caption{Stellar representation of the Dicke-like states $\ket{D_6^2}_H$ relevant to the action of various groups $H <\mathcal{S}_6$. 
			Acting by the full symmetric group $\mathcal{S}_6$ on a Bloch sphere we obtain the Dicke state $\ket{D_6^2}$. By action of a cyclic group $\mathcal{C}_6$ or the dihedral group $\mathcal{D}_{12}$, we obtain the cyclic state $\ket{C_6}$. Choosing action of $\mathcal{S}_4 \times \mathcal{C}_2$, as described in \cref{Platonic states}, we obtain the 
			Platonic state $\ket{P^e}_6$.}
		\label{AllStellar}
	\end{figure}
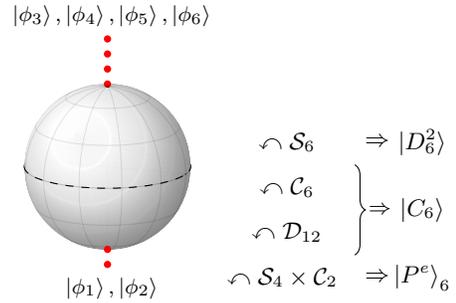

\subsection{Regular $m$-polytope families}
	
	The line of thinking applied to the Platonic solids can be extended to the regular polytopes in higher dimensions. Such families include the self-dual $m$-simplices and $m$-hypercubes with dual $m$-orthoplexes. Each of these polytopes provides us with a set of $k$-uniform hypergraphs for $1 \leq k \leq m-1$ defined by the set of their $k$-dimensional hyperedges. For the sake of this section we will denote the states corresponding to the $k$-dimensional hyperedges of $m$-simplex as $\ket{\text{S}^k_m}$ and analogously $\ket{\text{B}^k_{2^m}}$ for $m$-hypercubes and $\ket{\text{O}^k_{2m}}$ for $m$-orthoplexes, where the lower index stands for the number of subsystems $N$ in the state ($N=m$, $2^m$, $2m$ for $m$-simplex, $m$-hypercube, and $m$-orthoplexe respectively). 
	
	First, we may state a simple observation concerning $m$-simplices
	
	\begin{observation}
		The states related to the $m$-simplices are equivalent to the Dicke states,
		$$
			\ket{\text{S}^k_m} = \ket{D^k_m}.
		$$
		and the symmetry group is trivially $\mathcal{S}_m$.
	\end{observation}

	The symmetry group of $\ket{\text{B}^k_{2^m}}$ and $\ket{\text{O}^k_{2m}}$ is given by the non-trivial hyperoctahedral group $B_m$. We may first formulate a simple observation about the orthoplectic states
	
	\begin{observation}
		The states $\ket{\text{O}^{m-1}_{2m}}$ are separable with respect to
		the partition $12|34|...|(2m-1)2m$, with each pair of vertices being 2-distance or, in other words, lying on a common diagonal of the related $m$-orthoplex. 
	\end{observation}
	
	\noindent
	Furthermore, one can easily calculate the concurrence $C_{vw}$ for the states $\ket{\text{O}^{2}_{2m}}$ connected to the 2-edges of the $m$-orthoplex,
{\small
	\begin{equation}
		C_{vw} (\text{O}^{2}_{2m}) = \begin{cases}
			\text{max}\{0,\,\frac{\sqrt{2 m^2-4 m+3}-2 m+4}{2 m-2 m^2}\}   &  
			\!\! \text{for} \  d(v,w)=1 \\
			2/m    &\!\! \text{for} \  d(v,w)=2\\
			0    &\!\! \text{for} \  d(v,w)>2,
		\end{cases}
	\end{equation}
	}
	and similarly the entanglement ratio
	\begin{equation}
		\Gamma_v (\text{O}^{2}_{2m}) = \begin{cases}
			1/(m-1)  \\ 
			 \frac{6 m^2-4 m\left(\sqrt{2 m^2-6 m+5}+5\right) +8 \sqrt{2 m^2-6 m+5}+19}{2 (m-1)^2},  
		\end{cases}
	\end{equation}
where the first case holds for $m \leq 3$
and the second one for $m>3$.
	\begin{corollary}
		Similarly to Dicke states, the entanglement ratio for
		 $\ket{\text{O}^{2}_{2m}}$ states converges to nonzero value
		\begin{equation}
			\lim_{n\rightarrow\infty} \Gamma_v (\text{O}^{2}_{2m}) = 3 - 2\sqrt{2}\approx 0.1716.
		\end{equation}
	\end{corollary}

	The situation is much simpler for the hypercubic states $\ket{\text{B}^{2}_{2^m}}$, where the concurrence occurs only between the distance-2 vertices,
	 $C_{12} (\text{B}^{2}_{2^m})= 2^{3-m}/m$, while the entanglement ratio reads 
	$$\Gamma_v (\text{B}^{2}_{2^m}) = \frac{4 (m-1)}{\left(2^m-2\right) m},$$ which asymptotically tends to zero.

	\begin{figure}[h!]
		\begin{tikzpicture}			
		\begin{axis}[  
				width=0.99\linewidth,
				height=5cm,
				axis lines=left,
				grid=both,
				legend style={draw=none, font=\small},
				legend cell align=left,
				legend style={at={(0.45,-0.45)}, anchor=north},
				label style={font=\small},
				ticklabel style={font=\footnotesize},
				xlabel style={at={(0.5, -0.05)}},
				ylabel style={at={(0.04, 0.5)}},
				ymin=0,
				xmin=6,
				ymax=0.55,
				xmax=61.5,
				xtick={6,14,...,60},
				minor tick num = 1,
				ytick={0, 0.1,...,.7},
				xlabel=Number of parties ($ N$),
				ylabel=Entanglement ratio   $\Gamma_v$,
				legend entries={{Ent. ratio  $\Gamma_v$ for orthoplectic states $\ket{\text{O}^2_{2m}}$, $N=2m$} , {Distance one nodes contribution to  $\Gamma_v$}, {Distance two nodes contribution to  $\Gamma_v$ }},
				]
				\addplot[color=violet, mark=*, thick, mark options={scale=.8}] coordinates  {
					(6, 0.5)  (8,0.341977) (10,0.28125) (12,0.251) (14,0.233264) (16,0.221744) (18,0.213718) (20,0.207833) (22,0.203348) (24,0.199823) (26,0.196986) (28,.194655) (30,.192708) (32,.191058) (34, .189643) (36,0.188417) (38,0.187344) (40,0.186398) (42,0.185558) (44,0.184807) (46,0.184132) (48,0.183522) (50,0.182967) (52,0.182462) (54,.181999) (56,0.1815737) (58,0.181181) (60,0.180818) (62,0.180481)};
				\addplot[color=orange, mark=diamond*, thick, mark options={scale=.7}] coordinates  {
					(6, 0)  (8,0.00864388) (10,0.0510002) (12,0.0665973) (14,0.0665973) (16,0.0788871) (18,0.0887181) (20,0.0967221) (22,0.103348) (24,0.108914) (26,0.113652) (28,0.117732) (30,0.121279) (32,0.124391) (34,0.127143) (36,0.129593) (38,0.131789) (40,0.133767) (42,0.135558) (44,0.137188) (46,0.138677) (48,0.140043) (50,0.141301) (52,0.142462) (54,0.143537) (56,0.144536) (58,0.145467) (60,0.146335) (62,0.147148)};
					\addplot[color=gray, mark=triangle*, thick, mark options={scale=.8}] coordinates  {
					(6, 0.5)  (8,0.333333) (10,0.25) (12,0.2) (14,0.166667) (16,0.142857) (18,0.125) (20,0.111111) (22,0.1) (24,0.0909091) (26,0.0833333) (28,0.0769231) (30,0.0714286) (32,0.0666667) (34, 0.0625) (36,0.0588235) (38,0.0555556) (40,0.0526316) (42,0.05) (44,0.047619) (46,0.0454545) (48,0.0434783) (50,0.0416667) (52,0.04) (54,0.0384615) (56,0.037037) (58,0.0357143) (60,0.0344828) (62,0.0333333)};
			\end{axis}
		\end{tikzpicture}
		\caption{Entanglement ratio $\Gamma_v$ quantifies the ratio of the bipartite and multipartite 
			entanglement  of a  particular subsystem $v$. 
			Distribution of $\Gamma_v$ for the family $\ket{\text{O}^2_{2m}}$ of states with $N= 2m$ parties and $k=2$ excitations, related to the $m$-orthoplexes, is presented. 
			Contribution from distance two vertices is dominating for small systems. 
			For larger systems the contribution from distance-two nodes 
			dominates $\Gamma_v$, 
			while related graph becomes almost complete.
			The entanglement ratio at infinity goes to a nonzero value,
			   $3 - 2\sqrt{2}\approx 0.17$.
		}
		\label{Comparison_polytope}
	\end{figure}
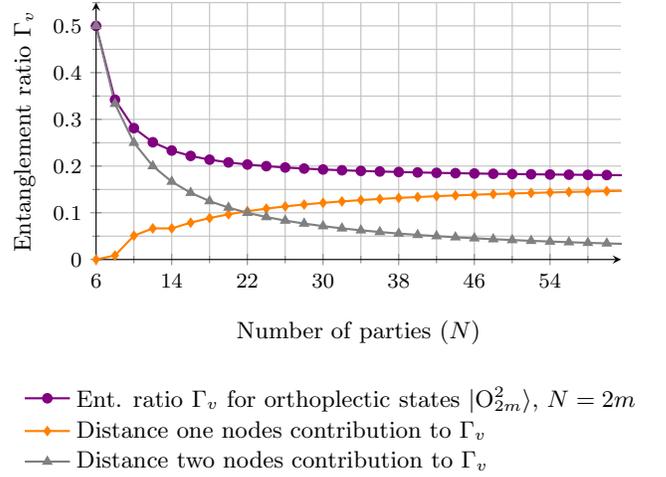

	\subsection{Plane regular tilings}
	
	Other highly symmetric objects, including regular and semi-regular tessellation of the plane, 
	can also be used to generate symmetric quantum states.
	In the regular tiling, the plane is covered by regular polygons, all of the same shape. In semi-regular tiling, tiles of more than one shape are allowed but every corner is identically arranged. 
	Regular tiling provides an edge-transitive graph, while semi-regular tiling
	leads to a vertex-transitive graph. 
	Even though the local properties of such tiling might be various, the global pattern itself can be categorized into 17 wallpaper groups. 
	
	Each such a group consists of two types of generators: shifts and reflections. 
	In general, tessellations take the form of infinite patterns, but one may restrict to the arbitrary regular region by appropriate gluing on the boundary.
	In such a way, we obtain a finite tiling, which might be related to excitation-states under construction provided.  By the property of edge-transitivity the 
	regular tilings correspond to Dicke-like states.  
	The group of symmetry is given by a relevant wallpaper group, restricted to the chosen size of the tiling. 
	Since semi-regular tilings do not do exhibit edge-transitivity, they in general 
	do not correspond to  Dicke-like states. 

	\cref{Tiling} presents the hexagonal tiling of the plane and relevant excitation-state $\ket{H_N}$. 
	For hexagonal tiling, the concurrence in bipartite subsystems takes the following value:
	\small
	\begin{align*}
		C_{v|w}({{H_N}}) &=
		\begin{cases}
			2/N \quad \; \; &\text{for} \quad d(v,w)=2\\
			0 \quad &\text{otherwise} 
		\end{cases} ,
	\end{align*} 
	\normalsize
	for the tiling restricted to $N$ nodes (minimal size of a cut is $3\times 3$).  
	The entanglement ratio $\Gamma_v$ for a given node $v$ reads,
	\[
	\Gamma_v ({H_N}) =\dfrac{4}{3}\dfrac{1}{N-2}.
	\]
	Even though the value of the two-partite concurrence $C_{v|w}$ is the same as in the cyclic case, the parameter $\Gamma_v$ takes a larger value. 
	Indeed, each node in the hexagonal case has 6 distance-two vertices, which causes the increase of the entanglement ratio $\Gamma_v$.
	
	\begin{figure} 
		\begin{tikzpicture}[scale=0.75,pics/hexi/.style={code={\stepcounter{hexi}
					\node[draw,regular polygon,regular polygon sides=6,minimum width=1.5cm]
					(hexi-\number\value{hexi}) {};
					\foreach \Corner in {1,...,6}
					{\ifodd\Corner
						\draw[fill=black] (hexi-\number\value{hexi}.corner \Corner) circle[radius=1.5pt];
						\else
						\draw[fill=white] (hexi-\number\value{hexi}.corner \Corner) circle[radius=1.5pt];
						\fi}
			}},bullet/.style={circle,fill,inner sep=0.5pt}]
			\clip (0,1) rectangle (9.8,6.5);
			\definecolor{greeo}{RGB}{91,173,69};
			\draw[thick,greeo] (5,4.332) edge (5,2.6) ;
			\draw[thick,greeo] (5,4.332) edge (5,6.064) ;
			\draw[thick,greeo] (5,4.332) edge (6.5,5.2) ;
			\draw[thick,greeo] (5,4.332) edge (6.5,3.466) ;
			\draw[thick,greeo] (5,4.332) edge (3.5,5.2) ;
			\draw[thick,greeo] (5,4.332) edge (3.5,3.466) ;
			\path foreach \X in {1,...,6} {
				foreach \Y [evaluate=\Y as \Z using {int(mod(33-\Y-\X,5)+1)}] in {1,...,4} { \ifodd\X 
					({\X*(1+cos(60))},{\Y*(2*sin(60))})
					\else
					({\X*(1+cos(60))},{\Y*(2*sin(60))-sin(60)})
					\fi pic{hexi=\Z}}};
			\draw[red] (1.7,0) 
			-- (1.7,6.4) node[bullet]{}
			-- (7.7,6.4) node[bullet]{}
			-- (7.7,0);
			\tikzstyle{every node}=[circle,inner sep=4pt,draw=gray!95!white,fill=greeo!10!white,thick,anchor=base];
			\node(1) at (5,4.332) {\scriptsize 1};
		\end{tikzpicture}
		\caption{Hexagonal tiling of a plane. One may choose a rectangular region 
			(indicated by the red line) 
			and glue the boundary left-right and up-down to form a torus. In this specific choice of the rectangular cut, there are 4 nodes in each row and 6 in each column. For chosen node $v$, the concurrence on two-parties subsystem $C_{1 v}$ is vanishing except for distance 2 parties (indicated by green lines). In such a case $C_{1 v}=1/12$, 
			and the entanglement ratio 
			$\Gamma_v\approx 0.06$. For the similar tiling by triangles (restriction to 25 nodes), the entanglement ratio  takes value $\Gamma_v \approx 0.072$. Even though the triangular network is twice as dense as the hexagonal, for such a network much more entanglement is concentrated in the bipartite
			subsystems of distance-two.
		}
		\label{Tiling}
	\end{figure}
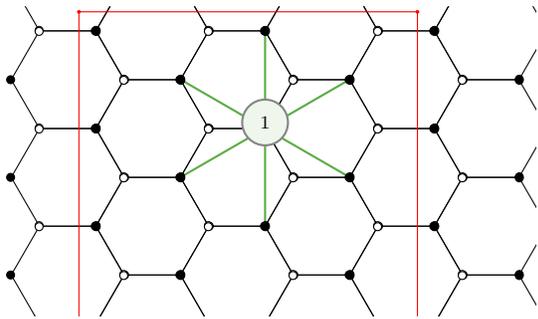

	\subsection{General local-network}
	
	Approach presented allows one to construct multipartite quantum states
	associated to local networks. 
	We assume that certain properties of the network are the same at each node. 
	In particular, all vertices have the same degree $d$. 
	Our analysis is precise for a sufficiently large number of nodes $N$. 
	By locality, we understand
	that the number of edges $|E|$ scales as the number of vertices,
	$|E| =\mathcal{O} (N)$.
	This assumption leads to vanishing of the concurrence for all 
	bipartite reduced systems corresponding to two vertices at distance one. 
	For simplicity we restrict our attention to the case of graph-networks, 
	i.e. edges connect always two vertices.
	
	For a given node labeled as $v$, define the following function:
	\begin{equation}
		\label{hmm}
		\gamma_v := \sum_{v : d(v,v')=2} n_{vv'}^2,
	\end{equation}
	where the sum runs over all distance-two vertices, 
	while the joint neighborhood $n_{vv'}$ (defined in \cref{cccc}) simplifies here to the number of paths of length $2$ from $v$ to $v'$. 
	Observe that by regularity, the number of paths connecting distance-two vertices takes always the same value:
	\[
	\sum_{v' : d(v,v')=2} n_{vv'} =d(d-1).
	\]
	The possible displacement of distance-two nodes responds for the parameter $\Gamma_v$ for the given network. 
	Indeed, from \cref{A,B} follows that
	\[
	\Gamma_v  ({\rm network}) = \dfrac{2}{N-2} \dfrac{1}{d^2} \gamma_v ,
	\]
	and the only non-global parameter above is $\gamma_v $. 
	Therefore, optimisation of the entanglement ratio $\Gamma_v$ is equivalent to optimisation of $\gamma_v$. 
	In order to maximize \cref{hmm}, distance-two paths from node $v$ should lead to the minimal number of nodes. 
	Conversely, in order to minimize (\ref{hmm}), 
	the distance-two nodes should be connected by a single path only. 
	For instance, \cref{PlatonicStates} presents the network, in which all
	vertices of  distance-two  are connected by two paths, while in the network presented on \cref{Tiling} all such vertices are connected by a unique path.

	\begin{corollary}
		For a regular local network, i.e $|E| =\mathcal{O} (N)$ 
		the entanglement ratio $\Gamma_v$ of relevant excitation-states have the following lower bound:
		\[
		\Gamma_v  ({\rm local \ network})    \geq \dfrac{2}{N-2} \dfrac{d-1}{d}.
		\]
		Moreover, it has the following asymptotic complexity $\Gamma_v =\mathcal{O}(1/N)$.
	\end{corollary}
	
	\noindent
	Observe that the bound above is tight for cyclic states $\ket{C_N}$ and for states $\ket{H_N}$ related to the hexagonal tiling. Hypercubic states $\ket{\text{B}^{2}_{2^m}}$ achieve exactly twice this bound. 
	
	Remarkably, the entanglement ratio $\Gamma_v$ is strongly related to the curvature of the network. 
	For a flat-like network one can expect that the number of distance-two nodes will be twice larger than the number of distance-one nodes. 
	This is the case of the hexagonal tiling, see \cref{Tiling}, and all
	the regular tessellations but not  all semi-regular  tilings.
	By assumption of flatness (understood as above), and that the 
	same number of paths connecting all distance-two vertices, equal to  $(d-1)/2$, 
	the entanglement ratio $\Gamma_v$ can be estimated as follows,
	\[
	\Gamma_v ({\rm flat \ local \ network}) \; \simeq \; \dfrac{1}{N-2} \dfrac{(d-1)^2}{d}.
	\]
	Notice, that the above estimation is not precise
	if  the number of paths connecting distance-two vertices is not uniformly distributed. 
	Observe that hypercubic states $\ket{\text{B}^{2}_{2^m}}$ are not relevant to the flat networks, and indeed, their entanglement ratio scales differently
	$$\Gamma_v 
	(\text{B}^{2}_{2^m})
	= \frac{4}{2^m-2} \dfrac{(m-1)}{m}=\frac{4}{N-2} \dfrac{(d-1)}{d},$$
	with respect to the local degree $d$ parameter.

	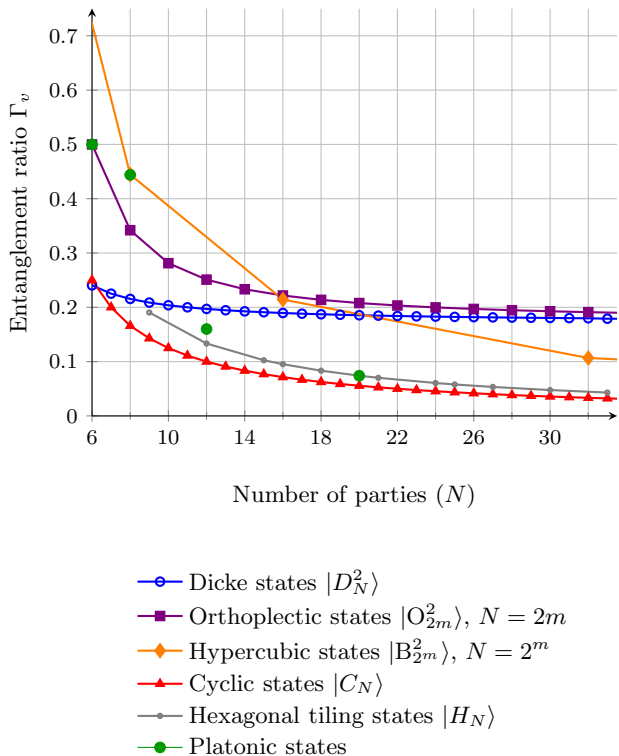
\begin{figure}[h!]
		\begin{tikzpicture}
			\begin{axis}[  
				width=0.99\linewidth,
				height=7cm,
				axis lines=left,
				grid=both,
				legend style={draw=none, font=\small},
				legend cell align=left,
				legend style={at={(0.5,-0.35)}, anchor=north},
				label style={font=\small},
				ticklabel style={font=\footnotesize},
				xlabel style={at={(0.5, -0.05)}},
				ylabel style={at={(0.04, 0.5)}},
				ymin=0,
				xmin=6,
				ymax=0.75,
				xmax=33.5,
				xtick={6,10,...,34},
				minor x tick num = 1,
				ytick={0, 0.1,...,.72},
				xlabel=Number of parties ($N$),
				ylabel=Entanglement ratio   $\Gamma_v$,
				legend entries={{Dicke states $\ket{D^2_N}$},{Orthoplectic states $\ket{\text{O}^2_{2m}}$, $N=2m$} , {Hypercubic states $\ket{\text{B}^2_{2^m}}$, $N=2^m$},{Cyclic states $\ket{C_N}$},{Hexagonal tiling states $\ket{H_N}$}, {Platonic states}},
				]
				\addplot[color=blue, mark=o, thick, mark options={scale=0.8}] coordinates {
					(3,1) (4, 0.333333 )(5, 0.267949)  (6, 0.240408) (7,0.225148) (8,0.215438) (9,0.208712) (10,0.203777) (11, 0.2) (12,0.197017)(13,0.194601) (14,0.192604)(15,0.190925) (16,0.189495) (17,0.188262) (18,0.187188)(19, 0.186244)(20,0.185407)(21, 0.184661) (22,0.183991) (23,0.183386) (24,0.182838) (25,0.182338) (26,0.18188)(27,0.18146) (28, 0.181073) (29,0.180715) (30,0.180383) (31,0.180074)(32, 0.179787) (33, 0.1789787) (34, 0.1789787)};
					\addplot[color=violet, mark=square*, thick,  mark options={scale=.8}] coordinates  {
					(6, 0.5)  (8,0.341977) (10,0.28125) (12,0.251) (14,0.233264) (16,0.221744) (18,0.213718) (20,0.207833) (22,0.203348) (24,0.199823) (26,0.196986) (28,.194655) (30,.192708) (32,.191058) (34, .189643)};
				\addplot[color=orange, mark=diamond*, thick, mark options={scale=1.2}] 
				coordinates  {
					(4,1.) (8,0.444444) (16,0.214286) (32,0.106667) (64,0.0537634)};
					\addplot[color=red,  mark=triangle*,thick, mark options={scale=.8}] coordinates  {
				(3, 1) (4, 1) (5, 0.3333) (6, 0.25) (7,0.2) (8,0.166) (9,0.14285) (10,0.125) (11, 0.1111111) (12, 0.1) (13, 0.090909090) (14, 0.08333) (15, 
  0.0769230769) (16, 0.0714285) (17, 0.06666666) (18, 0.0625) (19, 
  0.05882352) (20, 0.05555555555) (21, 0.0526315) (22, 0.05) (23, 
  0.04761904) (24, 0.04545454) (25, 0.04347826) (26, 0.0416666) (27, 
  0.04) (28, 0.0384) (29, 0.03703) (30, 0.035714) (31, 
  0.034482758) (32, 0.03333333) (33, 0.032258) (34, 0.03125)};
				\addplot[color=gray,  mark=o,thick, mark options={scale=.4}] coordinates  {
					(9, 0.19047) (12, 0.13333333) (15, 0.1025641) (16, 0.09523) (18, 
  0.08333333) (20, 0.0740) (21, 0.0701754) (24, 0.06060606) (25, 
  0.0579) (27, 0.05333333) (30, 0.0476190) (33,0.04301)};
				\addplot[color=black!40!green, mark=*, mark size=4pt,  mark options={scale=0.5}] 	
				coordinates {(6,0.5)  };
				\addplot[color=black!40!green, mark=*, mark size=4pt,  mark options={scale=0.5}] coordinates {
					(8,0.444)  };
				\addplot[color=black!40!green, mark=*, mark size=4pt,  mark options={scale=0.5}] coordinates {
				(12,0.16)  };
				\addplot[color=black!40!green, mark=*, mark size=4pt,  mark options={scale=0.5}] 
				coordinates {
				(20,0.074)  };
			\end{axis}
		\end{tikzpicture}
		\caption{Distribution of entanglement for Dicke-like states with $k=2$ excitations is compared. 
			The entanglement ratio $\Gamma_v$ quantifies the ratio of the bipartite and multipartie 
			entanglement  of a  particular subsystem $v$.
			The distribution of entanglement varies for states relevant to local and global networks, i.e. graph with asymptotically different densities. 
			For states related to dense graphs, $\ket{D^2_N}$ and $\ket{\text{O}^2_{2m}}$, the ratio $\Gamma_v$ tends to a nonzero value $3 - 2\sqrt{2}\approx 0.17$. 
			For families of states in which the local degree does not scale with the size of a graph ($\ket{\text{B}^2_{2^m}}$, $\ket{C_N}$, $\ket{H_N}$), the complexity of $\Gamma_v$ is $\mathcal{O} (1/N)$.
	Exact scaling factor for local networks relies on a local degree $d$ in a way depending on the curvature of the network. 
	For states related to flat ($\ket{C_N}$, $\ket{H_N}$) and spherical (Platonic states) networks one has, $\Gamma_v \sim d/N$, while for hyperbolic case, 
	$\ket{\text{B}^2_{2^m}}$), the ratio scales as $\Gamma_v \sim 1/N$. 
			For flat networks, the exact scaling parameter does not exceed unity ($\Gamma_v \sim  d/2N$ for cyclic states $\ket{C_N}$, and $\Gamma_v \sim 4 d/9N$ for states related to a hexagonal tiling), and exceeds unity for a spherical networks (Platonic states).}
		\label{Comparison}
	\end{figure}

	\section{Phase transitions}
	\label{Phase transitions}
	
	As discussed before, the entanglement ratio $\Gamma_v$ measures the amount of entanglement not present in the bipartite entanglement. 
	Therefore, it indicates the robustness of an entangled state. 
	We investigate the asymptotic behavior of the parameter $\Gamma_v$ of an excitation-state
	$\ket{G}$ with the change of the local degree $d$ of the relevant hypergraph.
	
	For simplicity consider Dicke states with two excitations, $\ket{D_N^2}$. 
	According to exact calculations, the asymptotic entanglement ratio $\Gamma_v$ is 
	\[
	\Gamma_v  ({D_N^2})  \; \sim \; (3- 2 \sqrt{2}).
	\]
	Recall that for local networks one has $\Gamma_v \gtrsim d/N$. 
	Therefore the complexity of entanglement ratio $\Gamma_v$ for Dicke states and a general local-network is not the same, $\mathcal{O} (N)$ and $\mathcal{O} (1/N)$ respectively. 
	
	This phenomenon is related to the fact that for regular networks, the concurrence 
	$C_{v|w}$  between  two parties
	may take two distinct non-vanishing values according to the distance $d(v,w)$, see \cref{A}. 
	While non-vanishing concurrence is always present for distance-two nodes, for neighboring nodes it takes positive value, iff $n_{vw}^2\geq (d(N-4) +2)/2$. 
	Since the number of common neighbours $n_{vw}$ is bounded by the degree $d$, it might occur only if $d\gtrapprox N/2$. 
	Notice that the regular graph satisfying this inequality is highly connected since at least half of all possible edges are present. 
	In such a highly connected graph, the vertices are of distance-one or two. 
	Therefore, the joint neighborhood might be estimated by $n_{vw} \sim d^2/N$ (under the assumption of random distribution of edges). 
	Hence, by \cref{A}, for a graph with vertices of degree:
	\[
	d =\dfrac{1}{\sqrt[3]{2}} \;N  \simeq 0.794 \; N,
	\]
	we observe a first phase transition, in which the entanglement between distance-one parties starts
	to appear. 
	At this moment, however, the entanglement between such parties is negligible in comparison to the entanglement between distance-two vertices. 
	
	The second phase transitions refer to the threshold in which entanglement between distance-one vertices is protruding over the entanglement present between vertices of distance two. 
	We quantify this intuition by comparing the contribution of those two distinct types of two-partite concurrence on the right side of \cref{eq5}. 
	Regarding the previous discussion, this phenomenon takes place for a graph with vertices of degree:
	\[
	d\Big( n_{vw} -\sqrt{dN/2} \Big)^2 \approx (N-d) n_{vw}^2 .
	\]
	Taking into account that $n_{vw} \sim d^2/N$ for such strongly connected graphs, the second phase transition takes place for
	\[
	d \approx 0.973 \; N ,
	\]
	which was computed numerically.


	\begin{corollary}
		Consider asymptotic behaviour of an excitation-state $\ket{G}$ corresponding to a graph 
		of $N$ vertices with respect to the average degree of each vertex equal to $d$. 
		Two phase transitions occur:
		The first one, for the average vertex degree $d=0.794 N$, refers to the threshold in which the entanglement between distance-one parties starts to appear. 
		The second one, for the average vertex degree $d = 0.973 N$, refers to the threshold in which entanglement between distance-one vertices is protruding over the entanglement present 
		between vertices of distance two. 
	\end{corollary}
	
	Note that there is no substantial difference between local networks 
	with average degree $d$ of each vertex being  a finite number
	and global networks with the degree $d$ depending on $N$,
	until the limit $d\simeq0.794 N$  is achieved. 
	Indeed, in both cases, the entanglement ratio $\Gamma_v \sim d/N$, and the entanglement is present only between distance-two nodes. 
	Once this limit is exceeded, the bipartite entanglement between distance one nodes starts to appear. 
	Its amount is, however, negligible in comparison to the entanglement between distance-two vertices. 
	
	Entanglement between distance-one nodes prevails once the threshold $d \approx 0.973 N $ is exceeded. 
	Since the entanglement between distance-one nodes is significantly weaker than the entanglement corresponding to nodes at distance-two, 
	this leads to a global consequence in entanglement distribution. 
	Indeed, the entanglement ratio $\Gamma_v$ (adequate to amount of entanglement in bipartite systems in comparison to the amount of entanglement shared in multipartite way) changes its scaling from $\Gamma_1 \sim d/N$ to $\Gamma_1 \sim (3- 2 \sqrt{2})$. 
	Therefore, the resulting states share much less bipartite entanglement. 
	Hence the global amount of entanglement remains on the same level, states become more robust and persistent, which is perfectly observed at the limit case of the Dicke states.
	
	The aforementioned transitions are clearly visible for $\ket{\text{O}^2_{2m}}$ family of states, related to the $m$-orthoplexes, see \cref{Comparison_polytope}. 
	Indeed, for a small number of particles,
	 $N=2m$, the entanglement is mainly present between distance-two vertices. 
	Already for $N=8$ particles system $\ket{\text{O}^2_{8}}$, modest entanglement between distance-one nodes can be detected. 
	Contribution from distance-one vertices becomes dominating for larger systems,  staring from systems with $N=22$ particles. 
	This reflects the fact, that for a large number $N$, the related graph becomes almost complete.
	
	\section{Quantum circuits} 
	\label{Circuit}
	
	We present below a quantum circuit efficiently transforming a fully separable tensor product state into the excitations-state, proposed in this work. 
	Construction of this circuit was inspired by similar circuits for Dicke states
	developed in \cite{DickeCircuits,brtschi2019deterministic}. However, it refers explicitly to the graph structure and is more intuitive in that sense. 
	As we shall see, our scheme uses between $\sim 4 |V|$ and $\sim 10 |E|$ CNOT gates, depending on the structure of a graph. 
	For clarity of demonstration, we present the construction of a unitary circuit transforming excitations-state into the separable state. 
	The provided construction can be extended in a simple manner to the case of hypergraphs. 
	
	We construct a quantum circuit iteratively relative to the set of vertices. 
	For a vertex $v \in V$, we apply $d_v-1$ three-qubit unitary operations and one two-qubit operation. 
	This can be represented graphically by deleting edges adjacent to the consecutive vertices. 
	Although any three-qubit unitary operation can be simulated by at most 21 CNOT gates \cite{1629135}, 
	considered operations can be simulated by $6$ or $10$ two-qubit gates only, according to  \cite{10.5555/2011517.2011522}. 
	We present the exact form of unitary operations below and, in parallel, its realizations by accessible 2-qubit gates, where we assume standard implementation of Toffoli gate with 6 CNOTs. 
	The exact form of quantum circuit depends on the order of deleted vertices and adjacent edges, the procedure itself, however, might be performed for any such order. 
	
	\textbf{First step}. Choose a vertex $v$, and take all adjust adjacent vertices $v_1 ,\ldots , v_d$ where $d$ is the degree of $v$. 
	Without loss of generality, suppose that each $v_i$ is related to the $i$th particle, while $v$ is related to $d+1$ particle. 
	We shall consecutively apply the following three-qubit gates $U_{ \{i, d, d+1 \}}^{(1)}$ on parties $i,d,d+1$:
	\begin{align}
		\label{al0}
		\ket{101} + \sqrt{i-1} \ket{011} &\longmapsto \sqrt{i} \ket{011} \\
		\ket{110} &\longmapsto \ket{110} \nonumber \\
		\ket{100} &\longmapsto \ket{100} \nonumber \\
		\ket{010} &\longmapsto \ket{010} \nonumber \\
		\ket{001} &\longmapsto \ket{001} \nonumber \\
		\ket{000} &\longmapsto \ket{000} \nonumber ,
	\end{align}
	for $i=1, \ldots,d-2$. 
	Applying $U_{ \{i,d, d+1 \}}^{(1)}$ operation is relevant to the graphical operation of deleting an edge $e= \{i,d+1 \}$. 
	It may be simulated by $10$ CNOT gates, as indicated in \cref{circuit}. 
	
	Secondly, we apply the following three-qubit operator $U_{ \{d-1, d, d+1 \}}^{(2)}$ on parties $d-1,d,d+1$
	\begin{align}
		\label{al1}
		\ket{101} + \sqrt{d-2} \ket{001} &\longmapsto \sqrt{d-1} \ket{001} \\
		\ket{011} &\longmapsto \ket{011} \nonumber \\
		\ket{100} &\longmapsto \ket{100} \nonumber \\
		\ket{010} &\longmapsto \ket{010} \nonumber \\
		\ket{000} &\longmapsto \ket{000} \nonumber ,
	\end{align}
	which is relevant to the graphical operation of deleting an edge $e= \{d-1,d+1 \}$. 
	Observe that \cref{al1} is defined only on 5-dimensional subspace, and is arbitrary on the remaining part of Hilbert space. 
	This leaves room for optimization of simulation of $U_{ \{d-1,d, d+1 \}}^{(2)}$. 
	Indeed, it may be simulated only by $6$ CNOT gates or a single Toffoli gate, see \cref{circuit}. 
	For comparison, for simulation of \cref{al0} already $10$ CNOT gates were needed. 
	
	Finally we apply the following two-qubit operator $U_{ \{d, d+1 \}}^{(3)}$ on parties $d, d+1$:
	\begin{align}
		\label{al2}
		\ket{11} + \sqrt{d-1} \ket{01} &\longmapsto \sqrt{d} \ket{01} \\
		\ket{00} &\longmapsto \ket{00} \nonumber \\
		\ket{10} &\longmapsto \ket{10}  \nonumber  
	\end{align}
	which is relevant to deleting the only remaining edge: $e= \{d,d+1 \}$. 
	Those three operations are presented on \cref{circuit}. 
	There is a simple logic behind those operations. 
	Consecutively, we combine all terms having excitations on position $d+1$ into a single term with an excitation on this position.
	Observe that after applying these 
	operations, the state takes the form of the following superposition:
	\begin{equation*}
		\dfrac{1}{\sqrt{|E|}}
		\Big( \sqrt{d} \ket{1}_{v} \otimes \ket{0\ldots 0}+  \ket{0}_{v} \otimes  
		\sum_{e \in E  \setminus v} \ket{\psi_e}
		\Big)  ,
	\end{equation*}
	where $E  \setminus v$ denotes a set of edges which do not contain vertex $v$. 

	\textbf{Second and the next steps}. Consider now the graph $G' =(V \setminus v, E\setminus v)$. Notice that the degrees of vertices initially adjacent to the vertex $v$ have changed. 
	We faithfully repeat the procedure from the first step for an arbitrary vertex $w$ from the graph $G'$. 
	Observe that after applying the aforementioned gates, the state takes the following form:
	\begin{align*}
		\dfrac{1}{\sqrt{|E|}}
		\Big( &\sqrt{d_v} \ket{1}_{v} \otimes \ket{0\ldots 0}_{v^c}+ \\  
		& \sqrt{d_w'} \ket{1}_{w} \otimes \ket{0\ldots 0}_{w^c}+  
		\ket{00}_{vw} \otimes  
		\sum_{e \in E  \setminus v,w} \ket{\psi_e}
		\Big)  ,
	\end{align*}
	where $d_{v} $ is the degree of the vertex $v$, while $d_{w}'$ is the degree of the vertex $w$ of the reduced graph, after deleting the vertex $v$, 
	(the degree of $w$ can decrease by $1$). 
	We repeat this procedure further, for $G'' =(V \setminus v,w; E\setminus v,w)$, until we delete $N-1$ vertices, which fully separates the initial graph.
	
	\textbf{Final step}. After applying presented procedure iteratively $N-1$ times, the state takes the following form:
	\begin{equation*}
		\dfrac{1}{\sqrt{|E|}} \sum_{v\in V} \sqrt{d_{v}'} \ket{1}_{v} \otimes \ket{0\ldots 0}_{v^c} ,
	\end{equation*}
	where $d_v'$ denotes the degree of the vertex $v$ of the graph reduced 
	 according to 
	the presented procedure. 
	Note that the form of above state is similar to the state $\ket{\text{W}_N}$,
	as it is equivalent to it with respect to the transformations from the SLOCC class. 
	The separation of this state is a well-known procedure \cite{Wgenerating}, and might be obtained by performing two-qubit gates $U_{1i}^{(4)}$:
	\begin{align*}
		\scriptscriptstyle
		\sqrt{d_{v_1}'+ \ldots +d_{v_{i-1}}' } 
		\textstyle
		\ket{10} +
		\scriptscriptstyle
		\sqrt{d_{v_i}'} 
		\textstyle
		\ket{01} &\longmapsto 
		\scriptscriptstyle
		\sqrt{d_{v_1}'+ \ldots + d_{v_{i}}' } 
		\textstyle
		\ket{10} \\
		\ket{00} &\longmapsto \ket{00} \nonumber 
	\end{align*}
	on particles $1$ and $i$. 
	This gate may be simulated only by $3$ CNOT gates, see \cref{circuit}.
	
	\textbf{Cost estimation}. We start with a connected graph with $|V|$ vertices and $|E|$ edges, where the degrees of each vertex are $d_1,\,d_2,\,\hdots < |V|$. In order to delete a given vertex of degree $d$ we need $1 + 6\theta(d-1) + 10\theta(d-2)\cdot(d-2)$, where $\theta(\cdot)$ is the Heaviside theta function. By considering reverse engineering, the spanning tree of the graph costs $|V| - 1$ CNOTs. In the best case the spanning tree will be a star graph with degrees $\{|V| - 1, 1, \hdots, 1\}$. In the next step we pair up the degree-1 vertices and connect them with $\frac{|V|-1}{2}$ edges.
	 Depending on oddness or evenness we are left with an additional edge, which takes up the circuit cost to $6\lceil\frac{|V|-1}{2}\rceil$ CNOTs and leaves all the vertices of degree at least 2. The remaining edges cost in total $10\left(|E| - \lceil\frac{|V|-1}{2}\rceil - |V| + 1\right)\approx10|E| - 5(|V|-1)$ CNOTs. The final addition are the $U^{(4)}_{ij}$ operations which cost at most $3(|V| - 1)$ CNOTs. Thus we arrive at three different regimes:
	\begin{enumerate}
		\item $|E|\sim|V|$, for which the main cost comes from the spanning tree and can be estimated as $4(|V|-1)$ CNOTs.
		\item $|E|\sim\frac{3}{2}|V|$, for which the cost is based based on vertices of degree degree 1 and 2, and is given as $7(|V| - 1)$
		\item $|E|\gg\frac{3}{2}|V|$, for which the cost is dominated by the degree greater than two and is estimated as $10|E| + 2|V| - 2$ CNOTs.
	\end{enumerate}

	
	
	\begin{figure*}[t]
		\begin{tikzpicture}
			\definecolor{bluegray}{rgb}{0.4, 0.6, 0.8}
			
			\begin{scope}[scale=.4]
				\clip (-0.5,0) rectangle (15,13);
				\tikzstyle{every node}=[circle,draw]
				\node(a) at (1,5) {1};
				\node(b) at (7,10) {2};
				\node[fill=bluegray!50](d) at (6,5) {4};
				\node(e) at (12,7) {};
				\node(f) at (11,4) {};
				\node(c) at (8,2) {3};
				\node(c1) at (7,-2) {};
				\node(a1) at (-1,2) {};
				\node(a2) at (-1,10) {};
				\node(b1) at (6,12) {};
				\node(b2) at (7,18) {};
				\node(e1) at (16,9) {};
				\node(f1) at (16,1) {};

				\draw[ultra thick,gray] (a)--(d);
				\draw[ultra thick,gray] (b)--(d);
				\draw[ultra thick,gray] (c)--(d);
				
				\draw[ultra thick,gray] (c)--(f);
				\draw[ultra thick,gray] (b)--(e);
				\draw[ultra thick,gray] (f)--(e);
				
				\draw[ultra thick,gray] (e)--(e1);
				\draw[ultra thick,gray] (f)--(f1);
				\draw[ultra thick,gray] (e)--(e1);
				\draw[ultra thick,gray] (c)--(c1);
				\draw[ultra thick,gray] (a)--(a1);
				\draw[ultra thick,gray] (a)--(a2);
				\draw[ultra thick,gray] (b)--(b1);
				\draw[ultra thick,gray] (b1)--(a2);
				\draw[ultra thick,gray] (b1)--(b2);
				
				\draw [cyan,thick] plot [smooth cycle, tension=1] coordinates { (0,5) (6,6.5)  (10,1.6) (8,0.5) (4,1.5) };
				\draw [red, thick] plot [smooth cycle, tension=0.7] coordinates { (5,5) (6,6)  (8,3) (9,2) (8,1) (6,4)};
				\draw [blue,thick] plot [smooth cycle, tension=1] coordinates { (6, 10) (8,10.5) (8,5)  (9.5,1) (6,2) (3.5,5)};
			\end{scope}
			
			\definecolor{ceruleanblue}{rgb}{0.16, 0.32, 0.75}
			\begin{scope}[shift={(7.5,0)},scale=.8]
				\node[left] at (0,5) {$\ket{2}$};
				\node[left] at (0,4) {$\ket{3}$};
				\node[left] at (0,3) {$\ket{4}$};
				\node[left] at (0,2) {$\ket{1}$};
				\node[left] at (0,1) {$\ket{\text{others}}$};
				
				\draw[thick,-] (0,5)--(3,5);
				\draw[thick,-] (0,4)--(1,4);
				\draw[ultra thick,-,bluegray] (0,3)--(1,3);
				\draw[thick,-] (0,2)--(1,2);
				\draw[thick,-] (0,1)--(8,1);
				
				\draw[thick,cyan] (1,1.5) rectangle (2,4.5) node [pos=.5]{$U^{(1)}_{1,4}$};
				\draw[thick,blue] (3,2.5) rectangle (4,5.5) node [pos=.5]{$U^{(2)}_{2,4}$};
				\draw[thick,red] (5,2.5) rectangle (6,4.5) node [pos=.5]{$U^{(3)}_{3,4}$};
				
				\draw[ultra thick,-,dashed,rounded corners, gray] (8, 0.5)--(8,4.5)--node[midway,below]    
				{\large next steps}(11,4.5)--(11,0.5);

				\draw[thick,-] (2,4)--(3,4);
				\draw[ultra thick,-,bluegray] (2,3)--(3,3);
				\draw[thick,-] (2,2)--(8,2);
				
				\draw[thick,-] (4,5)--(6.5,5);
				\draw[thick,-] (4,4)--(5,4);
				\draw[ultra thick,-,bluegray] (4,3)--(5,3);
				
				\draw[thick,-] (6,4)--(6.5,4);
				\draw[ultra thick,-,bluegray] (6,3)--(6.5,3);

				\draw[thick,-] (6.5,5)--(7.5,4);
				\draw[thick,-] (6.5,4)--(7.5,3);
				\draw[line width=2mm,-,white] (6.5,3)--(7.5,5);
				\draw[ultra thick,-,bluegray] (6.5,3)--(7.5,5);

				\draw[ultra thick,-,bluegray] (7.5,5)--(12,5);
				\draw[thick,-] (7.5,4)--(8,4);
				\draw[thick,-] (7.5,3)--(8,3);
				
				\draw[thick,-] (11,4)--(12,4);
				\draw[thick,-] (11,3)--(12,3);
				\draw[thick,-] (11,2)--(12,2);
				\draw[thick,-] (11,1)--(12,1);
				
				=
				\draw[thick,decorate,decoration={brace,mirror},xshift=0pt,yshift=0pt]
				(0.5,0) -- (6.5,0) node [black,midway,yshift=-0.35cm] {Deleting vertex $4$};
				\draw[thick,decorate,decoration={brace,mirror},xshift=0pt,yshift=0pt]
				(7.5,0) -- (11.5,0) node [black,midway,yshift=-0.35cm] {Deleting rest of vertices};
			\end{scope}

			\begin{scope}[shift={(14,0.5)},scale=.15]
				\clip (-0.5,-1) rectangle (15,16);
				\tikzstyle{every node}=[circle,draw,scale=0.6]
				\node(a) at (1,5) {};
				\node(b) at (7,10) {};
				\node[fill=bluegray!50](d) at (6,5) {4};
				\node(e) at (12,7) {};
				\node(f) at (11,4) {};
				\node(c) at (8,2) {};
				\node(c1) at (7,-2) {};
				\node(a1) at (-1.5,2) {};
				\node(a2) at (-1.5,10) {};
				\node(b1) at (6,12) {};
				\node(b2) at (7,18) {};
				\node(e1) at (16,9) {};
				\node(f1) at (16,1) {};

				\draw[ultra thick,gray] (c)--(f);
				\draw[ultra thick,gray] (b)--(e);
				\draw[ultra thick,gray] (f)--(e);
				
				\draw[ultra thick,gray] (e)--(e1);
				\draw[ultra thick,gray] (f)--(f1);
				\draw[ultra thick,gray] (e)--(e1);
				\draw[ultra thick,gray] (c)--(c1);
				\draw[ultra thick,gray] (a)--(a1);
				\draw[ultra thick,gray] (a)--(a2);
				\draw[ultra thick,gray] (b)--(b1);
				\draw[ultra thick,gray] (b1)--(a2);
				\draw[ultra thick,gray] (b1)--(b2);
				
			\end{scope}
			
			\begin{scope}[shift={(0,-4.3)}]
				
				\node at (3.5,3) {$U^{(1)}_{\{i,j,k\}}$};
				
				\draw[gray] (0,0)--(7,0);
				\draw[gray] (0,1)--(7,1);
				\draw[gray] (0,2)--(7,2);
				
				\draw[thick, black, fill=white] (.5,1) circle (.25);
				\draw[thick, black, fill=black] (.5,2) circle (.1);
				\draw[thick, black] (.5,2)--(.5,0.75);
				\draw[thick, black] (.75,1)--(.25,1);
				
				\draw[thick, black, fill=white] (1.5,2) circle (.25);
				\draw[thick, black, fill=black] (1.5,1) circle (.1);
				\draw[thick, black] (1.5,2.25)--(1.5,1);
				\draw[thick, black] (1.75,2)--(1.25,2);
				
				\draw[thick, black, fill=white] (3.5,2) circle (.25);
				\draw[thick, black, fill=black] (3.5,1) circle (.1);
				\draw[thick, black, fill=black] (3.5,0) circle (.1);
				\draw[thick, black] (3.5,2.25)--(3.5,0);
				\draw[thick, black] (3.75,2)--(3.25,2);
				
				\draw[thick, black, fill=white] (5.5,2) circle (.25);
				\draw[thick, black, fill=black] (5.5,1) circle (.1);
				\draw[thick, black] (5.5,2.25)--(5.5,1);
				\draw[thick, black] (5.75,2)--(5.25,2);
				
				\draw[thick, black, fill=white] (6.5,1) circle (.25);
				\draw[thick, black, fill=black] (6.5,2) circle (.1);
				\draw[thick, black] (6.5,2)--(6.5,0.75);
				\draw[thick, black] (6.75,1)--(6.25,1);

				\draw[thick,black,fill=white] (2.1,1.6) rectangle (2.9,2.4) node [pos=.55]{$P^{\color{white}\dag}$};
				\draw[thick,black,fill=white] (4.1,1.6) rectangle (4.9,2.4) node [pos=.55]{$P^\dag$};
			\end{scope}
			
			\begin{scope}[shift={(8,-4.3)}]
				
				\node at (1.5,3) {$U^{(2)}_{\{i,j,k\}}$};
				
				\draw[gray] (0,0)--(3,0);
				\draw[gray] (0,1)--(3,1);
				\draw[gray] (0,2)--(3,2);
				
				\draw[thick, black, fill=white] (1.5,2) circle (.25);
				\draw[thick, black, fill=black] (1.5,1) circle (.1);
				\draw[thick, black, fill=black] (1.5,0) circle (.1);
				\draw[thick, black] (1.5,2.25)--(1.5,0) ;
				\draw[thick, black] (1.75,2)--(1.25,2);
				
				\draw[thick, black, fill=white] (0.5,1) circle (.25);
				\draw[thick, black] (.5,.75)--(0.5,1.25) ;
				\draw[thick, black] (.75,1)--(.25,1);
				
				\draw[thick, black, fill=white] (2.5,1) circle (.25);
				\draw[thick, black] (2.5,.75)--(2.5,1.25) ;
				\draw[thick, black] (2.75,1)--(2.25,1);
				
				\draw[thick,black,fill=white] (.1,1.6) rectangle (.9,2.4) node [pos=.55]{$P^{\color{white}\dag}$};
				\draw[thick,black,fill=white] (2.1,1.6) rectangle (2.9,2.4) node [pos=.55]{$P^\dag$};
			\end{scope}
			
			\begin{scope}[shift={(12,-4.8)}]
				
				\node at (2.5,3.5) {$U^{(4)}_{\{i,j\}}$};
				
				\draw[gray] (0,2)--(5,2);
				\draw[gray] (0,1)--(5,1);
				
				\draw[thick, black, fill=white] (2.5,2) circle (.25);
				\draw[thick, black, fill=black] (2.5,1) circle (.1);
				\draw[thick, black] (2.5,2.25)--(2.5,1) ;
				\draw[thick, black] (2.75,2)--(2.25,2);
				
				\draw[thick, black, fill=white] (.5,1) circle (.25);
				\draw[thick, black, fill=black] (.5,2) circle (.1);
				\draw[thick, black] (.5,.75)--(.5,2) ;
				\draw[thick, black] (.75,1)--(.25,1);
				
				\draw[thick, black, fill=white] (4.5,1) circle (.25);
				\draw[thick, black, fill=black] (4.5,2) circle (.1);
				\draw[thick, black] (4.5,.75)--(4.5,2) ;
				\draw[thick, black] (4.75,1)--(4.25,1);
				
				\draw[thick,black,fill=white] (1.1,1.6) rectangle (1.9,2.4) node [pos=.55]{$P^{\color{white}\dag}$};
				\draw[thick,black,fill=white] (3.1,1.6) rectangle (3.9,2.4) node [pos=.55]{$P^\dag$};
			\end{scope}
			
		\end{tikzpicture}
		\caption{Procedure sketched for vertex $4$ of degree $d=3$. 
			Firstly, we apply unitary three-qubit operations $U_{ \{1,3,4 \}}^{(1)}$ on parties $1,3,4$, secondly $U_{ \{2,3,4 \}}^{(2)}$ on parties $2,3,4$, and finally two-qubit gate $U_{ \{3,4 \}}^{(3)}$ on parties $3,4$. 
			As a result, we obtain the state separable with respect to the $4$-th particle. 
			In the next step, we 
			work with the graph with deleted all edges adjacent to the $4$th vertex. Presented unitary operations $U_{ \{1,3,4 \}}^{(1)},U_{ \{2,3,4 \}}^{(2)} $, and $U_{ \{3,4 \}}^{(3)}$ might be simulated by the following quantum gates. All three constructions rely on an auxiliary unitary gate $P$, which satisfies $P^\dag \sigma_x P = \sqrt{a - b}\sigma_z + \sqrt{b}\sigma_x$ with $a$ and $b$ set according to the corresponding operation.}
		\label{circuit}
	\end{figure*}
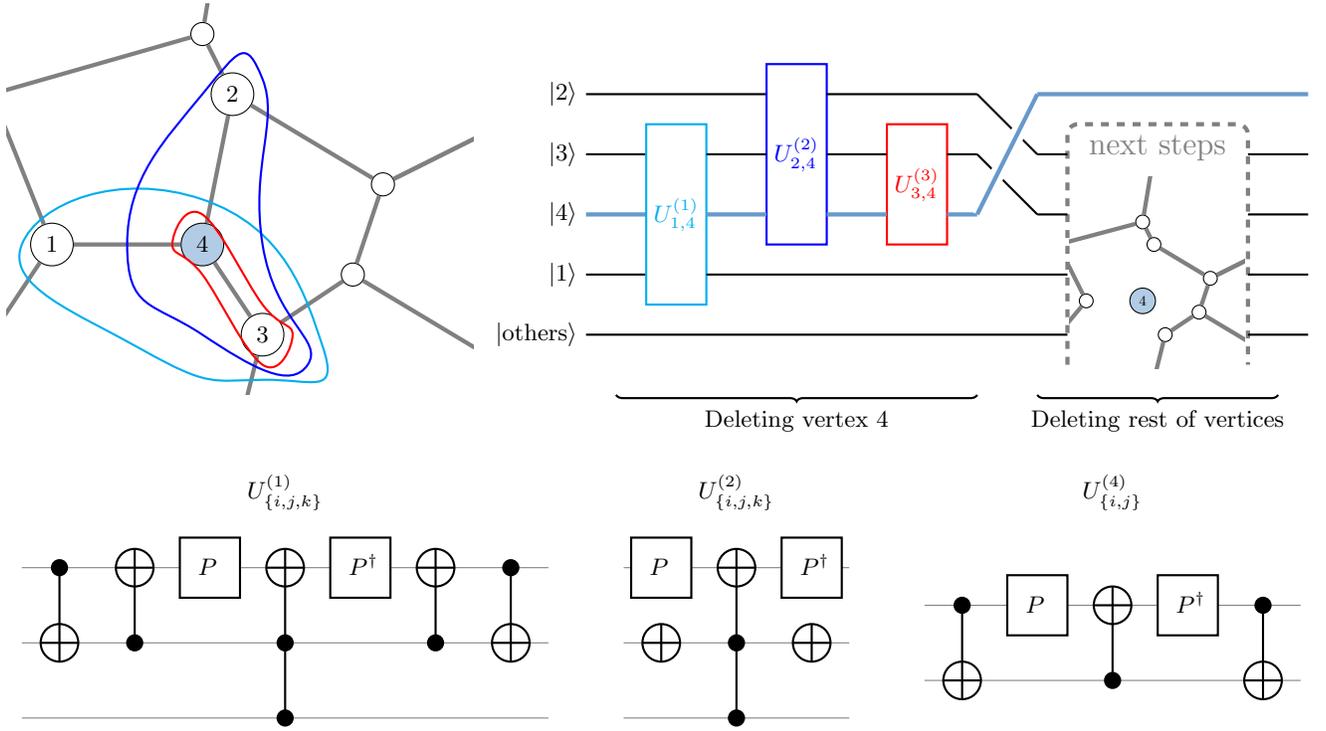
	

	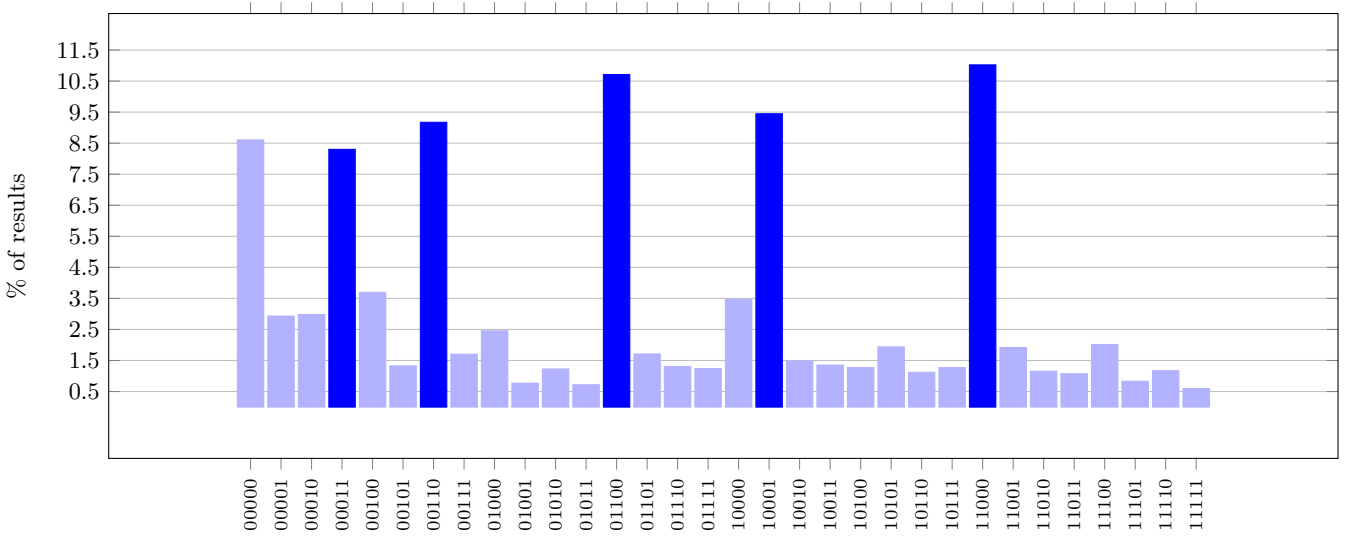
\begin{figure*}[t]
		\centering
		\begin{tikzpicture}
			\begin{axis}[ 
				ybar,
				ymajorgrids = true,
				height=7.5cm,
				width=\linewidth,
				enlargelimits=0.15,
				legend style={at={(0.5,-0.15)},
					anchor=north,legend columns=-1},
				ylabel={\% of results},
				symbolic x coords={00000, 00001, 00010, 00011, 00100, 00101, 00110, 00111, 01000, 01001, 01010, 01011, 01100, 01101, 01110, 01111, 10000, 10001, 10010, 10011, 10100, 10101, 10110, 10111, 11000, 11001, 11010, 11011, 11100, 11101, 11110, 11111},
				ytick={0.5,1.5,...,11.5},
				xtick=data,
				x tick label style={rotate=90,anchor=east, font = \scriptsize},
				]
				\addplot[color=blue!30, fill=blue!30, thin, bar shift = 0pt] coordinates {
					(00000, 8.604) 
					(00001, 2.931) 
					(00010, 2.983) 
					(00011, 0) 
					(00100, 3.696) 
					(00101, 1.327) 
					(00110, 0) 
					(00111, 1.703) 
					(01000, 2.457) 
					(01001, 0.77) 
					(01010, 1.227) 
					(01011, 0.719) 
					(01100, 0) 
					(01101, 1.71) 
					(01110, 1.309) 
					(01111, 1.239) 
					(10000, 3.466) 
					(10001, 0) 
					(10010, 1.483) 
					(10011, 1.353) 
					(10100, 1.274) 
					(10101, 1.937) 
					(10110, 1.12) 
					(10111, 1.271) 
					(11000, 0) 
					(11001, 1.92) 
					(11010, 1.155) 
					(11011, 1.074) 
					(11100, 2.015) 
					(11101, 0.831) 
					(11110, 1.172) 
					(11111, 0.597)};	
				\addplot[color=blue, fill=blue, thin, bar shift = 0pt] coordinates {
					(00011, 8.301) 
					(00110, 9.173) 
					(01100, 10.712) 
					(10001, 9.448) 
					(11000, 11.022)};	
			\end{axis}
		\end{tikzpicture}
		\caption{Mean results acquired over
		 740 000
		 samples acquired on IBM Vigo, Athens and Santiago Computers. Prominent peaks at all cyclic permutations of $\ket{00011}$ are plotted in dark,
		  with an additional significant peak at $\ket{00000}$. 
		 Due to effects of decoherence we observe a background of all
		 remaining states represented by light bars. 
		 }
		\label{cyclic_hist}
	\end{figure*}
	
%
%
%

\vspace{0.8cm}	
\section{5-qubit cyclic state $\ket{C_5}$ implementation}
	\label{5-qubit}
	
	In order to demonstrate the viability of the provided construction we have explicity evaluated the construction for the simplest nontrivial cyclic state $\ket{C_5}$. Initialization of the state is done basically in two steps. The first step involves initialization of excitations and distributing them in order to arrive at the state of the form
	
	\begin{equation*}
		\dfrac{1}{\sqrt{|E|}} \sum_{v\in V} \sqrt{d_{v}'} \ket{1}_{v} \otimes \ket{0\ldots 0}_{v^c}.
	\end{equation*}
	
	In the next step, the distribution of edges between the vertices is effected by applying the unitaries \eqref{al1} and \eqref{al2} in reversed order to finally reconstruct the graph corresponding to $\ket{C_5}$. The overall circuit requires a single operator $U^{(2)}$ and remaining operations are done by using $U^{(3)}$ and $U^{(4)}$, which totals to estimated $22$ CNOT operations, which does not take into account the topology of the quantum computer to be used.
	
	Such a circuit can be realised on the state-of-the-art 5-qubit quantum computers provided by IBM -- linear-topology Santiago and Athens with quantum volumes (QV) of 32 and Vigo with QV of 16 with T topology. In total we used 
more than 740 000 samples over all three computers, which gives distribution with significant values for all expected computational states proceeding from cyclic permutations of $\ket{00011}$ state with the probability $0.487$  of finding the system in one of them, see \cref{cyclic_hist}.
	
	In order to further ascertain we decided to compare the results to a model of noise composed from the exponential probability of decay to the base state and a constant probability of flip from 0 to 1. In order to perform the noise analysis we first divide the readouts into signal,
	
	$$
		\mathcal{S} = \left\{\left(i_1,\hdots,i_5\right): \sum_{j=1}^5i_j = 2, \exists_j i_j = i_{j+1} = 1\right\},
	$$
	and noise,
	$\mathcal{N} = \left\{0,1\right\}^{\otimes 5} \backslash \mathcal{S}$. Furthermore, the noise has been divided into the subsets with constant numbers of excitations,
	$
		\mathcal{N}_k = \left\{\left(i_1,\hdots,i_5\right) \in \mathcal{N},\,i_1 + \hdots + i_5 = k\right\}
	$, hence $\mathcal{N} = \mathcal{N}_0 \cup \hdots \cup \mathcal{N}_5$.
	
	This division allows us to introduce mean probability of noise with $n$ excitations,
	
\begin{equation}
		\overline{P}_k = \frac{1}{\left|\mathcal{N}_k\right|} \sum_{i_1,\hdots,i_5\in\mathcal{N}_k} p_{i_1,\hdots,i_5},
\label{pp}
\end{equation}
	where the probabilities are found from empirical data by dividing the number hits for specific state $N_{i_1,\hdots,i_5}$ by the total number of hits,
	
	$$
		p_{i_1,\hdots,i_5} = \frac{N_{i_1,\hdots,i_5}}{N_{tot}}\,.
	$$
	
	Second step was to fit the model with exponential and flip noises superposed, which resulted in a satisfactory fit with total squared distances equal to $\sum_{k=0}^4(\overline{P}_k - \overline{P}_{k \text{,fit}})^2 \approx 0.075$, see \cref{fig:noisefit}.
	
	In conclusion, the experiment shows that implementation of the introduced preparation procedure and subsequent executing it on a real-world quantum machine is not only viable, but follows simple predictions about the nature of errors present in such computers.
	
	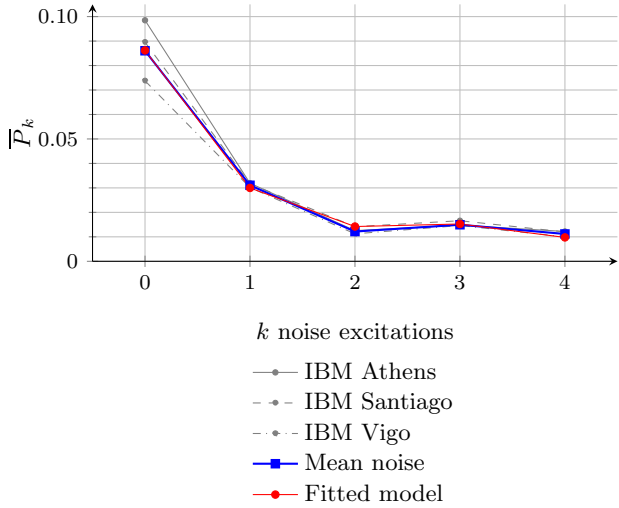
\begin{figure}[h!]
		\begin{tikzpicture}
			\begin{axis}[  
				width=0.99\linewidth,
				height=5cm,
				axis lines=left,
				grid=both,
				legend style={draw=none, font=\small},
				legend cell align=left,
				legend style={at={(0.5,-0.35)}, anchor=north},
				label style={font=\small},
				ticklabel style={font=\footnotesize},
				xlabel style={at={(0.5, -0.05)}},
				ylabel style={at={(0.04, 0.5)}},	
				yticklabel style={
					/pgf/number format/fixed,
					/pgf/number format/precision=2,
					/pgf/number format/fixed zerofill
				},
				ymin=0,
				xmin= -0.5,
				ymax= 0.105,
				xmax= 4.5,
				xtick={0,1,...,4},
				ytick={0,0.05,0.1},
				yticklabel={\ifdim\tick pt=0pt 0 \else\pgfmathprintnumber{\tick}\fi},
				minor y tick num = 4, 
				xlabel= $k$ noise excitations,
				ylabel= $\overline{P}_k$,
				legend entries={{IBM Athens},{IBM Santiago},{IBM Vigo},{Mean noise},{Fitted model}},
				]
				\addplot[color=gray, mark=*, thin, mark options={scale=.5}] coordinates {
					(0, .0985) (1, .0319) (2,.01254) (3,.01476) (4,.01219)};
				\addplot[color=gray, mark=*, thin, mark options={scale=0.5}, dashed] coordinates {
					(0, .08955) (1, .03166) (2,.01414) (3,.01658) (4,.012)};
				\addplot[color=gray, mark=*, thin, mark options={scale=0.5}, dashdotted] coordinates {
					(0, .07376) (1, .03015) (2,.01124) (3,.01462) (4,.01001)};
				\addplot[color=blue, mark=square*, thick, mark options={scale=0.7}] coordinates {
					(0, .08604) (1, .03107) (2,.01216) (3,.01494) (4,.01117)};
				\addplot[color=red, mark=*, thin, mark options={scale=0.7}] coordinates  {
					(0, .08618) (1, .02998) (2, .01421) (3, .01525) (4,.00976)};
			\end{axis}
		\end{tikzpicture}
		\caption{Mean probability of noise (\ref{pp}) for each machine 
		collected according to the number of excitations. Averaged noise (blue) has been taken as a basis for fitting  of the noise model. 
	composed of exponential decay to the ground state of the entire system superposed with a random 0-1 flip with equal probability for all qubits. 
Resulting fit, plotted in red,
 is satisfactory, with a total deviation of $\sum_{k=0}^4(\overline{P}_k - \overline{P}_{k \text{,fit}})^2 \approx 0.075$.}
		\label{fig:noisefit} 
	\end{figure}
	\section{Hamiltonians of excitation-states}
	\label{Hamiltonians}
	
	The Dicke states are ground states of the Dicke Hamiltonian
	with a well--defined number of excitations. The excitation-states introduced above
	can be consider as ground states of  analogous Hamiltonians, 
	with the same property of a well defined excitation number. 
	
	The single-mode Dicke Hamiltonian (known also as Tavis-Cummings model or generalized Jaynes-Cummings model) can be written as \cite{PhysRevE.67.066203,DickeModelRevisited}:
	\[
	H = \omega_0 \; J_z+ 
	\omega \; a^\dagger a +
	\dfrac{\lambda}{\sqrt{N}} \; \underbrace{(a^\dagger +a)(J_{-} + J_{+})}_{\text{interaction}}\; ,
	\]
	where $J_\alpha$ are the collective operators:
	\[
	J_z \equiv \sum_{i=1}^N \sigma_z^{i} , \quad 
	J_\pm \equiv \sum_{i=1}^N \sigma_\pm^{i} . 
	\]
	At zero interaction,  $\lambda =0$, the coupling term in the Dicke Hamiltonian vanishes. 
	Thus its eigenstates have the tensor product form with a factor representing the
	Fock states of the field and the other factor as an eigenstate
	of the collective angular momentum operator $J_z$. 
	The operator $J_z$  has the following eigenvalues: $- N/2, - N/2+1 ,\ldots ,N/2$, which are strongly degenerated. 
	The further partition of its eigenspaces is done by the effective total angular momentum operator $J$, which might be written as:
	\[
	J \equiv J_{+} J_{-} + J_z (J_z+ \mathbb{I} ),
	\]
	in terms of the raising and lowering operators. 
	The Dicke states $\ket{j,m}_N$ are eigenstates of both operators $J_z$ and $J$:
	\begin{align*}
		J_z \ket{j,m}_N &= m \ket{j,m}_N,\\
		J \ket{j,m}_N &= j(j+1) \ket{j,m}_N 
	\end{align*}
	with $m =-j, -j+1 ,\ldots ,j$ representing the number of excitations, while $j =1/2 ,3/2 ,\ldots ,N/2$ (for $N$ odd) and $j =0,1,\ldots , N/2$ (for $N$ even) corresponds
	to the cooperation number. 
	In general, atomic configurations for $N>2$ 
	contain entanglement \cite{Wang_2002}
	and are degenerated for $j < N/2$ \cite{DickeModelRevisited}. 
	Taking the maximal value $j=N/2$ results in the aforementioned fully symmetric Dicke states $\ket{D_N^k}$,
	\[
	\ket{D_N^k} := \ket{N/2 ,k +N/2}_N,
	\]
	In several physical situations, it is convenient to  
	consider the eigenvectors spanned by the fully
	symmetric Dicke states \cite{Wang_2002}. 
	Indeed, this significantly reduces the dimension 
	of space from $2^N$ to $N+1$. 
	
	Independently of the excitation number $k$, 
	the Dicke states $\ket{D_N^k}$ are eigenstates of the operator $J_{+} J_{-}$. 
	The latter operator 
	can be rewritten as a fully symmetric, two-body Hamiltonian, 
	\[
	H :=J_{+} J_{-} = \sum_{i=1}^N \sigma_{-}^{(i)} \sigma_{+}^{(i)}  +\sum_{i \neq j}^N \sigma_{-}^{(i)} \sigma_{+}^{(j)} ,
	\]
	in which excitations might be exchanged between any two parties. 
	%
	The Dicke states $\ket{D_N^k}$ are eigenvectors 
	corresponding to the non-degenerated maximal eigenvalue $k (N+1-k)$.
	In some physical problems the particles are not distributed symmetrically 
	as some interactions are more likely to occur. 
	Related two-body Hamiltonian of the system of $N$ particles with spin-$1/2$
	considered in \cite{Huber_2016} generalizes the  Heisenberg model. 
	
	
	We show that excitation-states, and hence Dicke-like states,
	are ground states for Hamiltonians with an interaction term  analogous
	to $J_{+} J_{-}$. 
	As in the Dicke case we discuss  the class of states with a well-defined excitation number 
	and analyze a single subspace of the $J_z$ operator. 
	Firstly, we restrict ourselves to the subspace of the operator $J_z$ relevant to two-excitations, i.e.  we set $k=2$. 
	For any graph $G$, consider the following four-body Hamiltonian, where excitations might be exchanged between any pair of edges:
	\begin{equation}
		H_G :=  J_{+}^G J_{-}^G ,
		\label{DH2}
	\end{equation}
	where 
	\[
	J_{\pm}^G \equiv \sum_{v v' \in E} \sigma_{\pm}^{(v)} \sigma_{\pm}^{(v')} \; .
	\]
	Notice that operators $J_{+}^G, J_{-}^G $ are
	constructed analogously to the operators $J_{+} J_{-}$. 
	The maximal eigenvalue of the Hamiltonian $H_G$ is equal to the squared number of edges $|E|^2$. 
	This energy level is non-degenerated 
	and the corresponding eigenstate is given by  the excitation-state $\ket{G}$. 
	Indeed,  the Hamiltonian $H_G$ is a sum of $|E|^2$ projections, which justifies the bound of the energy. 
	The state, $\ket{G}$ saturates this bound, contrary to any other state 
	with specified number of excitations, $k=2$. 
	
	The excitation-state $\ket{G}$ can
	be seen as a ground state of the three-body interaction Hamiltonian,
	\begin{equation}
	\label{3bodyH}
	H_G = \sum_{v \in V} \sum_{ \substack{ vv' \in E \\ v'v'' \in E}} \sigma_{-}^{(v)} \big( \sigma_{-}^{} \sigma_{+}^{} \big)^{(v')}  \sigma_{+}^{(v'')} .
	\end{equation}
	The above Hamiltonian describes the case, in which
	the pairs of excitations might be exchanged only between neighboring edges. 
	The maximal eigenvalue is equal to ${d \choose 2} |V|$, where $d$ is the degree of a vertex. 
	In analogy to the Hamiltonian (\ref{DH2}) the maximal eigenvalue is not degenerated
	and the state $\ket{G}$ forms the relevant eigenvector. 
	Notice that the Dicke state $\ket{D^2_N}$ is a ground state of the Hamiltonian associated with the full graph $G$ on $N$ vertices.
	
	It is worth noting that $\sigma_-\sigma_+ = (\mathbb{I} - \sigma_z)/2 = |1\rangle\langle1|$ and the Hamiltonian can be written in an alternative way:
\begin{equation}
	\label{3bodyHbis}
	H_G = \sum_{v \in V} \sum_{ \substack{ vv' \in E \\ v'v'' \in E}} \sigma_{-}^{(v)} |1\rangle\langle1|^{(v')}  \sigma_{+}^{(v'')} .
	\end{equation}
	In this form the term  $\sigma_-\sigma_+$ 
	 can be interpreted as conditional hopping interaction - if the site $v'$ is occupied, the hopping from $v$ to $v''$ is effected, otherwise no interaction happens. 
	The behaviour is very much akin to what is called the quantum transistor. 
	One of the simplest models of such a quantum transistor  \cite{loft2018quantum} involve left and right qubit and a two qubit gate, an interaction very similar to observed here.
	
	Similarly, excitation-states related to the $k$-uniform hypergraphs $G$ might be described as ground states of the Hamiltonian with $2k$-body interaction, $H= J_{+}^G J_{-}^G $. 
	The excitation-state $\ket{G}$ relevant to the $k$-hypergraph $G$ 
	is a  ground state of the Hamiltonian (\ref{DH2}),
	 where the operators $J_{\pm}^G$ are defined as
	\[
	J_{\pm}^G \equiv \sum_{ \{v^1 , \ldots ,v^k \} \in E} \sigma_{\pm}^{(v^1)} \cdots\sigma_{\pm}^{(v^k)} .
	\]
	
	In summary, the Dicke states $\ket{D_N^m}$ are uniquely determined as eigenvectors of operators $J_z$ and the two-body interaction Hamiltonian,
	 $H=J_{+} J_{-}$. 
	Excitation-states are also eigenvectors of an operator $J_z$, nevertheless the further division of its subspaces  
	differs from the one corresponding to the Dicke states. 
	Indeed, for any number $k$ of excitations and an excitation state $\ket{G}$ 
	one may construct a $(k+1)$-body interaction Hamiltonian with the state $\ket{G}$ being its non-degenerated ground state. 
	Note that the construction of the Hamiltonian presented above 
	depends on on the number of excitations. 
	
	
	\section{Classification problem}
	
	From a theoretical point of view symmetric states of $N$ qubits posses remarkably simple classification with respect to the stochastic local operations and classical communication (SLOCC) \cite{Bastin_2009} as they can 
	be classified by the \textit{Young diagrams} of a size $N$ \cite{Mathonet_2010}.  
	This is in a strong contrast to the classification of a general $N$-qubit state, 
	for which  the full classification is known for systems
	consisting of four qubits only \cite{FourQubits,Lamata_2007,Gharahi_Ghahi_2018}. 
	Classifying states with limited symmetries might be an accessible task 
	and a step toward the goal to extend the current limits
	of understanding of multipartite entanglement
	in the case of a large number of subsystems.

	\section{Concluding Remarks}
	
	In this work we advocate an original approach to construct genuinely entangled multipartite states by using a group action or, more generally, graphs and hypergraphs. 
	Obtained states resemble the Dicke states, whereas the excitations are chosen only between specific subsystems related to the structure of the graph. 
	We propose a quantum circuit generating the states constructed, 
	with the same complexity as similar circuits for the Dicke states. 
	We successfully simulated considered states on available quantum computers: IBM – Santiago, Vigo and Athens. 
	Moreover, discussed states are shown to form ground states of
	$3$-body Hamiltonians, analogous to the original Hamiltonian of Dicke. 
	Presented interaction, in turn, resembles quantum transistor behavior.
	

	In various realistic scenarios
	the full symmetry between collaborating systems is not required or not desirable. 
	Excitation-states appear as good candidates for such less-symmetric tasks. 
	For this reason, we investigate the possible symmetries of such states
	and demonstrate  that not all kinds of symmetry are possible.  
	This reflects the impossibility of arranging specifics tasks of a given symmetry. 
	Wherefore, special attention is dedicated to the family of excitation-states,
	   which exhibits a highly symmetric structure. 
	We refer to such states as Dicke-like, since they are a superposition of terms obtained by permutations of $\ket{1\ldots 1 0\ldots 0}$, 
	whereas only permutations from a specified subgroup $H$
	of the permutation group,
	$\mathcal{S}_N$ are taken. 
	We introduce specific families of quantum states related to highly symmetric objects, such as regular polygons, Platonic solids, and regular plane tilings. 
	
	We investigate the entanglement properties of introduced families of states. 
	Firstly, we provide the general separability conditions for excitation-states and present them in a twofold way: as specific conditions satisfied by a relevant graph or specific conditions imposed on the group of symmetry. 
	Secondly, we specifically quantify the entanglement presence in introduced families of states. 
	In order to compare the entanglement distribution, we compute the concurrence in bipartite systems $C_{vw}$ and the generalized concurrence $C_{v |\text{rest}}$ between particle $v$ and the rest of the system. 
	Analytic formula for these quantities are provided and illustrated on several examples. 
	For regular hypergraphs, with equal degrees of each node,
	the following general observation hold.
	
	\begin{enumerate}
		\item Concurrence $C_{v |\text{rest}}$
		describing entanglement between a given subsystem $v$ and the
		rest of the system  depends  only on the total number of nodes and uniformity of the graph
		and does not depend on its local structure.
		
		\item Concurrence $C_{vw}$ between nodes $v$ and $w$
		is positive only for distance-two nodes, and in states related to almost complete hypergraphs, i.e. $|E| \sim {{N}\choose{k}}$, also for distance-one nodes. 
		For all local networks, the second possibility does not occur. 
		\item Concurrence $C_{vw}$ between nodes of distance two 
		is proportional to the number of shared neighbours. 
	\end{enumerate}
	
	Moreover, we introduce the notion of entanglement ratio $\Gamma_v$, which measures the ratio of entanglement shared between particular subsystems in a bipartite way in comparison to the amount of entanglement shared in the multi-partite way. 
	We investigate the asymptotic behavior of the parameter $\Gamma_v$ of an excitation-state
	$\ket{G}$ with the change of the local degree $d$ of the relevant hypergraph. 
	Two phase transitions with respect to the bipartite entanglement in 
	the excitation-state network are reported. 
	For small average degree $d$ of vertices, the entanglement ratio
	scales as $\Gamma_v \sim d/N$, 
	and the entanglement is present only between distance-two nodes. 
	Once the limit $d\simeq0.794 N$ is exceeded, the bipartite entanglement between distance-one nodes starts to appear. 
	Its amount is, however, negligible in comparison to the entanglement between distance-two vertices. 
	Entanglement between distance-one nodes prevails once the threshold $d \approx 0.973 N $ is exceeded.

	
	
	\section{Acknowledgement}
	
	We are thankful to Wojciech Bruzda,  Micha\l{}  Horodecki, Felix Huber
	and Gon\c{c}alo Quinta
	for valuable remarks and fruitful discussions and to John Martin for helpful correspondence.
	Financial support by the Polish National Science Center under the
	grant numbers DEC-2015/18/A/ST2/00274 and 2019/35/O/ST2/01049 
	and by Foundation for Polish Science under 
	the Team-Net NTQC project is gratefully acknowledged.

\bibliographystyle{unsrt}	
	\bibliography{Physics.bib}

\begin{thebibliography}{10}

\bibitem{PhysRevLett.106.130506}
T.~Monz, P.~Schindler, J.~T. Barreiro, M.~Chwalla, D.~Nigg, W.~A. Coish,
  M.~Harlander, W.~H\"ansel, M.~Hennrich, and R.~Blatt.
\newblock 14-qubit entanglement: Creation and coherence.
\newblock {\em Phys. Rev. Lett.}, 106:130506, 2011.

\bibitem{Riofr_o_2017}
C.~A. Riofrío, D.~Gross, S.~T. Flammia, T.~Monz, D.~Nigg, R.~Blatt, and
  J.~Eisert.
\newblock Experimental quantum compressed sensing for a seven-qubit system.
\newblock {\em Nature Communications}, 8:15305, 2017.

\bibitem{PhysRevA.71.032349}
A.~Serafini, G.~Adesso, and F.~Illuminati.
\newblock Unitarily localizable entanglement of {G}aussian states.
\newblock {\em Phys. Rev. A}, 71:032349, 2005.

\bibitem{PhysRevLett.99.150501}
G.~Adesso and F.~Illuminati.
\newblock Strong monogamy of bipartite and genuine multipartite entanglement:
  The {G}aussian case.
\newblock {\em Phys. Rev. Lett.}, 99:150501, 2007.

\bibitem{PhysRevLett.98.060501}
A.~R.~Usha Devi, R.~Prabhu, and A.~K. Rajagopal.
\newblock Characterizing multiparticle entanglement in symmetric $n$-qubit
  states via negativity of covariance matrices.
\newblock {\em Phys. Rev. Lett.}, 98:060501, 2007.

\bibitem{Bergmann_2013}
M.~Bergmann and O.~Gühne.
\newblock Entanglement criteria for {D}icke states.
\newblock {\em Journal of Physics A: Mathematical and Theoretical}, 46:385304,
  2013.

\bibitem{Dicke54}
R.~H. Dicke.
\newblock Coherence in spontaneous radiation processes.
\newblock {\em Phys. Rev.}, 93:99--110, 1954.

\bibitem{PhysRevA.67.022112}
J.~K. Stockton, J.~M. Geremia, A.~C. Doherty, and H.~Mabuchi.
\newblock Characterizing the entanglement of symmetric many-particle
  spin-$\frac{1}{2}$ systems.
\newblock {\em Phys. Rev. A}, 67:022112, 2003.

\bibitem{Liu_2019}
Y.C. Liu, J.~Yu, X.D.and~Shang, H.~Zhu, and X.~Zhang.
\newblock Efficient verification of {D}icke states.
\newblock {\em Physical Review Applied}, 12:044020, 2019.

\bibitem{L_cke_2014}
B.~Lücke, J.~Peise, G.~Vitagliano, J.~Arlt, L.~Santos, G.~Tóth, and
  C.~Klempt.
\newblock Detecting multiparticle entanglement of {D}icke states.
\newblock {\em Phys. Rev. Lett.}, 112:155304, 2014.

\bibitem{Mazza_2015}
L.~Mazza, M.~Aidelsburger, H.~Tu, N.~Goldman, and M.~Burrello.
\newblock Methods for detecting charge fractionalization and winding numbers in
  an interacting fermionic ladder.
\newblock {\em New Journal of Physics}, 17:105001, 2015.

\bibitem{Wieczorek_2009}
W.~Wieczorek, R.~Krischek, N.~Kiesel, P.~Michelberger, G.~Tóth, and
  H.~Weinfurter.
\newblock Experimental entanglement of a six-photon symmetric {D}icke state.
\newblock {\em Phys. Rev. Lett.}, 103:020504, 2009.

\bibitem{PhysRevLett.86.910}
H.~J. Briegel and R.~Raussendorf.
\newblock Persistent entanglement in arrays of interacting particles.
\newblock {\em Phys. Rev. Lett.}, 86:910--913, 2001.

\bibitem{PhysRevLett.92.107901}
M.~Bourennane, M.~Eibl, S.~Gaertner, C.~Kurtsiefer, A.~Cabello, and
  H.~Weinfurter.
\newblock Decoherence-free quantum information processing with four-photon
  entangled states.
\newblock {\em Phys. Rev. Lett.}, 92:107901, 2004.

\bibitem{PhysRevLett.98.020503}
S.~Gaertner, C.~Kurtsiefer, M.~Bourennane, and H.~Weinfurter.
\newblock Experimental demonstration of four-party quantum secret sharing.
\newblock {\em Phys. Rev. Lett.}, 98:020503, 2007.

\bibitem{PhysRevA.59.156}
M.~Murao, D.~Jonathan, M.~B. Plenio, and V.~Vedral.
\newblock Quantum telecloning and multiparticle entanglement.
\newblock {\em Phys. Rev. A}, 59:156--161, 1999.

\bibitem{RevModPhys.90.035005}
L.~Pezz\`e, A.~Smerzi, M.~K. Oberthaler, R.~Schmied, and P.~Treutlein.
\newblock Quantum metrology with nonclassical states of atomic ensembles.
\newblock {\em Rev. Mod. Phys.}, 90:035005, 2018.

\bibitem{Tura2018separabilityof}
J.~Tura, A.~Aloy, R.~Quesada, M.~Lewenstein, and A.~Sanpera.
\newblock Separability of diagonal symmetric states: a quadratic conic
  optimization problem.
\newblock {\em {Quantum}}, 2:45, 2018.

\bibitem{Aloy_2021}
A.~Aloy, M.~Fadel, and J.~Tura.
\newblock The quantum marginal problem for symmetric states: applications to
  variational optimization, nonlocality and self-testing.
\newblock {\em New Journal of Physics}, 23:033026, 2021.

\bibitem{marconi2020entangled}
C.~Marconi, A.~Aloy, J.~Tura, and A.~Sanpera.
\newblock Entangled symmetric states and copositive matrices.
\newblock arXiv:2012.06631, 2020.

\bibitem{helwig2013absolutely}
W.~Helwig and W.~Cui.
\newblock Absolutely maximally entangled states: Existence and applications.
\newblock {\em arXiv:1306.2536}, 2013.

\bibitem{HelwigAME}
W.~Helwig, W.~Cui, A.~Riera, J.~Latorre, and Hoi-Kwong Lo.
\newblock Absolute maximal entanglement and quantum secret sharing.
\newblock {\em Phys. Rev. A}, 86:052335, 2012.

\bibitem{HIGUCHI2000213}
A.~Higuchi and A.~Sudbery.
\newblock How entangled can two couples get?
\newblock {\em Physics Letters A}, 273(4):213 -- 217, 2000.

\bibitem{Kimble_2008}
H.~J. Kimble.
\newblock The quantum internet.
\newblock {\em Nature}, 453:7198, 2008.

\bibitem{Ac_n_2007}
A.~Acín, J.~I. Cirac, and M.~Lewenstein.
\newblock Entanglement percolation in quantum networks.
\newblock {\em Nature Physics}, 3(4):256–259, 2007.

\bibitem{DistributionEntanglement}
S.~Perseguers, G.~Lapeyre, D.~Cavalcanti, M.~Lewenstein, and A.~Acín.
\newblock Distribution of entanglement in large-scale quantum networks.
\newblock {\em {R}eports on {P}rogress in {P}hysics. Physical Society (Great
  Britain)}, 76:096001, 2013.

\bibitem{Perseguers_2010}
S.~Perseguers, M.~Lewenstein, A.~Acín, and J.~I. Cirac.
\newblock Quantum random networks.
\newblock {\em Nature Physics}, 6(7):539–543, 2010.

\bibitem{PhysRevA.77.022308}
S.~Perseguers, J.~I. Cirac, A.~Ac\'{\i}n, M.~Lewenstein, and J.~Wehr.
\newblock Entanglement distribution in pure-state quantum networks.
\newblock {\em Phys. Rev. A}, 77:022308, 2008.

\bibitem{Gisin_2007}
N.~Gisin and R.~Thew.
\newblock Quantum communication.
\newblock {\em Nature Photonics}, 1(3):165–171, 2007.

\bibitem{ComplexNetworks2}
J.~Biamonte, M.~Faccin, and M.~Domenico.
\newblock Complex networks: from classical to quantum.
\newblock {\em Communications Physics}, 2:53, 2017.

\bibitem{PhysRevE.67.066203}
C.~Emary and T.~Brandes.
\newblock Chaos and the quantum phase transition in the {D}icke model.
\newblock {\em Phys. Rev. E}, 67:066203, 2003.

\bibitem{DickeModelRevisited}
B.~Garraway.
\newblock The {D}icke model in quantum optics: {D}icke model revisited.
\newblock {\em Philosophical transactions. Series A, Mathematical, physical,
  and engineering sciences}, 369(03):1137--55, 2011.

\bibitem{Szalay_2017}
S.~Szalay, G.~Barcza, T.~Szilvási, L.~Veis, and O.~Legeza.
\newblock The correlation theory of the chemical bond.
\newblock {\em Scientific Reports}, 7:2237, 2017.

\bibitem{Szalay_2015}
S.~Szalay, M.~Pfeffer, V.~Murg, G.~Barcza, F.~Verstraete, R.~Schneider, and
  O.~Legeza.
\newblock Tensor product methods and entanglement optimization forab
  initioquantum chemistry.
\newblock {\em International Journal of Quantum Chemistry},
  115(19):1342–1391, 2015.

\bibitem{ding2020concept}
L.~Ding, S.~Mardazad, S.~Das, S.~Szalay, U.~Schollwöck, Z.~Zimborás, and
  C.~Schilling.
\newblock Concept of orbital entanglement and correlation in quantum chemistry.
\newblock {\em J. Chem. Theory Comput.}, 17:79, 2021.

\bibitem{DickeCircuits}
M.~Plesch and V.~Bu\ifmmode~\check{z}\else \v{z}\fi{}ek.
\newblock Efficient compression of quantum information.
\newblock {\em Phys. Rev. A}, 81:032317, 2010.

\bibitem{brtschi2019deterministic}
A.~Bärtschi and S.~Eidenbenz.
\newblock Deterministic preparation of {D}icke states.
\newblock {\em Fundamentals of Computation Theory}, 11651:126--139, 2019.

\bibitem{RevModPhys.91.025001}
E.~Chitambar and G.~Gour.
\newblock Quantum resource theories.
\newblock {\em Rev. Mod. Phys.}, 91:025001, 2019.

\bibitem{Eltschka_2014}
C.~Eltschka and J.~Siewert.
\newblock Quantifying entanglement resources.
\newblock {\em Journal of Physics A: Mathematical and Theoretical}, 47:424005,
  2014.

\bibitem{PhysRevA.100.062329}
G.~M. Quinta, R.~Andr\'e, A.~Burchardt, and K.~\ifmmode~\dot{Z}\else
  \.{Z}\fi{}yczkowski.
\newblock Cut-resistant links and multipartite entanglement resistant to
  particle loss.
\newblock {\em Phys. Rev. A}, 100:062329, 2019.

\bibitem{horodecki2001mixedstate}
M.~Horodecki, P.~Horodecki, and R.~Horodecki.
\newblock Mixed-state entanglement and quantum communication.
\newblock {\em Quantum Information - Basic Concepts and Experiments}, 173:1007,
  2001.

\bibitem{Bennett_1996}
C.~H. Bennett, D.~P. DiVincenzo, J.~A. Smolin, and W.~K. Wootters.
\newblock Mixed-state entanglement and quantum error correction.
\newblock {\em Phys.Rev. A}, 54(5):3824–3851, 1996.

\bibitem{PMID:11323664}
J.~Pan, C.~Simon, C.~Brukner, and A.~Zeilinger.
\newblock Entanglement purification for quantum communication.
\newblock {\em Nature}, 410:6832, 2001.

\bibitem{coecke2004logic}
B.~Coecke.
\newblock The logic of entanglement.
\newblock arXiv:0402014, 2004.

\bibitem{Serrano_Ens_stiga_2020}
E.~Serrano-Ensástiga and D.~Braun.
\newblock Majorana representation for mixed states.
\newblock {\em Phys.Rev. A}, 101:022332, 2020.

\bibitem{Aulbach_2010}
M.~Aulbach, D.~Markham, and M.~Murao.
\newblock The maximally entangled symmetric state in terms of the geometric
  measure.
\newblock {\em New Journal of Physics}, 12:073025, 2010.

\bibitem{Markham_2011}
D.~J.~H. Markham.
\newblock Entanglement and symmetry in permutation-symmetric states.
\newblock {\em Phys.Rev. A}, 83:042332, 2011.

\bibitem{PhysRevA.85.032314}
W.~Ganczarek, M.~Ku\ifmmode~\acute{s}\else \'{s}\fi{}, and
  K.~\ifmmode~\dot{Z}\else \.{Z}\fi{}yczkowski.
\newblock Barycentric measure of quantum entanglement.
\newblock {\em Phys. Rev. A}, 85:032314, 2012.

\bibitem{CHGBSZ21}
C.~Chryssomalakos, L.~Hanotel, E.~Guzmán-González, D.~Braun,
  E.~Serrano-Ensástiga, and K.~\ifmmode~\dot{Z}\else \.{Z}\fi{}yczkowski.
\newblock Symmetric multiqudit states: Stars, entanglement, rotosensors.
\newblock arXiv:2103.02786, 2021.

\bibitem{ThreeQub}
W.~Dür, G.~Vidal, and J.~Cirac.
\newblock Three qubits can be entangled in two inequivalent ways.
\newblock {\em Phys. Rev. A}, 62:062314, 2000.

\bibitem{Poset}
T.~Brady.
\newblock A partial order on the symmetric group and new {K}({$\pi$}, 1)'s for
  the braid groups.
\newblock {\em Advances in Mathematics}, 161(07):20--40, 2001.

\bibitem{M_kel__2010}
H~Mäkelä and A~Messina.
\newblock N-qubit states as points on the bloch sphere.
\newblock {\em Physica Scripta}, T140:014054, 2010.

\bibitem{CrossRatioNqubits}
P.~Ribeiro and R.~Mosseri.
\newblock Entanglement in the symmetric sector of $n$ qubits.
\newblock {\em Phys. Rev. Lett.}, 106:180502, 2011.

\bibitem{Martin10}
J.~Martin, O.~Giraud, P.~A. Braun, D.~Braun, and T.~Bastin.
\newblock Multiqubit symmetric states with high geometric entanglement.
\newblock {\em Phys. Rev. A}, 81:062347, 2010.

\bibitem{Martin15}
D.~Baguette, F.~Damanet, O.~Giraud, and J.~Martin.
\newblock Anticoherence of spin states with point-group symmetries.
\newblock {\em Phys. Rev. A}, 92:052333, 2015.

\bibitem{Martin17}
S.~Designolle, O.~Giraud, and J.~Martin.
\newblock Genuinely entangled symmetric states with no $n$-partite
  correlations.
\newblock {\em Phys. Rev. A}, 96:032322, 2017.

\bibitem{Graph1}
M.~Hein, J.~Eisert, and H.~J. Briegel.
\newblock Multiparty entanglement in graph states.
\newblock {\em Phys. Rev. A}, 69:062311, 2004.

\bibitem{Graph2}
S.~Anders and H.~J. Briegel.
\newblock Fast simulation of stabilizer circuits using a graph-state
  representation.
\newblock {\em Phys. Rev. A}, 73:022334, 2006.

\bibitem{Hypergraphs_2013}
M.~Rossi, M.~Huber, D.~Bru{\ss}, and C.~Macchiavello.
\newblock Quantum hypergraph states.
\newblock {\em New Journal of Physics}, 15:113022, 2013.

\bibitem{LULCsupport}
M.~Van~den Nest, J.~Dehaene, and B.~De~Moor.
\newblock Local unitary versus local {C}lifford equivalence of stabilizer
  states.
\newblock {\em Phys. Rev. A}, 71:062323, 2005.

\bibitem{Raussendorf_2007}
R.~Raussendorf and J.~Harrington.
\newblock Fault-tolerant quantum computation with high threshold in two
  dimensions.
\newblock {\em Phys. Rev. Lett.}, 98:190504, 2007.

\bibitem{CNZ10}
B.~Collins, I.~Nechita, and K.~Życzkowski.
\newblock Random graph states, maximal flow and {F}uss–{C}atalan
  distributions.
\newblock {\em Journal of Physics A: Mathematical and Theoretical}, 43:275303,
  2010.

\bibitem{CNZ13}
B.~Collins, I.~Nechita, and K.~\ifmmode~\dot{Z}\else \.{Z}\fi{}yczkowski.
\newblock Area law for random graph states.
\newblock {\em Journal of Physics A: Mathematical and Theoretical}, 46:305302,
  2013.

\bibitem{ProblemTangle}
C.~Eltschka, A.~Osterloh, and J.~Siewert.
\newblock Possibility of generalized monogamy relations for multipartite
  entanglement beyond three qubits.
\newblock {\em Phys.Rev. A}, 80:032313, 2009.

\bibitem{PhysRevLett.78.5022}
S.~Hill and W.~K. Wootters.
\newblock Entanglement of a pair of quantum bits.
\newblock {\em Phys. Rev. Lett.}, 78:5022--5025, 1997.

\bibitem{EntMonotones}
G.~Vidal.
\newblock Entanglement monotones.
\newblock {\em Journal of Modern Optics}, 47(2-3):355--376, 2000.

\bibitem{DistributedEntanglement}
V.~Coffman, J.~Kundu, and W.~K. Wootters.
\newblock Distributed entanglement.
\newblock {\em Phys.Rev. A}, 61:052306, 2000.

\bibitem{Osborne_2006}
T.~J. Osborne and F.~Verstraete.
\newblock General monogamy inequality for bipartite qubit entanglement.
\newblock {\em Phys. Rev. Lett.}, 96:220503, 2006.

\bibitem{GeneralizedConcurrence}
V.~Bhaskara and P.~Panigrahi.
\newblock Generalized concurrence measure for faithful quantification of
  multiparticle pure state entanglement using {L}agrange's identity and wedge
  product.
\newblock {\em Quantum Information Processing}, 16:118, 2016.

\bibitem{BANG20151}
S.~Bang, A.~Dubickas, J.H. Koolen, and V.~Moulton.
\newblock There are only finitely many distance-regular graphs of fixed valency
  greater than two.
\newblock {\em Advances in Mathematics}, 269:1--55, 2015.

\bibitem{Goldberg_2020}
A.~Z. Goldberg, A.~B. Klimov, M.~Grassl, G.~Leuchs, and L.~L. Sánchez-Soto.
\newblock Extremal quantum states.
\newblock {\em AVS Quantum Science}, 2:044701, 2020.

\bibitem{Giraud_2010}
O.~Giraud, P.~Braun, and D.~Braun.
\newblock Quantifying quantumness and the quest for queens of quantum.
\newblock {\em New Journal of Physics}, 12:063005, 2010.

\bibitem{TranslationallyInvariantStates}
F.~Verstraete and J.~I. Cirac.
\newblock Matrix product states represent ground states faithfully.
\newblock {\em Phys. Rev. B}, 73:094423, 2006.

\bibitem{watson2020complexity}
J.~D. Watson, J.~Bausch, and S.~Gharibian.
\newblock The complexity of translationally invariant problems beyond ground
  state energies.
\newblock arXiv:2012.12717, 2020.

\bibitem{1629135}
V.~V. {Shende}, S.~S. {Bullock}, and I.~L. {Markov}.
\newblock Synthesis of quantum-logic circuits.
\newblock {\em IEEE Transactions on Computer-Aided Design of Integrated
  Circuits and Systems}, 25(6):1000--1010, 2006.

\bibitem{10.5555/2011517.2011522}
G.~Song and A.~Klappenecker.
\newblock Optimal realizations of controlled unitary gates.
\newblock {\em Quantum Info. Comput.}, 3(2):139–156, 2003.

\bibitem{Wgenerating}
E.~Jung, M.R. Hwang, Y.~Ju, M.S. Kim, S.K. Yoo, H.~Kim, D.~Park, J.W. Son,
  S.~Tamaryan, and S.K. Cha.
\newblock {G}reenberger-{H}orne-{Z}eilinger versus {W} states: {Q}uantum
  teleportation through noisy channels.
\newblock {\em Phys. Rev. A}, 78:012312, 2008.

\bibitem{Wang_2002}
X.~Wang and K.~Mølmer.
\newblock Pairwise entanglement in symmetric multi-qubit systems.
\newblock {\em The European Physical Journal D - Atomic, Molecular and Optical
  Physics}, 18(3):385–391, 2002.

\bibitem{Huber_2016}
F.~Huber and O.~Gühne.
\newblock Characterizing ground and thermal states of few-body {H}amiltonians.
\newblock {\em Phys. Rev. Lett.}, 117:010403, 2016.

\bibitem{loft2018quantum}
N.~J.~S. Loft, L.~B. Kristensen, C.~K. Andersen, and N.~T. Zinner.
\newblock Quantum spin transistors in superconducting circuits.
\newblock arXiv:802.04292, 2018.

\bibitem{Bastin_2009}
T.~Bastin, S.~Krins, P.~Mathonet, M.~Godefroid, L.~Lamata, and E.~Solano.
\newblock Operational families of entanglement classes for symmetric {N}-qubit
  states.
\newblock {\em Phys. Rev. Lett.}, 103:070503, 2009.

\bibitem{Mathonet_2010}
P.~Mathonet, S.~Krins, M.~Godefroid, L.~Lamata, E.~Solano, and T.~Bastin.
\newblock Entanglement equivalence of {N}-qubit symmetric states.
\newblock {\em Phys.Rev. A}, 81:052315, 2010.

\bibitem{FourQubits}
F.~Verstraete, J.~Dehaene, B.~De~Moor, and H.~Verschelde.
\newblock Four qubits can be entangled in nine different ways.
\newblock {\em Phys. Rev. A}, 65:052112, 2001.

\bibitem{Lamata_2007}
L.~Lamata, J.~León, D.~Salgado, and E.~Solano.
\newblock {I}nductive entanglement classification of four qubits under
  stochastic local operations and classical communication.
\newblock {\em Phys.Rev. A}, 75:022318, 2007.

\bibitem{Gharahi_Ghahi_2018}
M.~Gharahi~Ghahi and S.~Mancini.
\newblock Comment on “{I}nductive entanglement classification of four qubits
  under stochastic local operations and classical communication”.
\newblock {\em Phys. Rev. A}, 98:066301, 2018.

\end{thebibliography}
	
\end{document}